%% file: main.tex
\documentclass[sigconf,nonacm]{aamas}
\pdftrailerid{}
\usepackage{balance}
\usepackage{booktabs}
\usepackage{array}
\usepackage{graphicx}
\usepackage{amsmath}

\makeatletter
\@ifundefined{proposition}{\theoremstyle{plain}\newtheorem{proposition}{Proposition}}{}
\makeatother

\usepackage{pdfpages}

\setcopyright{none}
\AtBeginDocument{%
  \fancypagestyle{standardpagestyle}{\fancyhf{}%
    \fancyfoot[C]{\footnotesize\thepage}}%
  \fancypagestyle{firstpagestyle}{\fancyhf{}%
    \fancyfoot[C]{\footnotesize\thepage}}%
  \pagestyle{standardpagestyle}}

\input{tables/num_selfplay.tex}


\input{tables/num_scripted.tex}

\input{tables/num_mixed.tex}

\input{tables/num_selection.tex}



\newcommand{\popt}{p^{*}}

\title[Cooperation Without Adaptation: LLM Agents in a Collective-Risk Dilemma
 ]{Cooperation Without Adaptation: LLM Agents in a Collective-Risk Dilemma}

\author{Trung-Kiet Huynh}
\authornote{These authors contributed equally to this work.}
\affiliation{%
  \institution{University of Science, VNU-HCM}
  \city{Ho Chi Minh City}
  \country{Vietnam}}
\email{23122039@student.hcmus.edu.vn}

\author{Dao-Sy Duy-Minh}
\authornotemark[1]
\affiliation{%
  \institution{University of Science, VNU-HCM}
  \city{Ho Chi Minh City}
  \country{Vietnam}}
\email{23122041@student.hcmus.edu.vn}

\author{Chi-Nguyen Tran}
\authornotemark[1]
\affiliation{%
  \institution{University of Science, VNU-HCM}
  \city{Ho Chi Minh City}
  \country{Vietnam}}
\email{23122044@student.hcmus.edu.vn}

\author{Phu-Hoa Pham}
\affiliation{%
  \institution{University of Science, VNU-HCM}
  \city{Ho Chi Minh City}
  \country{Vietnam}}
\email{23122030@student.hcmus.edu.vn}

\author{Phu-Quy Nguyen-Lam}
\affiliation{%
  \institution{University of Science, VNU-HCM}
  \city{Ho Chi Minh City}
  \country{Vietnam}}
\email{23122048@student.hcmus.edu.vn}

\author{Le Hong Trang}
\authornote{Corresponding authors.}
\affiliation{%
  \institution{Ho Chi Minh City University}
  \institution{of Technology, VNU-HCM}
  \city{Ho Chi Minh City}
  \country{Vietnam}}
\email{lhtrang@hcmut.edu.vn}

\author{The Anh Han}
\authornotemark[2]
\affiliation{%
  \institution{Teesside University}
  \city{Middlesbrough}
  \country{United Kingdom}}
\email{t.han@tees.ac.uk}

\renewcommand{\shortauthors}{Huynh, Duy-Minh, Tran et al.}

\begin{abstract}
Language-model agents are increasingly deployed to act on behalf of people in settings where individual decisions shape shared outcomes. We examine whether five production models adapt their contributions to changing risks and partner behaviour, or instead follow a relatively fixed strategy. In a collective-risk dilemma game (CRG), six players contribute to a shared pool over ten rounds; if the pool falls short of its target, a catastrophe may destroy their remaining resources. Four models contributed approximately an equal share at every positive risk level, and four of the five also contributed when risk was zero, even though contributing could not improve their payoff. Explicitly stating that only the agent’s own payoff mattered reduced zero-risk contributions for most models, while a more neutral prompt did not. Several models correctly identified which action offered the higher expected payoff at different risk levels, yet their contributions changed little. Against fixed partners, no model consistently chose the payoff-maximising response, and in mixed-model groups, a single agent’s shortfall in the final round could cause an otherwise successful group to fail. Payoff-based selection also depended on population composition: in a single six-agent group, it favoured the model that contributed least and reduced group success, whereas in larger, randomly mixed populations, it favoured an exact fair-share model at moderate and high risk. These findings show that cooperative self-play alone is not evidence of strategic adaptability. Consistent with a machine-behaviour approach, evaluations of agents in multi-agent systems should test how they respond to incentives, partner actions, and population composition.
\end{abstract}


\begin{document}

\maketitle

\section{Introduction}
\label{sec:intro}
\begin{figure*}[t]
\centering
\includegraphics[width=0.9\textwidth]{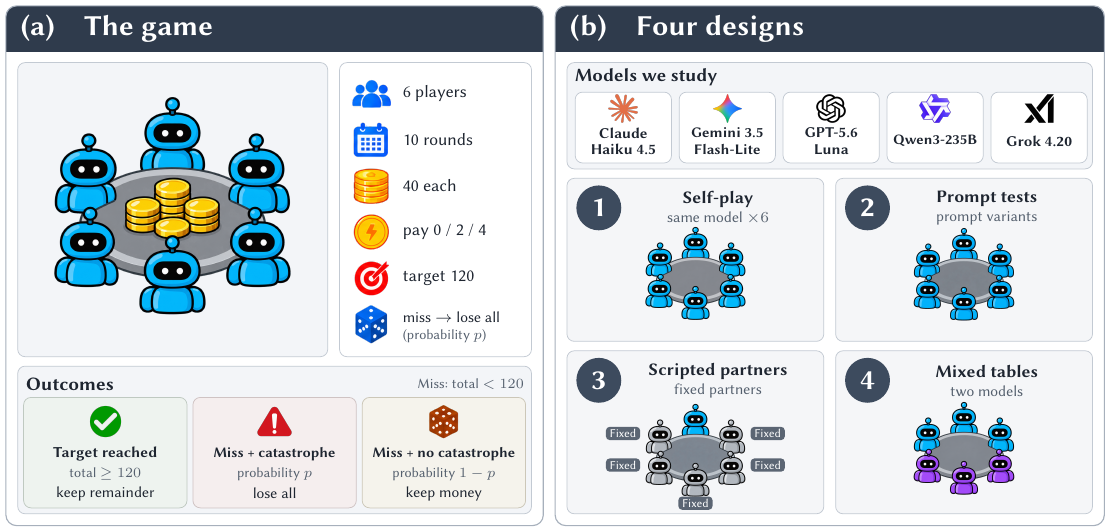}
\Description{A two-part overview. Left, the game: six robot players sit around a table with a pile of coins, beside a list of the rules: six players, ten rounds, 40 units each, contribute 0, 2 or 4, a target of 120, and a missed target loses everything with probability p. Below, three outcome boxes: target reached, a total of at least 120, keep the remainder; a miss with catastrophe, probability p, lose everything; a miss without catastrophe, probability 1 minus p, keep the money. Right, the logos and names of the five models, Claude Haiku 4.5, Gemini 3.5 Flash-Lite, GPT-5.6 Luna, Qwen3-235B and Grok 4.20, above four numbered design cards, each drawing a table of six robots: one, self-play, six copies of one model; two, prompt tests, the same table under prompt variants; three, scripted partners, one robot beside five grey robots labelled fixed; four, mixed tables, robots of two colours for two models.}
\caption{Overview. (a)~The game (also referred to as a table or group, Section~\ref{sec:rules}): six players hold 40 units each and have to decide to contribute either 0, 2 or 4 into a shared pool in each of ten rounds; if the pool misses the target of 120, a catastrophe destroys all their savings with probability $p$. (b)~The five models and the four designs (Table~\ref{tab:designs}): (1)~self-play over the risk grid; (2)~self-play under prompt and decoding controls, and with in-game value questions; (3)~one model beside five scripted partners that follow fixed rules; (4)~two models at one table, one of them in one to five of the six seats and the other in the rest. Selection (Section~\ref{sec:designs}) uses the base-prompt payoffs of designs 1, 2 and 4. %
}
\label{fig:overview}
\end{figure*}

Language-model agents increasingly act on behalf of people in settings where multiple agents jointly shape a shared outcome, such as budget allocation, shared compute use, or negotiations over common resources \cite{chan2024visibility,hammond2025multi}. These systems are therefore not only engineering artifacts: their behaviour in interactive settings is also an object of empirical study, as emphasized by the emerging field of machine behaviour \cite{rahwan2019machine}. In many such settings, avoiding collective loss requires agents to incur individual costs for the benefit of the group~\cite{santos2011risk,santos2019outcome,milinski2008collective}. This raises a central question for multi-agent systems (MAS) research: do language-model agents condition their cooperative behaviour on the strategic environment, or do they rely on relatively fixed behavioural defaults? From a system-design and machine-behaviour perspective, this question has three components. First, is an agent’s cooperation sensitive to the magnitude of collective risk? Second, does it respond to the behaviour of other agents in the group? Third, under selection pressures that retain or replicate higher-earning agents, which behavioural tendencies are favoured?

Success in self-play, a common criterion in MAS benchmarks~\cite{piatti2024cooperate,huang2025competing}, does not by itself answer these questions. When all positions in a group are occupied by copies of the same model, an agent that conditions its decisions on the stakes may generate the same observed pattern of play as an agent that follows a fixed rule. A machine-behaviour approach therefore calls for testing agents across systematically varied environments, rather than inferring their decision policies from a single setting. Success rates likewise obscure the extent to which a group falls short of the attainable optimum: a group may achieve the collective target while still dissipating a substantial share of its endowment. Moreover, a prompt that specifies both a group target and a fixed number of rounds may implicitly suggest an equal division of the burden, causing agents that simply reproduce this split to appear cooperative under conventional summary measures.

We study these questions using the collective-risk social dilemma game (CRG), a game introduced to study collective decision making for climate change mitigation~\cite{milinski2008collective}. Players contribute into a shared pool over several rounds, and if the pool misses a target, they all lose their savings with probability $p$ (Figure~\ref{fig:overview}a). For risk-neutral players, group welfare is maximised when no one contributes if $p < \tfrac{1}{2}$, and when the group meets the target exactly if $p > \tfrac{1}{2}$, for example through the \emph{fair share}, defined as an equal split of the target. The game also contains two situations in which contributing  cannot raise a player's own payoff: when there is no risk, and once the outcome is settled. Neither risk attitude nor concern for the group can explain  contribution in these situations, so they separate a default from a response to the stakes.

We introduce a game-theoretic evaluation that separates cooperative outcomes from incentive-sensitive behaviour by combining risk variation, dominated-action tests, exact best responses to fixed partners and selection in mixed populations, where each model is a strategy of an empirical game~\cite{wellman2006methods,tuyls2018generalised}. We apply it to five production models, Claude Haiku 4.5, Gemini 3.5 Flash-Lite, GPT-5.6 Luna, Qwen3-235B and Grok 4.20, in four designs (Figure~\ref{fig:overview}b), each compared with a benchmark computed from the game. We use standard equilibrium and finite-population tools, not a new solution concept. We call near-flat  contribution across risk levels under a fixed prompt a \emph{default}; the term describes behaviour, not internal reasoning.

We organise our findings around three dimensions. (1)~\emph{Stakes.} Most models contribute roughly the fair share at every positive risk level, even when there is no risk. Removing normative wording does not change this pattern. Stating that only the agent’s own payoff matters reduces most zero-risk contributions, but only GPT-5.6 Luna shifts from contributing little at low risk to contributing the fair share at high risk. Asked in separate calls, several models rank the two choices correctly and change that ranking with the risk, but they keep contributing  the same amount (Section~\ref{sec:risk}). (2)~\emph{Partners.} Beside scripted partners who make every  contribution a loss, no model plays the best response of contributing  nothing, and two keep contributing  after the outcome is already settled. In a mixed game (also referred to as a table or group), a single partner that stops contributing in the last round often causes a table of exact fair-share players to fail, while a table that overpays can absorb the shortfall. (Section~\ref{sec:partners}). (3)~\emph{Selection.} Population size changes who is selected under payoff-based imitation: a population of one table favours the smallest  contributor  and lowers group success, whereas a larger randomly mixed population favours the exact fair-share player at moderate and high risk (Section~\ref{sec:selection}).

\section{Related work}
\label{sec:related}

\paragraph{Collective risk and selection.}
The collective-risk dilemma was introduced to study why groups fail to avoid a shared disaster, and human groups reach the target mostly when the risk is high~\cite{milinski2008collective}; inequality and communication change how often they succeed~\cite{tavoni2011inequality}. Threshold public goods games have many equilibria~\cite{palfrey1984participation}, contributions decline over the rounds of repeated public goods games~\cite{andreoni1988free}, and the risk of collective failure can turn a defection trap into a coordination problem~\cite{santos2011risk,merhej2021cooperation}. In evolutionary simulations of the same game, constant fair sharers lose out to players who hold back early and compensate later~\cite{abouchakra2012evolutionary}. Selection on payoff relative to a small group favours whoever does better than its neighbours~\cite{schaffer1988evolutionarily}, while selection between groups can favour cooperation~\cite{traulsen2006evolution}. Finite-population analyses price the incentives that sustain cooperation in this game~\cite{duong2026cost} and show that maximising cooperation need not maximise welfare~\cite{han2024evolutionary}. This literature predicts that behaviour should vary with risk, so the absence of such variation is informative; we use  a similar game theoretical approach to analyse how different agent types perform when placed together in mixed groups/tables.

\paragraph{Language models as strategic players.}
Language models have been studied as economic subjects~\cite{horton2023large,chen2023emergence,mei2024turing} and in repeated and multi-player  games~\cite{fan2024can,akata2025playing,fontana2024nicer,pal2026large,kempinski2025game,pham2026humans}. Models tend to be more cooperative than people~\cite{mei2024turing}, and an even split can reflect a focal point rather than fairness~\cite{dodivers2025uncovering}. The weight of framing relative to game structure differs by model~\cite{lore2024strategic}, and models can hold a correct belief or plan and act against it~\cite{fan2024can,schmied2026llms}. FAIRGAME runs language-model agents through standard games and finds outcomes that depend on model, language and persona~\cite{buscemi2025fairgame}; later work extends it to a payoff-scaled prisoner's dilemma and a multi-player public goods game~\cite{huynh2025understanding}.

\paragraph{Language models in social dilemmas.}
Societies of agents over-use a common resource~\cite{piatti2024cooperate}, contributions to a public good can stay the same even when it returns nothing~\cite{huang2025competing}, and reasoning models free-ride~\cite{guzmanpiedrahita2025corrupted}. Agents stand in for people in a collective-risk game~\cite{slumbers2025using} and for developers, regulators and users in an evolutionary model of AI regulation~\cite{buscemi2025llms}. A study of 35 open models includes a collective-risk game at one risk level and observes their decisions far from the group optimum~\cite{kumar2026cooperative}. In a binary collective-risk game with strategies written by three models, payoff-biased imitation drives larger groups to selfish, low-welfare outcomes~\cite{willis2026evaluating}, and prompts that evolve under individual selection collapse into defection while group selection keeps them prosocial~\cite{celiktemel2026group}. A recent study audits rule recall, state tracking and payoff calculation before interpreting play in multi-player races to develop AI~\cite{pham2026humans}. It compares play with an evolutionary benchmark and human data, and finds that correct recall can coexist with weak expected-payoff calculation. We differ in three ways: we vary the risk from zero to certainty and add two tests in which contributing  cannot raise the  contributor's cash; we compute the exact best response to fixed partners; and we select among deployed models with payoffs from their real mixed tables, finding that at moderate and high risk the favoured model changes once the population exceeds one table.

\paragraph{Multiagent evaluation.}
Multiagent reinforcement learning studies cooperation in sequential social dilemmas~\cite{leibo2017multi,hughes2018inequity}, and ad hoc teamwork and cooperative AI ask for agents that work with partners unlike themselves~\cite{stone2010ad,dafoe2020open,leibo2021scalable}. Empirical game-theoretic analysis treats trained agents as strategies of a meta-game~\cite{wellman2006methods,tuyls2018generalised}, $\alpha$-Rank ranks them by evolutionary dynamics~\cite{omidshafiei2019alpha}, and evolutionary game theory is used to study emergent behaviour in multiagent systems~\cite{han2022emergent}. Language models show distinct strategic fingerprints in tournaments with classic strategies~\cite{payne2025strategic}, their cooperation evolves differently across base models~\cite{vallinder2025cultural}, and cooperation sustained by repetition weakens when co-players differ~\cite{tewolde2026coopeval}. In the collective-risk game, people who delegate to artificial agents succeed more often~\cite{fernandezdomingos2022delegation}, and in theory adaptive players compensate for the effort of artificial agents~\cite{terrucha2024art}. We add a multiplayer dilemma with exact best responses to fixed partners, and measure this compensation in real models.

\section{Game and evaluation design}
\label{sec:game}

\subsection{Collective-risk game and benchmarks}
\label{sec:rules}
Six players each start with $e=40$ units, and the game lasts $R=10$ rounds. In each round every player privately chooses a contribution $a_i^t\in\{0,2,4\}$ to a shared \emph{pool}, and after each round all players see what every seat has contributed. Let $x_i=\sum_t a_i^t$ be a player's total and $X=\sum_i x_i$ the pool. The target is $T=120$. A player keeps what it did not contribute, but if the pool misses the target, a catastrophe destroys all savings with probability $p$ (Table~\ref{tab:terms} lists the terms used):
\begin{equation}
u_i \;=\;
\begin{cases}
e-x_i & \text{if } X\ge T,\\
0 \text{ with probability } p, \text{ else } e-x_i & \text{if } X<T.
\end{cases}
\label{eq:payoff}
\end{equation}

\begin{table}[t]
\caption{Symbols and terms.}
\label{tab:terms}
\small
\setlength{\tabcolsep}{3pt}
\begin{tabular}{@{}>{\raggedright\arraybackslash}p{0.25\columnwidth}>{\raggedright\arraybackslash}p{0.72\columnwidth}@{}}
\toprule
Term & Meaning \\
\midrule
Seat; table & One player position; the six seats of one game \\
Fair share & The target split equally: 20 per seat, or 2 per round \\
Exact split & Contributing  2 in every one of the ten rounds \\
Pivot $\popt=\tfrac12$ & Risk above which reaching the target is optimal \\
$\mathrm{opt}(p)$ & Welfare-optimal expected payoff per seat \\
Dominated & Lowers the seat's own payoff whatever others do \\
Settled & Further contributions cannot affect the outcome: the pool has already reached $T$ or can no longer do so. \\
\bottomrule
\end{tabular}
\end{table}

\begin{proposition}
\label{prop:eq}
Consider the expected-payoff formulation of the game in which each player picks a total $x_i\in\{0,2,\dots,40\}$, and let $p<1$. A profile is a Nash equilibrium if and only if $x=0$, or $X=T$ and $x_i\le p\times e$ for every $i$. A target-reaching equilibrium exists if and only if $p\ge\popt=T/(6e)=\tfrac12$. The welfare-optimal outcome is $x=0$ for $p<\popt$ and $X=T$ for $p>\popt$, and both are optimal at $p=\popt$, so the welfare-optimal mean payoff per player is $\mathrm{opt}(p)=\max\{(1-p)e,\;e-T/6\}$; this welfare statement also holds at $p=1$.
\end{proposition}

\noindent\emph{Proof sketch.} If $X<T$, any player with $x_i>0$ gains by contributing  nothing, so the only equilibrium that misses the target is $x=0$. If $X>T$, then $X\ge T+2$ because totals are even, so a player can contribute 2 less and still reach the target. If $X=T$, the best deviation is to contribute nothing, worth $(1-p)e$ against $e-x_i$, which gives $x_i\le pe$; summing over players gives $p\ge T/(6e)$, and conversely $x_i=20$ for every $i$ is an equilibrium when $p\ge\tfrac12$. Welfare is $6e-X$ when the target is met and $(1-p)(6e-X)$ otherwise, so it peaks at $X=T$ or $X=0$. The ten-round game, in which strategies can react to the history, has the same totals in pure-strategy Nash equilibrium (see Supplementary Material, Proposition~S1, Section~S2.3, for the full proof). We claim nothing about subgame perfection, and $\mathrm{opt}(p)$ is a welfare bound that does not depend on the equilibrium concept.\hfill$\square$

Proposition~\ref{prop:eq} gives a \emph{symmetric welfare benchmark}: contribute nothing below the pivot and, above it, contribute the fair share needed to reach the threshold. Asymmetric allocations that reach exactly $T$ achieve the same group welfare. This benchmark is neither a unique equilibrium nor a prediction of individual play: contributing nothing is also an equilibrium at every risk level. Moreover, contributing a positive amount is strictly dominated in own expected payoff when $p=0$ and, for $p<1$, once the outcome is settled. These comparisons require neither risk neutrality nor beliefs about partners. %

\begin{proposition}
\label{prop:table}
For two seats $i$ and $j$ in the same game, the difference in their expected payoffs over the lottery is $c\,(x_j-x_i)$, with $c=1$ if the target was reached and $c=1-p$ otherwise.
\end{proposition}

\noindent Both seats face the same pool and the same lottery, so only their own contributions differ. Hence, for $p<1$, whoever contributes  less at a table earns more, whatever the group achieves.

\subsection{Models, protocol and measures}
\label{sec:setup}

We test the five models of Section~\ref{sec:intro}: one production model per provider, chosen to reflect the kind of model plausibly used for deploying many agents at scale. Qwen3-235B is the instruction-tuned build without a thinking mode, and Grok 4.20 is its non-reasoning variant.

\paragraph{Prompts.}
In each round a seat receives a prompt with the rules, the full history of every seat's contributions, and a request for making a contribution; the Supplementary Material provides every prompt word for word (Section~S4). The base prompt calls the pool a ``climate account'', says that the group ``must'' reach the target, calls the loss a ``disaster'' and names the game as a collective-risk social dilemma. It never states the player's goal, and it hints at the equal split twice. The prompt controls in Section~\ref{sec:designs} change exactly these parts.

\paragraph{Protocol.}
The running pool is not shown,  so any seat that needs it must compute it from the history. The prompt is in English and no seat is assigned a persona. We parse the move from the reply's final \texttt{CONTRIBUTION:} line.

\paragraph{Measures.}
We report \emph{expected} payoffs, the average a seat would earn if the catastrophe lottery were run many times on the same contribution, and the share of $\mathrm{opt}(p)$ that is lost. Proposition~\ref{prop:eq}, $\mathrm{opt}(p)$ and the selection fitness below all use expected payoffs, which remove the lottery noise without bias; agents never see the draw, and the Supplementary Material lists the realised draws (Table~S6). The unit of analysis is the game, since seats in one game share a history and a lottery. Intervals are 95\% percentile bootstraps over games and tests are permutation tests over games;  Supplementary Material gives every value behind the statements below (Sections~S6--S12).

\subsection{Evaluation designs}
\label{sec:designs}

Table~\ref{tab:designs} summarises the experimental designs, elaborated below; each model plays the same number of games in every design.

\begin{table}[t]
\caption{Designs. A cell is one condition at a risk level: one model in self-play, one model with a scripted partner, or one pair of models with a split of seats. Each cell holds ten games.}
\label{tab:designs}
\small
\setlength{\tabcolsep}{3pt}
\begin{tabular}{llr}
\toprule
Design & Risk levels $p$ & Games \\
\midrule
Self-play grid & $0, 0.1, \dots, 1$ & \CntGamesBaseline \\
Rerun, with questions & $0.1, 0.5, 0.9$ & \CntGamesProbe \\
No hint & $0.1, 0.5, 0.9$ & \CntGamesNoCue \\
Neutral wording & $0, 0.1, 0.9$ & \CntGamesWording \\
Neutral wording, goal stated & $0, 0.1, 0.9$ & \CntGamesNeutral \\
Two rewordings & $0.1, 0.9$ & \CntGamesPara \\
Temperature zero & $0.1, 0.9$ & \CntGamesTempZero \\
One model, five scripts (4 kinds) & $0.1, 0.3, \dots, 0.9$ & \CntGamesScripted \\
Two models (10 pairs, 1 to 5 seats) & $0.1, 0.5, 0.9$ & \CntGamesMixed \\
\midrule
All & & \CntGames \\
\bottomrule
\end{tabular}
\end{table}

\paragraph{Self-play and prompt controls.}
Six copies of one model play the base prompt at 11 risk levels  $p=0, 0.1, \dots,1$, and again under prompt and decoding controls at fewer risk levels. \emph{No hint} deletes the two equal-split phrases. \emph{Neutral wording} replaces the normative and naming phrases with plain ones, such as ``a shared account''. \emph{Goal stated} keeps the neutral wording and adds an opening  saying the player is in a paid experiment where ``the only thing that matters to you is your own final cash payoff''. Two \emph{rewordings} change the phrasing but keep every rule and the hint, one condition sets the temperature to zero, and a \emph{rerun} repeats the base prompt with value questions.

\paragraph{Value questions.}
In the rerun, separate calls whose answers never enter the game ask the first seat at the start, middle and end of the game which of two whole-game strategies gives the higher average payoff: contributing 2 in every round and reaching the target, or contributing nothing and missing it. The correct answer switches at $\popt$ (Supplementary Material, Section~S8).

\paragraph{Fixed partners.}
One model plays with five scripted partners of the same type. They always contribute 0, always contribute 2, always contribute 4, or copy the group, which means contributing  2 in the first round and then the allowed amount nearest to the average of what the other seats contributed in the previous round. The model is not informed that its partners are scripted, but their moves are included in the history. We compute the \emph{best response} to each type, i.e. the play that earns the most, by enumerating every possible sequence of the model’s own ten moves. With partners who always contribute 0, the target cannot be reached; with partners who always contribute 4, the target is reached without the model. The best response to both is to contribute nothing, and once the history settles the outcome, contributing cannot improve the outcome and only reduces payoff. Against fair-share partners, the best response is to contribute nothing below $\popt$ and exactly the fair share above it; partners who copy the group never deviate from 2, thus resulting in the same best response.

\paragraph{Mixed tables.}
To mix defaults at a table, we seat every pair of models together, with one to five seats for the first model and the rest for the second, at low, moderate and high risk. Each model's seats form one block, fixed within a cell and alternating across cells, so every model occupies every seat equally often  overall.

\paragraph{Selection among models.}
We finally examine which model prevails under payoff-based retention or imitation, and how this outcome depends on population size. We treat each LLM as a strategy in an empirical game, i.e. using experimentally estimated payoffs from model interactions to define the game~\cite{wellman2006methods,tuyls2018generalised}. We write $\pi_A(k)$ for the mean expected payoff of an $A$ seat at a table with $k$ seats of $A$ and $6-k$ of $B$, from mixed tables for $1\le k\le5$ and from self-play (grid plus rerun) for $k=6$. A population of $N$ agents consists of $i$ agents of type $A$ and $N-i$ of type $B$, and each agent's five tablemates are drawn at random from the other $N-1$, as in e.g. ~\cite{santos2011risk,han2017evolution,santos2020picky}. The fitness of each type is its payoff averaged over the tables/interactions it participates in:
\begin{align}
f_A(i) &= \textstyle\sum_{k=1}^{6} \mathcal{H}(k-1;\,N-1,\,i-1)\;\pi_A(k), \label{eq:fitness}\\
f_B(i) &= \textstyle\sum_{k=0}^{5} \mathcal{H}(k;\,N-1,\,i)\;\pi_B(6-k), \nonumber
\end{align}
where $\mathcal{H}(k;\,N-1,\,m)$ is the hypergeometric chance that $k$ of five tablemates come from the $m$ other agents of type $A$. Under Fermi imitation~\cite{traulsen2006stochastic,han2017evolution}, a randomly selected agent copies a random other agent's strategy with probability $1/(1+e^{-\beta\Delta f})$, where $\Delta f$ is the other's fitness minus its own and $\beta=1$ per payoff unit. We assume rare mutation~\cite{fudenberg2006imitation,han2017evolution}, which has been shown to provide a strong approximation beyond the strict rare-mutation regime~\cite{han2012emergence,rand2013evolution,zisis2015generosity,sigmund2010social}: each newcomer either fixes or disappears before the next arrives. Pairwise fixation probabilities then define a Markov chain over the five homogeneous model populations.  Its \emph{stationary mass} gives each model's long-run weight, and EGTtools~\cite{fernandezdomingos2023egttools} was used to reproduce the fixation probabilities and stationary distributions of Figure~\ref{fig:selection}. With $N=6$ the whole population is one table; with $N=30$, fitness averages over random, overlapping tables. Changing $N$ changes both table compositions and fixation dynamics, not only whom an agent is compared with. We evaluate this finite-population process, not a scoring-only intervention.

\begin{figure*}[t]
\centering
\includegraphics[width=\textwidth]{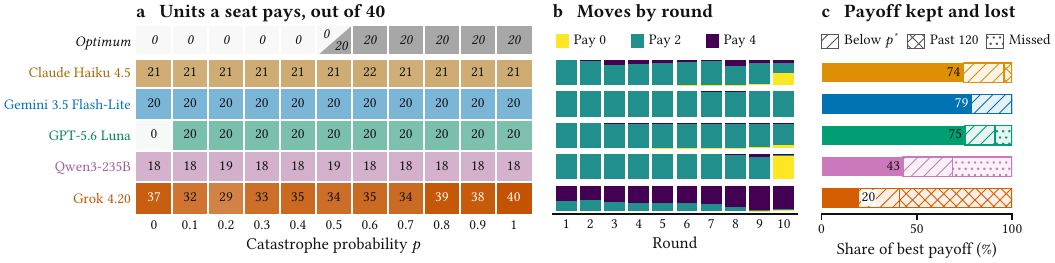}
\Description{Three panels that share one row per model, below a reference row. Left: a grid of the mean units a seat contributes  at each catastrophe probability from 0 to 1. The optimum row switches from 0 to 20 at one half, while the model rows barely change: about 20 for four models, 0 only for GPT-5.6 Luna at probability zero, and 29 to 40 for Grok 4.20. Middle: for each round, the share of seats contributing  0, 2 or 4; Qwen3-235B contributes  2 until round 10 and then almost always 0, and Grok 4.20 mostly contributes  4. Right: bars of the share of the best attainable payoff each model keeps; Grok 4.20 keeps about a fifth and loses most by contributing  past the target, and Qwen3-235B loses most by missing it.}
\caption{Self-play on the risk grid, 110 games per model. (a)~Mean units a seat contributes  out of 40, ten games per cell; the top row is the optimum of Proposition~\ref{prop:eq}, where both 0 and 20 are optimal at $p=0.5$. (b)~Share of seats contributing  0, 2 or 4 in each round, over the games with $p>0$. (c)~Expected payoff kept as a share of $\mathrm{opt}(p)$, and the loss split into contributing  toward the target below $\popt$, contributing  past 120, and missing the target at or above $\popt$.}
\label{fig:selfplay}
\end{figure*}

\section{Results}
\label{sec:results}
\subsection{Stakes:  contribution does not scale with risk}
\label{sec:risk}

\subsubsection{Self-play across risk}
\label{sec:grid}

Four of the five models contribute approximately  the fair share at every positive risk level (Figure~\ref{fig:selfplay}a). Claude Haiku 4.5, Gemini 3.5 Flash-Lite, GPT-5.6 Luna and Qwen3-235B contribute almost the same amount below the pivot as above it (a step of at most \GridStepMaxOther{} units), whereas the symmetric welfare benchmark jumps from contributing nothing to the fair share. In the original experiment with the same rules, human groups reached the target in half of their games at $p = 0.9$ and in none at $p = 0.1$ ~\cite{milinski2008collective}. Strong risk aversion could explain a flat contribution profile above zero risk: because catastrophe leaves players with nothing, a player with relative risk aversion near one may prefer a certain fair-share contribution even at low risk. Gemini 3.5 Flash-Lite provides the clearest example. Almost every seat plays the exact split at every risk level (Table~\ref{tab:models}), so this near-fixed rule attains the welfare optimum above the pivot but not below it.
\paragraph{Paying when there is nothing to protect.}
Four of the five models also contribute at $p=0$, about as much as at any other risk level (Table~\ref{tab:models}, first column). Contributing  there only lowers a player's own payoff and helps nobody, so neither risk attitude nor concern for others'  payoff can explain this contribution. GPT-5.6 Luna is the exception. It contributes  nothing at $p=0$ and its usual share at every positive risk, so it distinguishes zero risk from positive risk but does not respond to how high the risk is.

\paragraph{Three ways to lose.}
The models lose payoff in three different ways (Table~\ref{tab:models}, Figure~\ref{fig:selfplay}c). Gemini 3.5 Flash-Lite, Claude Haiku 4.5 and GPT-5.6 Luna lose mostly at low risk, where they contribute for a target that is not worth reaching. Grok 4.20 overpays: its groups reach the target well before the last round, and its seats keep contributing  into a pool that is already safe. Qwen3-235B fails at the end. It contributes  2 round after round, then almost all of its seats contribute nothing in the last round (Figure~\ref{fig:selfplay}b). Its groups therefore almost never reach the target, although every one of them would have succeeded had each seat simply repeated its previous move. Success rates hide all of this. Gemini 3.5 Flash-Lite and Grok 4.20 both reach the target in every game, yet Grok 4.20 keeps far less of its endowment . For $0.5<p<1$, Qwen3-235B and Grok 4.20 groups earn less in most games than in the worst equilibrium of Proposition~\ref{prop:eq}, in which nobody contributes. In the last round, Qwen3-235B skips as often when its contribution is needed as when the others' fair shares already suffice, so its skip follows the round number rather than the pool; human contributions similarly decline over the rounds of repeated public goods games~\cite{andreoni1988free}. GPT-5.6 Luna skips mostly when its contribution is not needed, and Claude Haiku 4.5 stops mostly once the target is reached (Supplementary Material, Section~S12.1).

\begin{table}[t]
\vspace{-0.1cm}
\caption{Self-play over the risk grid. \emph{Pays at $p{=}0$}: mean seat total where contributing  is dominated. \emph{Exact split}: seats contributing  2 in all rounds. \emph{Target}: games reaching it. \emph{Payoff}: expected payoff per seat; the optimum averages \GridOptMean{}. \emph{Lost}: share of $\popt$ lost. \emph{Main loss}: contributing  below $\popt$, contributing  past 120, or missing the target at or above $\popt$ in last round.}
\label{tab:models}
\small
\input{tables/tab_models.tex}
\end{table}

\subsubsection{Prompt and objective controls}
\label{sec:controls}

\paragraph{The default effect persists under rewording, but the resulting failures do not.}
Deleting the two equal-split phrases does not reduce any model’s contribution (Table~\ref{tab:robust}), indicating that the hint was not the source of contributions. Instead, it bounded contributions for some models while enabling others contribute exactly 2 per round. Gemini 3.5 Flash-Lite still plays the exact split, which it can rebuild from the target, while GPT-5.6 Luna keeps its average  contribution but plays the exact split less often. Qwen3-235B and Grok 4.20 contribute more without the hint; Qwen3-235B now reaches the target, but it earns less because it overpays. The two rewordings and a temperature of zero leave the contributions of the four models other than Grok 4.20 almost unchanged. The failures are less stable than the averages. Under the second rewording Grok 4.20 contributes  about the fair share, and how often Qwen3-235B drops out in the last round varies widely across prompts. Grok 4.20's rise with risk on the grid recurs at temperature zero but not in the rerun or the second rewording, so we do not interpret it as a stable response.

\paragraph{Stating the goal, not removing norms, reduces  zero-risk contributions.}
The normative words are not the main source of the default. With neutral wording, most models still contribute at $p=0$. At low and high risk the four models other than Grok 4.20 contribute on average what they contributed before, while Grok 4.20 contributes  even more and GPT-5.6 Luna plays the exact split less often (Table~\ref{tab:robust}). Adding the stated goal is what changes play at zero risk: Claude Haiku 4.5 and Gemini 3.5 Flash-Lite then contribute almost nothing, Qwen3-235B contributes  far less than before, and Grok 4.20 barely changes. The goal still does not make most models weigh the risk. Claude Haiku 4.5 and Gemini 3.5 Flash-Lite now separate zero from positive risk but already contribute the fair share at low risk, a pattern that strong risk aversion could also produce. Only GPT-5.6 Luna changes from contributing  little at low risk to the fair share at high risk. Qwen3-235B contributes  slightly more at high than at low risk but less than before at every risk level, and never reaches the target; Grok 4.20 still overpays. 

For most models, behaviour at zero risk changes depending on whether the prompt clearly states the goal. When risk above zero, models tend to contribute  similar amounts across risk levels, and this pattern remains under different wordings. The partner and selection results below use the base prompt.

\begin{table}[t]
\caption{Mean seat total under prompt and decoding changes. Top: pooled over $p=0.1$ and $0.9$ (20 games each). \emph{Rerun} repeats the base prompt with in-game questions; \emph{Para.\,1, 2} are the rewordings; \emph{Wording} is the neutral wording and \emph{Goal} adds the stated goal. Bold: a shift from \emph{Base} of at least 2 units with permutation $P<0.01$; every bold shift also has $P<0.001$, below the Bonferroni threshold for the 35 comparisons. Bottom: single risk levels (10 games each).}
\label{tab:robust}
\small
\setlength{\tabcolsep}{4pt}
\input{tables/tab_robust.tex}
\end{table}

\paragraph{Correct answers, unchanged play.}
Asked the value question in separate calls, Claude Haiku 4.5, Gemini 3.5 Flash-Lite and GPT-5.6 Luna mostly answer correctly and change their answer with the risk, yet the questioned seats contribute almost the same at low and high risk; Qwen3-235B and Grok 4.20 mostly answer that keeping pays more. This is a gap between answers and play, not evidence about what an agent computes; the Supplementary Material gives every answer and the answers against play (Section~S8.6). A strongly risk-averse agent could give these answers and still contribute; contribution at $p=0$ and after settlement provides the sharper test.

\subsection{Partners: simple responses,  unstable play}
\label{sec:partners}

\subsubsection{Fixed partners}
\label{sec:scripted}

\paragraph{No model keeps all its endowment  when that is the best response.}
The model contributes nothing only when paired with partners who always contribute either 0 or 4. Of the \ScrGames{} scripted games, only \ScrOptimalGames{} are played exactly optimally, and in every case the optimum is the fair share. In other words, the models reach the optimum only when it coincides with the default contribution. Since an agent unaware that its partners are scripted may reasonably contribute early to learn their behaviour, the cleaner evidence comes from contributions made after the outcome has settled (Table~\ref{tab:settled}).

\paragraph{Partners matter, risk does not.}
We split the variation in what a model contributes  across games with  three fixed-amount partners  into partner, risk and interaction parts. For Claude Haiku 4.5, Gemini 3.5 Flash-Lite, GPT-5.6 Luna and Qwen3-235B, partner type  explains most of the variation (\ScrOpponentShareMin\% to \ScrOpponentShareMax\%), while risk and its interaction with partner type explain at most  \ScrRiskShareMax\%. For an exact best responder in the same games,  risk would explain half of  the variation. Beside fair-share partners, where the best response jumps from nothing to the fair share at $\popt$, no model’s total contribution changes by more than 2 between the low- and high-risk conditions. Grok 4.20's total follows neither pattern: it is driven largely by its opening move, which is its largest contribution in around half of its games.

\begin{table}[t]
\caption{Units contributed per game before and after the outcome is settled, alongside five partners who always contribute 0 or always contribute 4 (50 games per entry). Given these partners, every unit contributed is wasted; units contributed after settlement are wasted even for an agent unaware that its partners are scripted.}
\label{tab:settled}
\small
\input{tables/tab_settled.tex}
\end{table}

\paragraph{Responses to partners take simple forms.}
The models adjust their contributions in response to their partners, but not in the way a best responder would (Table~\ref{tab:settled}; Figures~S4 and~S5 of the Supplementary Material show every round). Gemini 3.5 Flash-Lite contributes  its usual amount, while both it and Claude Haiku 4.5 stop contributing once the history indicates that the outcome is settled. With partners who always contribute 4, GPT-5.6 Luna contributes little even when the outcome is settled, remaining close to a best response as it learns its partners’ behaviour. Qwen3-235B contributes more, not less, when interacting with partners who contribute nothing. Both Qwen3-235B and Grok 4.20 also continue making large contributions after the outcome is settled -- a behaviour that no belief about their partners can justify. Against five fair-share partners, Qwen3-235B contributes  2 for nine rounds and then nothing, so the table misses the target in almost every game by that single payment; Gemini 3.5 Flash-Lite in the same seat reaches it every time.

\subsubsection{Mixed tables}
\label{sec:groups}

\begin{figure}[t]
\centering
\includegraphics[width=\columnwidth]{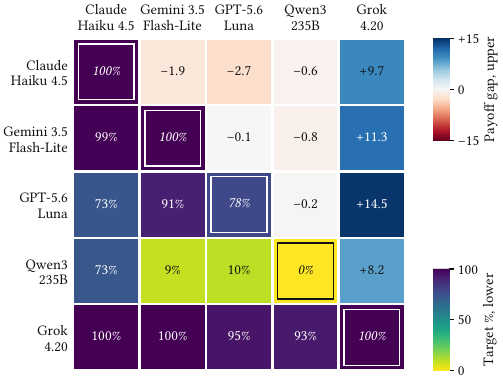}
\Description{One five-by-five matrix over pairs of models, split along the diagonal. Cells above the diagonal give the payoff of the row model's seats minus that of the column model's seats at the same table; the Grok 4.20 column is blue, meaning every other model earns more than Grok 4.20, and the other cells are pale, meaning small differences. Cells on and below the diagonal give the share of games that reach the target; the cells that pair Qwen3-235B with Gemini 3.5 Flash-Lite or GPT-5.6 Luna are bright yellow, the colour for a target rate near zero, and the Grok 4.20 row is dark.}
\caption{Two models play at each table, with 150 games per pair, pooled across seat assignments and risk levels. Above diagonal: expected payoff of  row model's seats minus  column model's, in the same games. On and below the diagonal: proportions of games reaching the target. Boxes mark self-play (60 games: grid games plus a rerun at  same risks).}
\label{fig:advantage}
\end{figure}

At a shared table every model earns more than Grok 4.20, while the payoff differences among the other four models are small (Figure~\ref{fig:advantage}, above the diagonal). This is consistent with   Proposition~\ref{prop:table}: Grok 4.20 contributes  the most, and at a two-model table the higher contributor earns less. Group success depends on the pair (Figure~\ref{fig:advantage}, below the diagonal). Tables that include Grok 4.20 reach the target in the great majority of games, while tables that mix Qwen3-235B with Gemini 3.5 Flash-Lite or GPT-5.6 Luna rarely do.

\begin{figure}[t]
\centering
\includegraphics[width=\columnwidth]{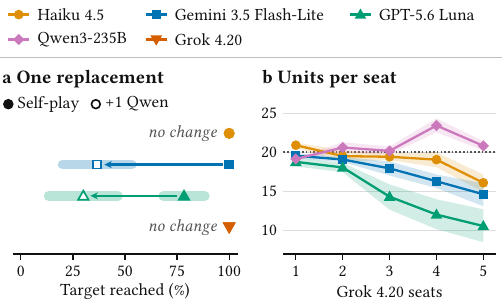}
\Description{A shared legend above both panels gives each model's colour and marker. Panel a draws, for each model, an arrow from self-play (filled marker) to one Qwen3-235B replacement (open marker) on a target-success axis. Gemini 3.5 Flash-Lite and GPT-5.6 Luna succeed less often with Qwen3-235B; Claude Haiku 4.5 and Grok 4.20 stay at full observed success and are marked no change. Panel b shows contributions for all five Grok 4.20 seat counts as one line per model with a shaded bootstrap band; Claude Haiku 4.5, Gemini 3.5 Flash-Lite and GPT-5.6 Luna contribute less as Grok 4.20 seats increase, while Qwen3-235B stays near or above the fair share. Translucent bars and bands are bootstrap intervals; the dotted line marks the fair share.}
\caption{Mixed groups, pooled over $p=0.1,0.5,0.9$. (a)~Self-play (filled; 60 games, grid plus rerun) versus one Qwen3-235B replacement (open; 30); arrows run from the first to the second. (b)~Mean contribution per focal seat beside one to five Grok 4.20 seats (30 games each); dotted: fair share. Bars and bands: 95\% game-bootstrap intervals, not stratified by risk. Lines join measured means.}
\label{fig:mixed}
\end{figure}

\paragraph{An exact fair-share target leaves no margin for error.}
One partner who stops contributing in the final round can cause a table of  fair-share players to fail (Figure~\ref{fig:mixed}a). When the pool is one round of fair shares short before the last round, Qwen3-235B contributes  nothing in \MixQwenPotZeroPct\% of cases and Gemini 3.5 Flash-Lite in \MixFlashPotZeroPct\%. Five Gemini 3.5 Flash-Lite seats with one Qwen3-235B seat succeed in only \MixFlashQwenReach{} of 30 games, although Gemini 3.5 Flash-Lite alone never misses in our sample; tables with GPT-5.6 Luna also succeed less often. Tables of Claude Haiku 4.5 or Grok 4.20, which contribute more, absorb the same partner. Replacing \emph{all} Qwen3-235B last-round zeros by 2, holding other recorded actions fixed, would rescue \MixQwenPivotalFix\% of the \MixMissWithQwen{} failed mixed tables containing it. This is an arithmetic counterfactual, not a tested intervention.  Conversely, one Grok 4.20 seat helps five Qwen3-235B seats succeed in most games, contributing  for the rescue. These comparisons depend on the base prompt and the fixed seat blocks.

\paragraph{Agents free-ride on generous partners.}
Most agents contribute less as the number of Grok 4.20 seats increases from one to five (Figure~\ref{fig:mixed}b). Alongside five Grok 4.20 seats,  GPT-5.6 Luna contributes  \MixLunaBesideFiveGrok{} units, and Gemini 3.5 Flash-Lite and Claude Haiku 4.5 also fall below the fair share; these two stop once the pool reaches the target, while GPT-5.6 Luna contributes  less from the start. Qwen3-235B does not show this decline. Within a cell, Grok 4.20's higher opening contributions are associated with lower total contributions by Gemini 3.5 Flash-Lite and Claude Haiku 4.5, mainly through earlier target completion. GPT-5.6 Luna also contributes  less while the pool is short, consistent with the compensation predicted alongside artificial agents~\cite{terrucha2024art}. Openings are not randomised and partners can react to a deviation, so this association identifies neither a causal response nor an exact best response.

\subsection{Selection: population-size dependent }
\label{sec:selection}

In a one-table population, selection favours the model that contributes least. Proposition~\ref{prop:table} explains why: within a game, contributing less leads to a higher individual payoff relative to the other players. In our data, Qwen3-235B contributes least relative to its tablemates, so selection gives it the greatest weight (Figure~\ref{fig:selection}a). But this individual advantage comes at a cost to the group: tables reach the target less often when models are selected this way than when a model is chosen uniformly at random. At $p=0.9$, for example, the target is reached in \SelWithinReachHigh\% of selected games, compared with \SelUniformReachHigh\% of randomly selected games. This illustrates the known result that selection based on payoffs relative to a small group favours whoever does better than its neighbours, even when this comes at a cost to the group~\cite{schaffer1988evolutionarily}. Here, we observe this dynamic in games played by deployed language models.

\begin{figure}[t]
\centering
\includegraphics[width=\columnwidth]{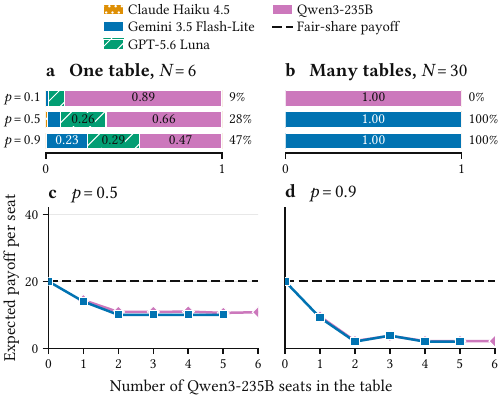}
\Description{Top: two sets of three stacked bars, one bar per catastrophe probability, split by the stationary mass of each model. In a population of one table, Qwen3-235B holds the largest part of every bar. In the larger population, Qwen3-235B fills the bar at low risk and Gemini 3.5 Flash-Lite fills the bars at moderate and high risk. Bottom: expected payoff of a Qwen3-235B seat and a Gemini 3.5 Flash-Lite seat as the number of Qwen3-235B seats changes in the same table, at moderate and high risk. The two lines show that the lower  contributor  can earn more than its tablemates even when the table misses the target; the dashed line is the fair-share payoff.}
\caption{Rare-mutation Fermi selection ($\beta=1$). (a,\,b)~Stationary mass in a population of one table ($N=6$) or a larger randomly mixed population ($N=30$); right of each bar, the self-play target rate weighted by mass. (c,\,d)~Expected payoff of a Qwen3-235B or Gemini 3.5 Flash-Lite seat by the number of Qwen3-235B seats at its table, at $p=0.5$ and $p=0.9$; dashed: fair-share payoff. Invasion diagrams are in the Supplementary Material (Figure~S7).}
\label{fig:selection}
\end{figure}

In  larger populations, selection favours Gemini 3.5 Flash-Lite at moderate and high risk (Figure~\ref{fig:selection}b). Its self-play attains the welfare optimum there, but neither optimality nor Nash equilibrium implies that imitation will select it. The ranking comes from the estimated payoffs at every composition along the fixation path, where a lower  contributor  can earn more than its tablemates even while lowering its own payoff (Figure~\ref{fig:selection}c,\,d). At low risk both population sizes favour Qwen3-235B, whose outcome remains far below the optimum. The switch to Gemini 3.5 Flash-Lite occurs already at seven or eight agents, consistent with finite-population theory~\cite{schaffer1988evolutionarily}. 

Each winner is stable under bootstrap resampling, stronger selection, larger $N$ and grid-only payoffs, but not under weak selection or when a candidate is removed (Supplementary Material, Sections~S11.4 and~S11.6).

\section{Conclusion and discussion}
\label{sec:discussion}

Our results have shown that cooperative-looking self-play is not sufficient evidence that an agent responds appropriately to incentives or adapts to its partners. Across the five models, contributions often remain close to the fair share as risk changes, including when contributing is payoff-dominated; some models also contribute after their actions can no longer affect the outcome. In mixed tables, one late shortfall can make an exact fair-share target fail. Selection outcomes depend on population structure: in a one-table population, selection can favour the lowest contributor at the group’s expense, while larger randomly mixed populations favour the fair-share model at moderate and high risk.

These findings matter for how we understand machine behaviour. A cooperative action need not reflect a stable preference for cooperation or an understanding of its consequences: it may instead be a default response to the prompt or a learned focal-point strategy. This empirical approach to auditing agent behaviour—treating models as subjects for observation in varied, interactive environments—aligns with the call to systematically study machine behaviour~\cite{rahwan2019machine}. Conversely, correctly identifying the payoff-maximising action when asked does not ensure that a model will choose it in play. We therefore distinguish what models can state from how they behave, and caution against treating self-play success as evidence of strategic competence, social preference, or reliable adaptation. Evaluations should vary incentives and partner behaviour, test mixed populations, and compare individual payoffs with group welfare. They should also report the welfare benchmark, identify dominated actions, and measure contributions after the outcome is settled. This can reveal wasted resources and cases where apparent cooperation does not improve welfare~\cite{han2024evolutionary}. 

The mechanisms behind these behaviours remain uncertain: our data show choices, not reasoning. One possibility is that an equal split acts as a focal point that models reconstruct even without an explicit hint~\cite{dodivers2025uncovering}. Another is that questions prompt calculation while requests to act elicit a behavioural default~\cite{fan2024can,schmied2026llms}. Post-training may also encourage models to contribute their share unless instructed otherwise, consistent with findings that models are more cooperative than people~\cite{mei2024turing} and that cooperation increases with the stakes in repeated prisoner’s dilemmas, contrary to evolutionary predictions~\cite{huynh2026payoff}. Stronger reasoning is not necessarily a remedy: reasoning models can free-ride more in public-goods games~\cite{guzmanpiedrahita2025corrupted}. These accounts are hypotheses, not explanations established by our experiments.\\

\noindent\textbf{Limitations, scope, and reproducibility.}
The goal sentence is tested in only one phrasing and in self-play; without a group-payoff goal, we cannot distinguish the effect of making the objective explicit from the effect of the instruction’s direction. Because the running pool is not displayed, arithmetic errors may explain some contributions made after the group has already met its target. The value questions also omit $p=0$.

We test five models, one from each provider, in English and in a single game. Our claims therefore concern this panel and setting, not language models in general. Ten games per condition limit statistical power, and a zero-width bootstrap interval for a constant sample should not be interpreted as evidence of deterministic behaviour. Our prompt rewordings cover only a small portion of the prompt space; welfare losses assume risk neutrality. The scripted, mixed, and selection designs use the base prompt. Fixed seat blocks confound player composition with position, and the selection analysis assumes rare entry.

\balance
\bibliographystyle{ACM-Reference-Format}
\bibliography{refs}

\clearpage
\IfFileExists{supplement.pdf}{

\section*{Supplementary material}

# Supplementary Material

## Cooperation Without Adaptation: LLM Agents in a Collective-Risk Dilemma

Trung-Kiet Huynh*, Dao-Sy Duy-Minh*, Chi-Nguyen Tran*, Phu-Hoa Pham, Phu-Quy Nguyen-Lam, Le Hong Trang and The Anh Han

*These authors contributed equally to this work.

This document supports the paper. The paper states its results in words and keeps its numbers in tables and figures; this supplement gives the values behind every statement, the full proofs, every prompt and question word for word, and the methods. Section numbers here start with S; "Section 4" without an S means a section of the paper.

# Contents



# S1 How to read this supplement

Table S1 maps each part of the paper to the sections here. Every results section (Sections S6 to S11) opens with a short table called "Numbers behind the main text", which lists each statement the paper makes in words, stated precisely, beside its value. The tables and figures that follow give the full breakdown by model, risk level and condition.

**Table S1:** Where to find what.

| Part of the paper | In this supplement |
|---|---|
| Section 3.1, rules and Propositions 1 and 2 | Section S2: full proofs, the ten-round game, the loss decomposition, risk aversion, best responses to scripted partners |
| Section 3.2, models, prompts, protocol | Sections S3 and S4: model identifiers, how replies are read, every prompt word for word |
| Sections 3.2 and 3.3, designs and measures | Section S5: games per design, the lottery draws, statistical methods, reproducibility |
| Section 4.1.1, self-play across risk | Section S6 |
| Section 4.1.2, prompt and decoding changes | Section S7 |
| Section 3.3, asking the agent | Section S8: every question word for word, all answers, and answers against play |
| Section 4.2.1, fixed partners | Section S9 |
| Section 4.2.2, mixed tables | Section S10 |
| Sections 3.3 and 4.3, selection | Section S11: method, payoff functions, all selection settings |
| Sections 4.1.1 and 4.2.2, last round, equilibria, rules, responses to openings | Section S12: including how each model plays the last round |

**Conventions.** The unit of analysis is the game, because the seats of one game share a history and a lottery. Payoffs are expected payoffs over the catastrophe lottery (Section S5.2). A seat's *total* is what it paid over the ten rounds, out of 40. Intervals are 95% percentile bootstrap intervals over games, written [lower, upper]. Models always appear in the same order: Claude Haiku 4.5, Gemini 3.5 Flash-Lite, GPT-5.6 Luna, Qwen3-235B, Grok 4.20.

# S2 The game and its benchmarks

## S2.1 Rules and notation

Six players each start with an endowment $e = 40$. In each of $R = 10$ rounds, every player chooses privately and at the same time a contribution $a_i^t \in \{0, 2, 4\}$ to a shared pool, and after each round every seat's contribution is shown to all. A player's total is $x_i = \sum_{t=1}^{R} a_i^t \in \{0, 2, \dots, 40\}$ and the pool is $X = \sum_{i=1}^{6} x_i$. If $X \geq T = 120$, every player keeps $e - x_i$. If $X < T$, a single lottery for the whole table destroys all savings with probability $p$, and otherwise every player keeps $e - x_i$. For a risk-neutral player the expected payoff is therefore

$$u_i(x) = \begin{cases} e - x_i & \text{if } X \geq T, \\ (1-p)\,(e - x_i) & \text{if } X < T. \end{cases} \tag{1}$$

The fair share is $T/6 = 20$ per player, or 2 per round, and the pivot is $p^* = T/(6e) = \frac{1}{2}$. The equilibrium results below assume $p < 1$. At $p = 1$ a missed target leaves every player with nothing whatever it paid, so many profiles that miss the target become weak equilibria and the characterisation changes.

## S2.2 Proposition 1 and its proof

**Proposition 1.** *Consider the game in which each player picks a total $x_i \in \{0, 2, \dots, 40\}$, with payoffs* (1) *and $p < 1$. A profile is a Nash equilibrium if and only if $x = 0$, or $X = T$ and $x_i \leq pe$ for every $i$. A target-reaching equilibrium exists if and only if $p \geq p^*$. The welfare-optimal outcome is $x = 0$ for $p < p^*$ and $X = T$ for $p > p^*$, and both are optimal at $p = p^*$, so the best attainable payoff per player is* $\text{opt}(p) = \max\{(1-p)e,\ e - T/6\}$.

*Proof.* We check the three cases $X < T$, $X > T$ and $X = T$ in turn.

*Case $X < T$.* Suppose some player has $x_i > 0$. If it pays nothing instead, the pool is still below $T$ and its payoff rises from $(1-p)(e - x_i)$ to $(1-p)e$, which is strictly more because $p < 1$. So a profile that misses the target can be an equilibrium only

if $x = 0$. Conversely, $x = 0$ is an equilibrium: a single player can raise the pool to at most $40 < T$, so any deviation $x_i' > 0$ yields $(1-p)(e - x_i') < (1-p)e$.

*Case* $X > T$. Totals are even and $T$ is even, so $X \geq T + 2$. Some player has $x_i > 0$ and can pay 2 less. The pool stays at or above $T$, and its payoff rises by 2. No such profile is an equilibrium.

*Case* $X = T$. Paying more lowers a player's payoff, since the target is met anyway. Paying less makes the pool miss the target, and the best such deviation is to pay nothing, which yields $(1-p)e$. So player $i$ has no profitable deviation if and only if $e - x_i \geq (1-p)e$, that is $x_i \leq pe$.

*Existence.* If $p \geq p^*$, then $pe \geq 20$ and the profile in which every player pays 20 is an equilibrium. If $p < p^*$, then $pe < 20$, every equilibrium total is below 20, and six of them cannot add up to $T = 120$.

*Welfare.* Total welfare is $6e - X$ when $X \geq T$ and $(1-p)(6e - X)$ when $X < T$. Over profiles that reach the target it is largest at $X = T$, with value $6e - T$; over profiles that miss it, it is largest at $X = 0$, with value $(1-p)\,6e$. The first is larger if and only if $p > T/(6e) = p^*$, and they are equal at $p^*$. Dividing by six gives $\mathrm{opt}(p)$. □

The welfare statement also holds at $p = 1$, as the paper states: a missed target then leaves welfare 0, below the $6e - T$ of a pool that reaches $T$ exactly, so $\mathrm{opt}(1) = e - T/6$, which is what the formula gives.

The equilibrium statement is also checked against every unilateral deviation, from every symmetric profile and from 3,000 random asymmetric profiles at each of the ten risk levels from 0 to 0.9, as part of the tests that will be released with the code; the case $p = 1$ is checked separately.

## S2.3 The ten-round game

In the ten-round game a pure strategy maps every history of past contributions to a contribution in $\{0, 2, 4\}$, and payoffs still depend only on the totals through (1).

**Proposition S1.** *Let $p < 1$. The totals $x$ reached on the path of a pure-strategy Nash equilibrium of the ten-round game are exactly the equilibrium profiles of Proposition 1. Each of them is reached by an equilibrium in which every player follows a fixed schedule that ignores the history.*

*Proof. Every profile of Proposition 1 is reached.* Let each player follow any fixed schedule of contributions in $\{0, 2, 4\}$ that adds up to $x_i$; one exists because $x_i \leq 40$. The other players' totals then do not depend on what player $i$ does, so whatever strategy $i$ uses, it ends with some total $x_i'$ against fixed totals of the others. Its payoff is the payoff of the one-shot totals game, which by Proposition 1 is at most the payoff of $x_i$. No deviation is profitable.

*Only such profiles are reached.* Fix a pure equilibrium and its path totals $x$.

(i) The strategy "always pay 0" guarantees at least $(1-p)e$ against anything, since it earns $e$ if the target is met and $(1-p)e$ otherwise. In equilibrium every player therefore earns at least $(1-p)e$.

(ii) If $X < T$, a player earns $(1-p)(e - x_i)$, which is at least $(1-p)e$ only if $x_i = 0$. So $x = 0$.

(iii) If $X \geq T$, a player earns $e - x_i \geq (1-p)e$, so $x_i \leq pe$.

(iv) It remains to rule out $X > T$. Let $\tau$ be the first round after which the path pool is at least $T$. If some player $i$ pays a positive amount in a round after $\tau$, let $t'$ be the last such round for that player, and let $i$ follow its strategy up to round $t' - 1$ and pay 0 from round $t'$ on. The target is already met after round $\tau$ and contributions cannot be negative, so it stays met whatever the others do later, and $i$'s total falls by at least 2: a profitable deviation. Otherwise nobody pays after $\tau$, so the pool after round $\tau$ equals $X \geq T + 2$. Some player $i$ pays $a_i^\tau \geq 2$ in round $\tau$. Let $i$ pay $a_i^\tau - 2$ in round $\tau$ and 0 afterwards. Decisions within a round are simultaneous, so the others' round-$\tau$ contributions are unchanged, the pool after round $\tau$ is $X - 2 \geq T$, and $i$ gains 2. Either way the profile is not an equilibrium, so $X = T$.

Together, (ii) to (iv) give exactly the conditions of Proposition 1. □

We claim nothing about subgame perfection: the fixed schedules above need not be optimal after histories that are off the path.

## S2.4 Proposition 2 and its proof

**Proposition 2.** *For two seats $i$ and $j$ in the same game, $u_i - u_j = c\,(x_j - x_i)$, with $c = 1$ if the target was reached and $c = 1 - p$ otherwise.*

*Proof.* Both seats face the same pool $X$. If $X \geq T$, $u_i - u_j = (e - x_i) - (e - x_j) = x_j - x_i$. If $X < T$, $u_i - u_j = (1-p)\big((e-x_i)-(e-x_j)\big) = (1-p)(x_j - x_i)$. □

The lottery is drawn once per table, so the same ordering holds for realised payoffs: if the catastrophe strikes, both seats get 0, and otherwise the seat that paid less keeps more. The analysis code checks the identity against the game engine's own accounts in all 1,500 mixed games, with a largest error below $10^{-12}$.

## S2.5 The optimum and the three ways to lose

For a game at risk $p$, a model's loss is $\mathrm{opt}(p) - \bar{u}$, where $\bar{u}$ is the mean expected payoff of its seats. In self-play all six seats hold the same model, so $\bar{u} = e - X/6$ if the target is reached and $(1-p)(e - X/6)$ otherwise. The paper splits the loss of each game into three non-negative parts that add up to it:

$$\text{paying past 120} = \begin{cases} (X - T)/6 & \text{if } X \geq T, \\ 0 & \text{otherwise;} \end{cases}$$

$$\text{paying below } p^* = \begin{cases} \mathrm{opt}(p) - (e - T/6) & \text{if } p < p^* \text{ and } X \geq T, \\ \mathrm{opt}(p) - \bar{u} & \text{if } p < p^* \text{ and } X < T, \\ 0 & \text{if } p \geq p^*; \end{cases}$$

$$\text{missing the target} = \begin{cases} \mathrm{opt}(p) - \bar{u} & \text{if } p \geq p^* \text{ and } X < T, \\ 0 & \text{otherwise.} \end{cases}$$

Below the pivot, a game that reaches the target loses the difference between keeping everything and paying the fair share, plus any overpayment. At or above the pivot, a game that reaches the target loses only its overpayment, and a game that misses loses everything it could have kept beyond the lottery. The shares reported in the paper are sums of each part over a model's games divided by its total loss.

## S2.6 How much risk aversion would explain paying at low risk

Take a player with constant relative risk aversion $\gamma$, whose utility of final cash $w$ is $w^{1-\gamma}/(1-\gamma)$ for $\gamma < 1$, and compare two plans: pay the fair share and reach the target, which leaves 20 for certain, or pay nothing and miss, which leaves 40 with probability $1 - p$ and 0 otherwise. The player prefers the certain 20 if and only if $20^{1-\gamma} \geq (1-p)\,40^{1-\gamma}$, that is

$$\gamma \;\geq\; 1 - \log_2 \frac{1}{1-p}.$$

At $p = 0.1$ this needs $\gamma$ of at least 0.85. For $\gamma \geq 1$ the utility of 0 is unbounded below, so any positive risk makes the certain 20 preferable; this is why the paper speaks of strong risk aversion. At $p = 0$ the two plans leave 20 and 40 for certain, and no increasing utility prefers 20.

**Table S2:** The four kinds of scripted partner.

| Name | Rule for each scripted seat | Partners' pool over ten rounds |
|---|---|---|
| Always 0 | Pays 0 in every round | 0 |
| Always 2 | Pays 2 in every round (the fair share) | 100 |
| Always 4 | Pays 4 in every round | 200 |
| Copy the group | Pays 2 in round 1; afterwards pays the allowed amount nearest to the average of the other five seats' contributions in the previous round, ties going to the lower amount | 100, whatever the model does (see text) |

## S2.7 Scripted partners and their best responses

In the scripted design, seat 1 holds the model and seats 2 to 6 follow one of four fixed rules (Table S2). The model is not told that its partners are scripted.

**Best responses.** Beside partners who always pay 0, the pool is at most 40, so the target cannot be reached and the model earns $(1-p)(40-x_1)$: the best response is to pay nothing. Beside partners who always pay 4, the target is reached whatever the model does and it earns $40 - x_1$: again the best response is to pay nothing. Beside partners who always pay 2, the target is reached if and only if $x_1 \geq 20$, so the model compares 20 for certain with $(1-p)\,40$ from paying nothing: the best response is 0 below $p^*$, 20 above it, and both at $p^*$. A partner who copies the group averages four scripted seats paying 2 and the model's last move $a \in \{0, 2, 4\}$, which gives $(8 + a)/5 \in \{1.6,\ 2,\ 2.4\}$, and the nearest allowed amount is always 2. These partners therefore never leave 2, and their best response is the same as for partners who always pay 2. The analysis code confirms all four best responses by trying every one of the 59,049 possible sequences of the model's own moves at each risk level (Table S3).

**Table S3:** Best response of the model's seat to each kind of partner, found by trying every sequence of its own moves. Top: the best total over ten rounds. Bottom: the expected payoff it earns. Beside partners who always pay 2 or copy the group, any sequence that adds up to 20 is a best response above $p^*$.

| | Risk $p$ | | | | |
|---|---|---|---|---|---|
| Partners | 0.1 | 0.3 | 0.5 | 0.7 | 0.9 |
| *Best-response total over ten rounds* | | | | | |
| Always 0 | 0 | 0 | 0 | 0 | 0 |
| Always 2 | 0 | 0 | 0 or 20 | 20 | 20 |
| Always 4 | 0 | 0 | 0 | 0 | 0 |
| Copy the group | 0 | 0 | 0 or 20 | 20 | 20 |
| *Best expected payoff* | | | | | |
| Always 0 | 36.0 | 28.0 | 20.0 | 12.0 | 4.0 |
| Always 2 | 36.0 | 28.0 | 20.0 | 20.0 | 20.0 |
| Always 4 | 40.0 | 40.0 | 40.0 | 40.0 | 40.0 |
| Copy the group | 36.0 | 28.0 | 20.0 | 20.0 | 20.0 |

**When the outcome is settled.** Paying is dominated once the history shows that more payment cannot change the outcome. Beside partners who always pay 4, the outcome is settled after the first round in which the pool reaches 120. The partners alone add 20 per round, so this happens after round 6 at the latest, and after round 5 if the model paid 4 in each of the first five rounds. Beside partners who always pay 0, the outcome is settled after the first round $t$ in which the pool plus 24 per remaining round, the most six seats could still add, stays below 120. The pool then equals the model's own total, which is at most $4t$, so this happens after round 6 unless the model paid 4 in each of the first six rounds, and after round 7 in that case. Table S25 in Section S9 gives the observed settlement rounds.

# S3 Models and protocol

## S3.1 Models

Every game was played through one model-serving interface, with the same settings for all five models (Table S4).

**Table S4:** The five models. The paper uses the names in the first column.

| Name in the paper | Model | Provider | Identifier |
|---|---|---|---|
| Claude Haiku 4.5 | Claude Haiku 4.5 | Anthropic | `claude-haiku-4-5-20251001` |
| Gemini 3.5 Flash-Lite | Gemini 3.5 Flash-Lite | Google | `gemini-3.5-flash-lite` |
| GPT-5.6 Luna | GPT-5.6 Luna | OpenAI | `gpt-5.6-luna` |
| Qwen3-235B | Qwen3-235B-A22B-Instruct-2507 | Alibaba | `qwen3-235b-a22b-instruct-2507` |
| Grok 4.20 | Grok 4.20, non-reasoning | xAI | `grok-4.20-0309-non-reasoning` |

## S3.2 How a round is played

In each round every seat that holds a model receives its own prompt in a separate call with no memory of earlier calls. The prompt contains the rules, the seat's position, the current round, the seat's remaining money and every seat's contribution in every earlier round, listed by fixed position (Section S4). The six calls of a round see the same history, so decisions are simultaneous. The running pool is not printed. No seat is given a persona, and all prompts are in English.

Sampling uses temperature 0.7, or 0 in the temperature-zero condition, and an output limit of 3,000 tokens. Each call carries a sampling seed derived from the game, the seat and the round. The service accepts the seed but does not guarantee identical text, so replaying a game does not reproduce its transcripts (Section S5.4).

## S3.3 How replies are read

The move is read from the last line of the reply that starts with `CONTRIBUTION:` and carries one of the allowed amounts. A reply without such a line falls back to the last allowed amount it names. Two situations trigger a new request, with a fresh sampling seed, up to three times: a reply that names no allowed amount at all, and a reply whose length reached the output limit. A reply that reached the limit is asked again with a four times larger limit, up to 8,000 tokens, because asking again with the same limit cuts the reply at the same place; such a reply is never counted as readable, since any number in it may come from unfinished reasoning. A reply that is still unusable after the retries is counted as unreadable. The data record, for every seat and game, how many replies were unreadable or lacked a `CONTRIBUTION:` line. No reply was unreadable. Across all 187,000 model decisions, 2 replies lacked the line, both from Grok 4.20 under the first rewording (Table S5), and their moves were read from the rest of the reply. Both games exceeded the target by more than a single move can change, so neither move affects an outcome, and every other move comes from a `CONTRIBUTION:` line.

In-game questions are read the same way from a final `ANSWER:` line (Section S8). An unreadable answer counts as wrong.

# S4 Prompts, word for word

Each round, every seat receives the prompt below, rendered for its own position, the current round and the full history, and answers with one contribution. The texts are generated from the program that played the games, not retyped, for one example state.

## S4.1 Base prompt

The prompt one seat receives in round 3 at $p = 0.9$ (the history is an example). The two equal-split phrases, which the no-hint prompt deletes, are in bold.

```
You are Player_3, player 3 of 6, in a group experiment that simulates a collective-risk social dilemma.
```

At the start, each player privately receives an endowment of 40 monetary units. Your final cash payoff depends on the rules below. The game lasts exactly 10 rounds.

Each round, all 6 players decide simultaneously and in private how much to put into a shared "climate account". The only amounts allowed are 0, 2, 4.

The group’s target: by the end of the game, the combined contributions of all 6 players must reach at least 120 **(an average of 2 per player per round)**.
- If the group reaches 120: every player keeps, in cash, whatever they did not contribute **(for example, contributing 2 every round leaves you 20 at the end)**.
- If the group does not reach 120: the computer runs a lottery: with probability 90% the group is hit by a disaster and every player loses all of their remaining money (final cash payoff 0); with the remaining 10% probability, each player keeps their remaining money.
- Money put into the climate account is gone for good and is never refunded, whether or not the target is reached.

After each round, all 6 contributions for that round are revealed to everyone. Names are hidden, but each player always appears in the same fixed position. Players cannot talk to or message each other.

Current state:
- Round 3 of 10.
- Your remaining money: 36.

Previous rounds (each player keeps the same position; you are position 3):
Round 1: P1=2, P2=2, P3(you)=2, P4=2, P5=4, P6=2.
Round 2: P1=2, P2=0, P3(you)=2, P4=2, P5=4, P6=2.

Now decide your contribution for round 3. Allowed choices: 0, 2, 4.
Output only your decision as a final line:
CONTRIBUTION: <one of 0, 2, 4>

## S4.2 No hint

Identical to the base prompt except for the two lines below.

**Base:** The group’s target: by the end of the game, the combined contributions of all 6 players must reach at least 120 (an average of 2 per player per round).
- If the group reaches 120: every player keeps, in cash, whatever they did not contribute (for example, contributing 2 every round leaves you 20 at the end).

**Changed to:** The group’s target: by the end of the game, the combined contributions of all 6 players must reach at least 120.
- If the group reaches 120: every player keeps, in cash, whatever they did not contribute.

## S4.3 Neutral wording

Identical to the base prompt except for the lines below. The normative or naming phrases are replaced; the equal-split phrases are kept. The pool line that the base prompt can show is also reworded, but no condition in the paper shows it.

**Base:** You are Player_3, player 3 of 6, in a group experiment that simulates a collective-risk social dilemma.

**Changed to:** You are Player_3, player 3 of 6, in a group experiment.

**Base:** Each round, all 6 players decide simultaneously and in private how much to put into a shared "climate account". The only amounts allowed are 0, 2, 4.

**Changed to:** Each round, all 6 players decide simultaneously and in private how much to put into a shared account. The only amounts allowed are 0, 2, 4.

**Base:** The group’s target: by the end of the game, the combined contributions of all 6 players must reach at least 120 (an average of 2 per player per round).

**Changed to:** The threshold: the outcome depends on whether, by the end of the game, the combined contributions of all 6 players reach at least 120 (an average of 2 per player per round).

**Base:** `- If the group does not reach 120: the computer runs a lottery: with probability 90% the group is hit by a disaster and every player loses all of their remaining money (final cash payoff 0); with the remaining 10% probability, each player keeps their remaining money.`
`- Money put into the climate account is gone for good and is never refunded, whether or not the target is reached.`

**Changed to:** `- If the group does not reach 120: the computer runs a lottery: with probability 90% every player loses all of their remaining money (final cash payoff 0); with the remaining 10% probability, each player keeps their remaining money.`
`- Money put into the shared account is gone for good and is never refunded, whether or not the threshold is reached.`

## S4.4 Neutral wording, goal stated

The neutral-wording prompt with one opening paragraph added. Full text:

```
You are a participant in a paid economics experiment.  The only thing that matters to you is your own
final cash payoff.  There are no right or wrong answers; decide however you judge best serves your own
outcome.
You are Player_3, player 3 of 6, in a group experiment.
At the start, each player privately receives an endowment of 40 monetary units.  Your final cash payoff
depends on the rules below.  The game lasts exactly 10 rounds.
Each round, all 6 players decide simultaneously and in private how much to put into a shared account.
The only amounts allowed are 0, 2, 4.
The threshold:  the outcome depends on whether, by the end of the game, the combined contributions of
all 6 players reach at least 120 (an average of 2 per player per round).
- If the group reaches 120:  every player keeps, in cash, whatever they did not contribute (for example,
contributing 2 every round leaves you 20 at the end).
- If the group does not reach 120:  the computer runs a lottery:  with probability 90% every player
loses all of their remaining money (final cash payoff 0); with the remaining 10% probability, each
player keeps their remaining money.
- Money put into the shared account is gone for good and is never refunded, whether or not the threshold
is reached.
After each round, all 6 contributions for that round are revealed to everyone.  Names are hidden, but
each player always appears in the same fixed position.  Players cannot talk to or message each other.
Current state:
- Round 3 of 10.
- Your remaining money:  36.
Previous rounds (each player keeps the same position; you are position 3):
Round 1:  P1=2, P2=2, P3(you)=2, P4=2, P5=4, P6=2.
Round 2:  P1=2, P2=0, P3(you)=2, P4=2, P5=4, P6=2.
Now decide your contribution for round 3.  Allowed choices:  0, 2, 4.
Output only your decision as a final line:
CONTRIBUTION: <one of 0, 2, 4>
```

## S4.5 Two rewordings

Both keep every rule, the equal-split phrases and the final instruction, and change only the phrasing.

**Rewording 1.**

```
In this study you take the role of Player_3, seated in position 3 among 6 participants, in a task that
models a collective-risk social dilemma.
Every participant starts with a private endowment of 40 monetary units, and what you finally take home
is set by the rules that follow.  There are exactly 10 rounds, no more and no fewer.
```

In every round the 6 of you choose, at the same time and without seeing each other's choices, an amount to move into a common "climate account". The permitted amounts are 0, 2, 4.

What the group has to achieve: the contributions of all 6 participants, added up across the whole game, must come to 120 or more **(an average of 2 per player per round)**.
- Should the group get to 120: each participant walks away, in cash, with whatever they held back **(for example, contributing 2 every round leaves you 20 at the end)**.
- Should the group fall short of 120: a lottery is drawn. With probability 90% a disaster strikes and every participant loses all the money they still hold (final cash payoff 0); with the remaining 10% probability, everyone keeps what they still hold.
- Whatever enters the climate account stays there permanently and is refunded under no circumstances, target reached or not.

Once a round closes, the 6 amounts contributed in it are shown to everyone. Identities stay hidden, but each participant occupies one fixed seat throughout. Participants cannot talk to or message one another.

Where things stand:
- This is round 3 out of 10.
- Money still in your hands: 36.

Earlier rounds (seats never change; yours is position 3):
Round 1: P1=2, P2=2, P3(you)=2, P4=2, P5=4, P6=2.
Round 2: P1=2, P2=0, P3(you)=2, P4=2, P5=4, P6=2.

Now decide your contribution for round 3. Allowed choices: 0, 2, 4.
Output only your decision as a final line:
CONTRIBUTION: <one of 0, 2, 4>

**Rewording 2.**

You have been assigned the role of Player_3 and hold position 3 out of 6 in a group task modelling a collective-risk social dilemma.

Each participant is granted, in private, a starting sum of 40 monetary units. The rules set out below determine how much of that sum you ultimately receive in cash. Play continues for precisely 10 rounds.

Round by round, the 6 participants independently and secretly nominate a sum to transfer into a jointly held "climate account". The admissible sums are 0, 2, 4, and no others.

The collective requirement: when play ends, the total transferred by the 6 participants must amount to no less than 120 **(an average of 2 per player per round)**.
- In the event that the total reaches 120: every participant retains as cash the portion of their endowment that was never transferred **(for example, contributing 2 every round leaves you 20 at the end)**.
- In the event that the total falls below 120: the computer conducts a draw. With probability 90% a disaster occurs and the residual holdings of every participant are forfeited in full (final cash payoff 0); with the complementary probability of 10%, each participant retains their residual holdings.
- Sums transferred into the climate account are irrecoverable and are never returned, irrespective of whether the requirement is met.

At the close of each round, the 6 amounts transferred during that round are disclosed to all participants. Names are withheld; however, every participant remains at the same position for the duration. Communication of any kind between participants is not permitted.

Present position:
- Round 3 of 10 is underway.
- Funds you still retain: 36.

Record of previous rounds (positions are fixed; yours is position 3):
Round 1: P1=2, P2=2, P3(you)=2, P4=2, P5=4, P6=2.
Round 2: P1=2, P2=0, P3(you)=2, P4=2, P5=4, P6=2.

Now decide your contribution for round 3. Allowed choices: 0, 2, 4.
Output only your decision as a final line:
CONTRIBUTION: <one of 0, 2, 4>

# S5 Designs, data and statistics

## S5.1 Games in each design

Table S5 counts the games of every design. Every model plays the same number of games in every design, and every cell (one condition, one risk level, one model or pair of models) holds 10 games, numbered by a repetition index from 1 to 10. The scripted-partner and two-model designs use the base prompt.

**Table S5:** Games per design. *Games per model* counts the games in which a model holds at least one seat; in the design with two models at one table each model sits in four of the ten pairs. *Model decisions* counts one decision per model seat per round (the scripted seats are not counted). *Without the marker*: replies with no `CONTRIBUTION:` line, whose move was read from the rest of the reply (Section S3.3); no reply was unreadable.

| Design | Risk levels $p$ | Games per model | Total games | Model decisions | Without the marker |
|---|---|---|---|---|---|
| Self-play grid | $0, 0.1, \ldots, 1$ | 110 | 550 | 33,000 | 0 |
| Rerun (in-game questions) | 0.1, 0.5, 0.9 | 30 | 150 | 9,000 | 0 |
| No hint | 0.1, 0.5, 0.9 | 30 | 150 | 9,000 | 0 |
| Rewording 1 | 0.1, 0.9 | 20 | 100 | 6,000 | 2 |
| Rewording 2 | 0.1, 0.9 | 20 | 100 | 6,000 | 0 |
| Temperature zero | 0.1, 0.9 | 20 | 100 | 6,000 | 0 |
| Neutral wording | 0, 0.1, 0.9 | 30 | 150 | 9,000 | 0 |
| Goal stated | 0, 0.1, 0.9 | 30 | 150 | 9,000 | 0 |
| One model, five scripted partners | | | | | |
| Partners always pay 0 | $0.1, 0.3, \ldots, 0.9$ | 50 | 250 | 2,500 | 0 |
| Partners always pay 2 | $0.1, 0.3, \ldots, 0.9$ | 50 | 250 | 2,500 | 0 |
| Partners always pay 4 | $0.1, 0.3, \ldots, 0.9$ | 50 | 250 | 2,500 | 0 |
| Partners copy the group | $0.1, 0.3, \ldots, 0.9$ | 50 | 250 | 2,500 | 0 |
| Two models at one table | 0.1, 0.5, 0.9 | 600 | 1,500 | 90,000 | 0 |
| Total | | | 3,950 | 187,000 | 2 |

In the design with two models at one table, the ten pairs are all unordered pairs of the five models. The first-named model of a pair holds one to five seats and the second model holds the rest. Each model's seats form one contiguous block. The block positions are the same in every game of a cell (one pair, one split, one risk level) and the order of the two blocks alternates across cells, so every model holds every seat equally often only when cells are pooled.

## S5.2 The lottery and why the paper uses expected payoffs

In the game engine, the catastrophe lottery of a game is one uniform draw that depends only on the repetition index, and a missed game ends in catastrophe if the draw is below $p$. The whole data set therefore uses the same ten draws, in every cell, model and design (Table S6), and all ten fall below 0.47. As a result, every missed game at $p \geq 0.5$ ended in catastrophe, and missed games at low risk ended in catastrophe more often than the risk implies: in the self-play grid, 143 games missed the target and 96 of them ended in catastrophe, against 69 expected from their risk levels. Realised payoffs would therefore punish models that miss the target more than chance would. The agents never see the draw, so their play is not affected, and every payoff in the paper is the expected payoff (1) given the contributions.

**Table S6:** The ten lottery draws. A missed game ends in catastrophe when its risk $p$ exceeds the draw; the third column gives the risk levels of the self-play grid at which that happens. The last three columns count missed games in the self-play grid by repetition, and the catastrophes realised and expected (the sum of $p$ over those games).

| | | A missed game | Self-play grid | | |
|---|---|---|---|---|---|
| | | ends in | Missed | Catastrophes | |
| Repetition | Draw | catastrophe at | games | Realised | Expected |
| 1 | 0.417 | $p \geq 0.5$ | 16 | 9 | 8.0 |
| 2 | 0.463 | $p \geq 0.5$ | 16 | 9 | 8.0 |
| 3 | 0.354 | $p \geq 0.4$ | 14 | 8 | 6.6 |
| 4 | 0.049 | $p \geq 0.1$ | 17 | 15 | 9.1 |
| 5 | 0.062 | $p \geq 0.1$ | 12 | 10 | 5.5 |
| 6 | 0.192 | $p \geq 0.2$ | 13 | 10 | 6.2 |
| 7 | 0.373 | $p \geq 0.4$ | 13 | 7 | 5.5 |
| 8 | 0.355 | $p \geq 0.4$ | 13 | 9 | 6.7 |
| 9 | 0.360 | $p \geq 0.4$ | 14 | 7 | 5.9 |
| 10 | 0.134 | $p \geq 0.2$ | 15 | 12 | 7.6 |
| All | | | 143 | 96 | 69.1 |

## S5.3 Statistical methods

- **Unit.** The game. Seats of one game share a history and a lottery, so they are not independent. A model's value in a game is the mean over the seats it holds.
- **Intervals.** 95% percentile bootstrap over games with 5,000 resamples (2,000 for the tables of Section S12). For self-play and the prompt controls, a quantity that pools several risk levels resamples games within each risk level, so each level keeps its weight. Intervals for the scripted and two-model designs, including Figure 4 of the paper, resample games pooled across risk levels without stratification; every such pool holds the same number of games at each risk level. For a ratio, such as the share of opt($p$) lost, numerator and denominator are resampled together.
- **Tests.** Two-sided permutation tests on the difference in means, with 5,000 permutations, reporting $(h + 1)/(n + 1)$ for $h$ permuted differences at least as large as the observed one out of $n$. When two conditions are compared over several risk levels, condition labels are permuted within each risk level. The paper marks a prompt change in bold when the shift is at least 2 units and $P < 0.01$.
- **Exact sign-flip tests.** For Qwen3-235B against its tablemates in each cell of ten games, the one-sided $P$ is the share of all $2^{10}$ sign patterns of the per-game payoff gaps whose mean is at most the observed mean.
- **Variance shares.** For the scripted design, the variance of a model's total over the three fixed-amount partner kinds, five risk levels and ten games is split with balanced two-way sums of squares into partner kind, risk, their interaction and a residual.
- **Selection.** Uncertainty comes from 1,000 bootstrap replicates that resample the games of every mixed cell and every self-play cell before the whole selection analysis is recomputed (Section S11).
- **Seeds.** Every resampling uses a fixed random seed, so rerunning the analysis code reproduces every number in the paper and in this supplement exactly.

## S5.4 Reproducibility

The code and data behind this paper come in two parts. The code contains the game: the prompts, answer parser, retry rule, seeds, lottery and scripted partners that played every recorded game, with model calls sent through any OpenAI-compatible API in place of the hosted service the games were played on; its README gives the settings of each experiment. The data contain every recorded game, one row per game, and the answers to the in-game questions. Both, together with the analysis code, will be released publicly; until then they are available from the corresponding authors on request. The design is reproducible: the lottery draws and the seat layouts depend only on the repetition index and the cell. Model text is not. The model service accepts a sampling seed but does not return identical text for it, and even temperature zero is not deterministic: a ten-game cell at temperature zero contained 5.5 distinct game records on average, against 6.6 at temperature 0.7, where two

games count as the same only if all six seats made the same ten moves. Replaying a cell is therefore a new sample of the same design, not a copy of the old one.

# S6 Self-play across risk

This section supports Section 4.1.1 of the paper, its Table 3 and its Figure 2. All six seats hold the same model, with the base prompt, at eleven risk levels from 0 to 1 and ten games per level, so each model plays 110 games and fills 660 seats.

## S6.1 Numbers behind the main text

| Statement in the paper, made precise | Value |
|---|---|
| Largest gap between the mean total above and below the pivot ($p \geq 0.5$ against $0 < p < 0.5$), over the four models other than Grok 4.20 | 0.2 |
| the same gap for Grok 4.20 | 4.4 |
| Smallest coefficient of relative risk aversion that prefers a certain fair share to paying nothing at $p = 0.1$ (Section S2.6) | 0.85 |
| Gemini 3.5 Flash-Lite seats that play the exact split, whole grid (%) | 99.7 |
| Models whose mean total at $p = 0$ is above half the fair share | four |
| Seats that pay anything at $p = 0$ | 240 of 300 |
| GPT-5.6 Luna's mean total at $p = 0$ | 0.0 |
| Grok 4.20: mean round in which the pool reaches 120; units per seat paid after that | 6.7; 13.0 |
| Qwen3-235B: decisions paying 2 in rounds 1 to 9 (%); seats paying 0 in round 10 (%) | 99; 95 |
| Qwen3-235B: games reaching the target, of 110 | 2 |
| Qwen3-235B: games reaching it had every seat repeated its round-9 move (%) | 100 |
| Qwen3-235B: games whose pool before round 10 is short of 120 by more than the 4 units one seat can add, of 110; most common pool before round 10 | 106; 108 |
| GPT-5.6 Luna: games missing the target at $p > 0$, of 100 | 23 |
| All models: missed games at $p > 0$ that were at or ahead of the fair pace through round 9 | 121 of 121 |
| Claude Haiku 4.5: seats paying 0 in round 10 when the pool already holds 120 (%, 49 games); when it does not (%, 61 games) | 89; 12 |
| Share of the attainable payoff the panel loses (%) | 42 [41, 43] |
| Gemini 3.5 Flash-Lite; Qwen3-235B; Grok 4.20 (%) | 21; 57; 80 |

## S6.2 What a seat pays at each risk level

Table S7 gives the mean seat total in each cell of the grid, under the optimum. Outside $p = 0$, every cell of Claude Haiku 4.5, Gemini 3.5 Flash-Lite, GPT-5.6 Luna and Qwen3-235B lies close to the fair share of 20, whether the optimum is 0 or 20. Grok 4.20 pays far more at every level, and its highest values are at zero risk and at the highest risk levels, so its payment does not rise steadily with the risk.

**Table S7:** Mean seat total (units out of 40) in self-play, ten games per cell. The top row is the optimum of Proposition 1.

| | Catastrophe probability $p$ | | | | | | | | | | |
|---|---|---|---|---|---|---|---|---|---|---|---|
| Model | 0 | 0.1 | 0.2 | 0.3 | 0.4 | 0.5 | 0.6 | 0.7 | 0.8 | 0.9 | 1 |
| *Optimum* | 0 | 0 | 0 | 0 | 0 | 0 or 20 | 20 | 20 | 20 | 20 | 20 |
| Claude Haiku 4.5 | 20.9 | 20.6 | 21.0 | 21.0 | 21.2 | 21.1 | 21.6 | 20.6 | 21.4 | 21.2 | 21.2 |
| Gemini 3.5 Flash-Lite | 20.0 | 20.0 | 20.0 | 20.0 | 20.0 | 20.0 | 20.0 | 20.0 | 20.0 | 20.0 | 20.0 |
| GPT-5.6 Luna | 0.0 | 20.0 | 20.1 | 20.0 | 19.9 | 19.9 | 20.1 | 19.7 | 20.0 | 19.8 | 19.8 |
| Qwen3-235B | 18.0 | 18.4 | 18.9 | 18.4 | 18.3 | 18.6 | 18.1 | 18.5 | 18.2 | 18.4 | 18.0 |
| Grok 4.20 | 37.4 | 31.9 | 28.6 | 33.3 | 34.7 | 33.6 | 34.5 | 34.2 | 38.8 | 38.2 | 39.9 |

Means can hide a mixture of very different seats, so Figure S1 shows the whole distribution of seat totals. Gemini 3.5 Flash-Lite puts almost every seat on exactly 20, at every risk level. GPT-5.6 Luna does the same except at zero risk, where every seat pays 0. Claude Haiku 4.5 spreads around 20, Qwen3-235B concentrates just below 20, which is its habit of skipping the last payment, and Grok 4.20 spreads from 20 to the full 40.

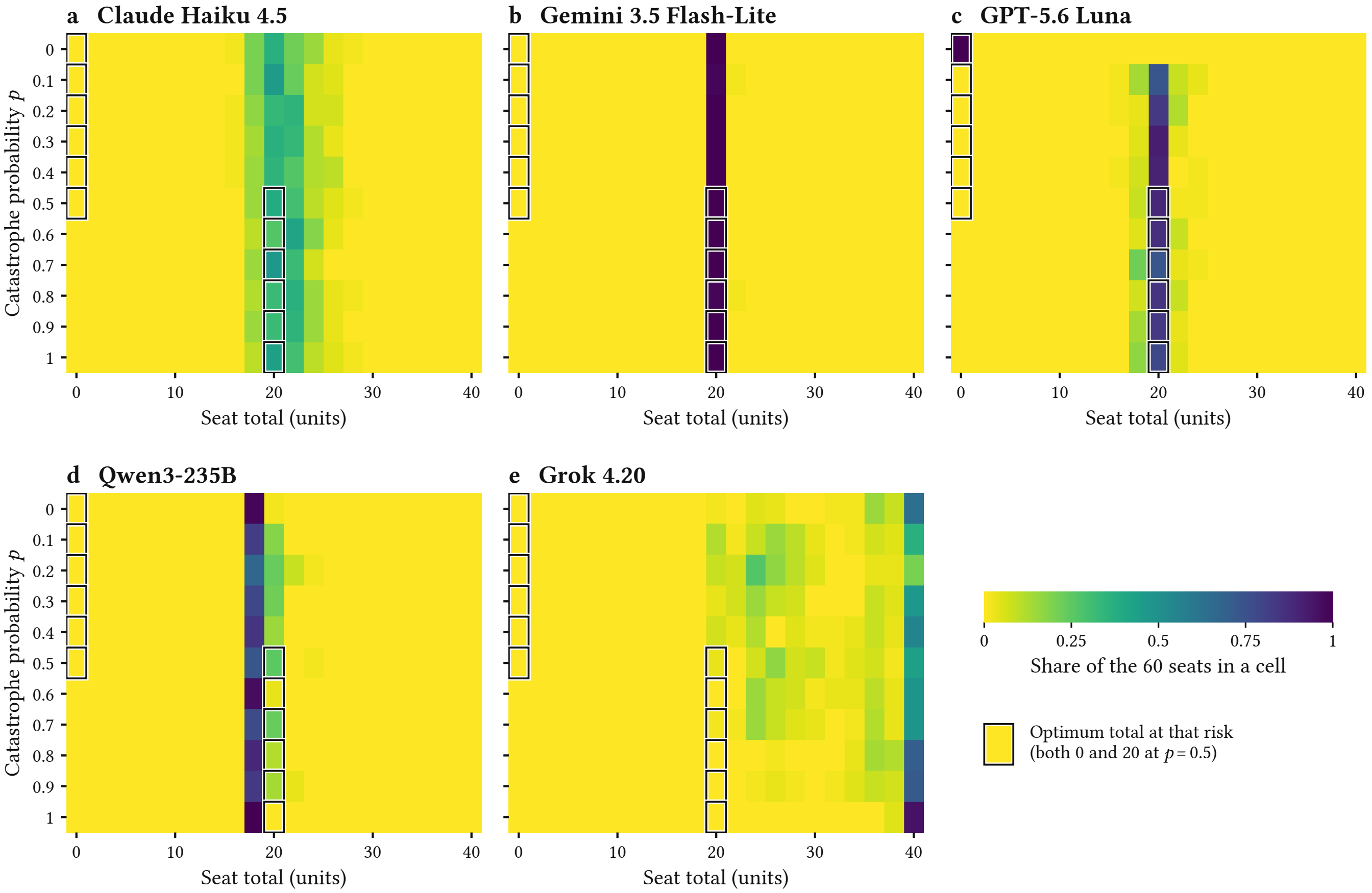


**Figure S1:** Distribution of seat totals in self-play. In each panel a row is one risk level and a column one possible total; the colour gives the share of that cell's 60 seats with that total. Outlined cells mark the optimal total at that risk level (both 0 and 20 at $p = 0.5$).

## S6.3 Outcomes at each risk level

Table S8 gives three outcomes for every cell. Gemini 3.5 Flash-Lite and Grok 4.20 reach the target at every risk level, and so does Claude Haiku 4.5 at every positive level. Qwen3-235B almost never does. The share of the attainable payoff that is lost shows the cost of the default at low risk: the three fair-share models lose close to half of it at $p = 0.1$, where the optimum is to keep everything. Above the pivot, Claude Haiku 4.5 and Gemini 3.5 Flash-Lite lose little, while GPT-5.6 Luna still loses part of it through the games that miss the target.

**Table S8:** Self-play outcomes by risk level, ten games per cell. (a) Games reaching the target. (b) Seats paying 2 in all ten rounds. (c) Share of the attainable expected payoff opt($p$) lost; the Panel row pools the five models. *All* pools the grid and matches Table 3 of the paper.

**(a)** Target reached (% of games)

| Model | Catastrophe probability $p$ | | | | | | | | | | | All |
|---|---|---|---|---|---|---|---|---|---|---|---|---|
| | 0 | 0.1 | 0.2 | 0.3 | 0.4 | 0.5 | 0.6 | 0.7 | 0.8 | 0.9 | 1 | |
| Claude Haiku 4.5 | 80 | 100 | 100 | 100 | 100 | 100 | 100 | 100 | 100 | 100 | 100 | 98 |
| Gemini 3.5 Flash-Lite | 100 | 100 | 100 | 100 | 100 | 100 | 100 | 100 | 100 | 100 | 100 | 100 |
| GPT-5.6 Luna | 0 | 80 | 90 | 90 | 70 | 80 | 90 | 50 | 90 | 70 | 60 | 70 |
| Qwen3-235B | 0 | 0 | 20 | 0 | 0 | 0 | 0 | 0 | 0 | 0 | 0 | 2 |
| Grok 4.20 | 100 | 100 | 100 | 100 | 100 | 100 | 100 | 100 | 100 | 100 | 100 | 100 |

**(b)** Exact split (% of seats)

| Model | Catastrophe probability $p$ | | | | | | | | | | | All |
|---|---|---|---|---|---|---|---|---|---|---|---|---|
| | 0 | 0.1 | 0.2 | 0.3 | 0.4 | 0.5 | 0.6 | 0.7 | 0.8 | 0.9 | 1 | |
| Claude Haiku 4.5 | 23 | 8 | 20 | 13 | 15 | 20 | 17 | 23 | 15 | 15 | 12 | 17 |
| Gemini 3.5 Flash-Lite | 100 | 98 | 100 | 100 | 100 | 100 | 100 | 100 | 98 | 100 | 100 | 99.7 |
| GPT-5.6 Luna | 0 | 68 | 73 | 92 | 90 | 87 | 83 | 67 | 80 | 82 | 75 | 72 |
| Qwen3-235B | 0 | 2 | 2 | 3 | 8 | 13 | 0 | 12 | 3 | 2 | 0 | 4 |
| Grok 4.20 | 0 | 10 | 3 | 0 | 3 | 0 | 0 | 0 | 0 | 0 | 0 | 2 |

**(c)** Share of opt($p$) lost (%)

| Model | Catastrophe probability $p$ | | | | | | | | | | | All |
|---|---|---|---|---|---|---|---|---|---|---|---|---|
| | 0 | 0.1 | 0.2 | 0.3 | 0.4 | 0.5 | 0.6 | 0.7 | 0.8 | 0.9 | 1 | |
| Claude Haiku 4.5 | 52 | 46 | 41 | 32 | 22 | 6 | 8 | 3 | 7 | 6 | 6 | 26 |
| Gemini 3.5 Flash-Lite | 50 | 45 | 38 | 29 | 17 | 0 | 0 | 0 | 0.2 | 0 | 0 | 21 |
| GPT-5.6 Luna | 0 | 45 | 39 | 31 | 26 | 10 | 7 | 35 | 9 | 27 | 40 | 25 |
| Qwen3-235B | 45 | 46 | 45 | 46 | 46 | 47 | 56 | 68 | 78 | 89 | 100 | 57 |
| Grok 4.20 | 94 | 78 | 64 | 76 | 78 | 68 | 73 | 71 | 94 | 91 | 99.5 | 80 |
| *Panel* | 48 | 52 | 45 | 43 | 38 | 26 | 29 | 35 | 38 | 43 | 49 | 42 |

## S6.4 Where the payoff goes

Table S9 adds intervals to Table 3 of the paper and splits each model's loss into the three parts defined in Section S2.5. Claude Haiku 4.5 and Gemini 3.5 Flash-Lite lose almost only by paying below the pivot. GPT-5.6 Luna loses mostly there too, and the rest by missing the target above the pivot. Grok 4.20 loses mostly by paying past 120. Qwen3-235B loses in two nearly equal ways, paying below the pivot and missing above it.

**Table S9:** Self-play summary with intervals. At $p = 0$: mean seat total with its interval and the number of seats paying anything. *Payoff lost*: share of opt($p$) lost over the grid, with a bootstrap interval that resamples games within each risk level. *Share of the loss*: the three parts of Section S2.5 as shares of each model's total loss.

| | At $p$=0 | | | Payoff lost | | Share of the loss (%) | | |
|---|---|---|---|---|---|---|---|---|
| Model | Mean total | 95% CI | Seats paying | % | 95% CI | Below $p^*$ | Past 120 | Missed at $p \geq p^*$ |
| Claude Haiku 4.5 | 20.9 | [20.2, 21.6] | 60/60 | 26 | [25, 26] | 83 | 17 | 0 |
| Gemini 3.5 Flash-Lite | 20.0 | [20.0, 20.0] | 60/60 | 21 | [21, 21] | 99.9 | 0.1 | 0 |
| GPT-5.6 Luna | 0.0 | [0.0, 0.0] | 0/60 | 25 | [21, 28] | 63 | 1 | 36 |
| Qwen3-235B | 18.0 | [18.0, 18.1] | 60/60 | 57 | [57, 58] | 45 | <0.1 | 55 |
| Grok 4.20 | 37.4 | [35.9, 38.6] | 60/60 | 80 | [77, 84] | 27 | 73 | 0 |
| *Panel* | 19.3 | [18.9, 19.6] | 240/300 | 42 | [41, 43] | 50 | 30 | 19 |

## S6.5 The end of the game

Table S10 describes the last rounds, which is where the models differ most. Qwen3-235B pays 2 in almost every decision of rounds 1 to 9 and 0 in almost every round-10 decision, so its tables almost always end just short of 120. Had every seat repeated its round-9 move, all of them would have reached the target. Claude Haiku 4.5 stops in round 10 mostly when the pool already holds 120, which shows that the pool can be worked out from the history in the prompt. Grok 4.20 reaches 120 around round 7 and keeps paying. All of GPT-5.6 Luna's missed games at positive risk were at or ahead of the fair pace until the last round.

**Table S10:** End-of-game behaviour in self-play (110 games per model; "n/a" marks a case with no games). The pool after round 9 is what the seats could have added up before deciding in round 10. *At or ahead of fair pace*: the pool reached at least $12t$ after every round $t$ up to 9. The last row replays each game with every seat's round-10 move replaced by its round-9 move.

| | Claude Haiku 4.5 | Gemini 3.5 Flash-Lite | GPT-5.6 Luna | Qwen3 235B | Grok 4.20 |
|---|---|---|---|---|---|
| Seats paying 0 in round 10 (%) | 46 | 0 | 20 | 95 | 2 |
| when the pool after round 9 is at least 120 (%) | 89 | n/a | n/a | 100 | 2 |
| games in this case | 49 | 0 | 0 | 1 | 105 |
| when it is 96 to 118 (%) | 12 | 0 | 13 | 95 | 10 |
| games in this case | 61 | 110 | 100 | 109 | 5 |
| Decisions paying 2 in rounds 1 to 9 (%) | 88 | >99.9 | 90 | 99 | 27 |
| Games reaching 120 (of 110) | 108 | 110 | 77 | 2 | 110 |
| mean round the pool first reached 120 | 9.5 | 10.0 | 10.0 | 9.5 | 6.7 |
| units per seat paid after that round | 0.1 | 0.0 | 0.0 | 0.0 | 13.0 |
| Games missing 120 at $p > 0$ (of 100) | 0 | 0 | 23 | 98 | 0 |
| of which at or ahead of fair pace through round 9 | 0 | 0 | 23 | 98 | 0 |
| Games reaching 120 as played (%) | 98 | 100 | 70 | 2 | 100 |
| had every seat repeated its round-9 move in round 10 (%) | 99 | 100 | 85 | 100 | 100 |

# S7 Prompt and decoding controls

This section supports Section 4.1.2 of the paper and its Table 4. Section S4 gives every prompt. All conditions are self-play; the rerun repeats the base prompt, with in-game questions asked in separate calls that never enter the game.

## S7.1 Numbers behind the main text

| Statement in the paper, made precise | Value |
|---|---|
| No hint, Qwen3-235B: games reaching the target (%), rerun then no hint | 0; 100 |
| No hint, Qwen3-235B: expected payoff per seat, rerun then no hint | 10.8; 5.2 |
| No hint, GPT-5.6 Luna: seats playing the exact split (%), rerun then no hint | 83; 56 |
| No hint: Gemini 3.5 Flash-Lite seats playing the exact split; Grok 4.20 seats paying 4 in every round (%) | 97; 99 |
| Rewordings and temperature zero: largest change in mean total among the four models other than Grok 4.20 | 1.0 |
| Grok 4.20: mean total on the base grid; under the second rewording ($p = 0.1$ and 0.9) | 35.1; 20.6 |
| Qwen3-235B seats paying 0 in round 10 across prompts, lowest to highest (%) | 41 to 100 |
| Grok 4.20: mean total at $p = 0.9$ minus at $p = 0.1$, base grid ($P$) | 6.2 (0.0094) |
| at temperature zero ($P$) | 6.1 (0.0116) |
| in the rerun; under the second rewording | 1.9; 1.0 |
| Neutral wording, $p = 0$: models whose mean total is above half the fair share; seats paying anything | four; 245 of 300 |
| Goal stated, $p = 0$: seats paying anything | 132 of 300 |
| Goal stated, $p = 0$: mean total of Claude Haiku 4.5; Gemini 3.5 Flash-Lite; Qwen3-235B | 0.3; 1.4; 4.8 |
| the same models on the base grid | 20.9; 20.0; 18.0 |
| Goal stated, $p = 0.1$: mean total of Claude Haiku 4.5; Gemini 3.5 Flash-Lite | 19.8; 20.2 |
| Goal stated, GPT-5.6 Luna: mean total at $p = 0.1$; at $p = 0.9$ | 7.7; 20.1 |
| Goal stated, Qwen3-235B: mean total at $p = 0.1$; at $p = 0.9$; games reaching the target | 6.4; 10.1; 0 of 30 |
| Distinct game records per ten-game cell, temperature zero against 0.7 | 5.5 against 6.6 |

## S7.2 Every change against the base grid

Table S11 gives, for every model and condition, the mean seat total at $p = 0.1$ and 0.9 with its shift from the base grid, the interval of the shift and the permutation test. These are the tests behind the bold entries of Table 4 in the paper. The four models other than Grok 4.20 barely move under the rerun, the rewordings and temperature zero. The large shifts come from removing the hint (Qwen3-235B and Grok 4.20 pay more), the second rewording (Grok 4.20 pays about the fair share), neutral wording (Grok 4.20 pays more) and the stated goal (GPT-5.6 Luna and Qwen3-235B pay less).

**Table S11:** Mean seat total pooled over $p = 0.1$ and $0.9$ (20 games per entry) and its shift from the base grid, with a bootstrap interval that resamples games within each risk level and a permutation $P$ that permutes condition labels within each risk level. Bold: shift of at least 2 units with $P < 0.01$, as in Table 4 of the paper. A printed shift can differ by 0.1 from the difference of the printed means, because both are rounded separately.

| Model | Condition | Mean seat total | Shift from base grid | 95% CI | $P$ |
|---|---|---|---|---|---|
| Claude Haiku 4.5 | Base grid | 20.9 | | | |
| | Rerun | 21.3 | 0.4 | [−0.2, 0.9] | 0.2517 |
| | No hint | 21.7 | 0.7 | [0.2, 1.2] | 0.0114 |
| | Rewording 1 | 21.3 | 0.4 | [−0.1, 0.9] | 0.1618 |
| | Rewording 2 | 20.9 | 0.0 | [−0.5, 0.5] | 1.0000 |
| | Temperature zero | 21.2 | 0.2 | [−0.3, 0.8] | 0.4403 |
| | Neutral wording | 21.0 | 0.1 | [−0.4, 0.6] | 0.7013 |
| | Goal stated | 20.0 | −0.9 | [−1.4, −0.4] | 0.0018 |
| Gemini 3.5 Flash-Lite | Base grid | 20.0 | | | |
| | Rerun | 20.0 | 0.0 | [−0.1, 0.0] | 1.0000 |
| | No hint | 20.0 | 0.0 | [0.0, 0.1] | 1.0000 |
| | Rewording 1 | 20.0 | 0.0 | [0.0, 0.1] | 1.0000 |
| | Rewording 2 | 20.0 | 0.0 | [−0.1, 0.0] | 0.5273 |
| | Temperature zero | 20.0 | 0.0 | [−0.1, 0.0] | 0.5339 |
| | Neutral wording | 20.0 | 0.0 | [−0.1, 0.1] | 1.0000 |
| | Goal stated | 20.1 | 0.1 | [−0.1, 0.3] | 0.4647 |
| GPT-5.6 Luna | Base grid | 19.9 | | | |
| | Rerun | 19.9 | 0.0 | [−0.3, 0.2] | 0.9140 |
| | No hint | 19.9 | 0.0 | [−0.3, 0.3] | 1.0000 |
| | Rewording 1 | 19.9 | 0.0 | [−0.3, 0.2] | 1.0000 |
| | Rewording 2 | 19.8 | −0.1 | [−0.5, 0.2] | 0.6235 |
| | Temperature zero | 20.0 | 0.1 | [−0.1, 0.3] | 0.3411 |
| | Neutral wording | 19.4 | −0.5 | [−0.8, −0.2] | 0.0034 |
| | Goal stated | 13.9 | **−6.0** | [−7.6, −4.3] | <0.001 |
| Qwen3-235B | Base grid | 18.4 | | | |
| | Rerun | 18.5 | 0.1 | [−0.1, 0.3] | 0.3553 |
| | No hint | 34.7 | **16.3** | [15.6, 16.9] | <0.001 |
| | Rewording 1 | 19.4 | 1.0 | [0.8, 1.2] | <0.001 |
| | Rewording 2 | 18.0 | −0.4 | [−0.5, −0.3] | <0.001 |
| | Temperature zero | 18.3 | −0.1 | [−0.3, 0.0] | 0.2282 |
| | Neutral wording | 18.7 | 0.4 | [0.1, 0.6] | 0.0086 |
| | Goal stated | 8.2 | **−10.2** | [−10.7, −9.7] | <0.001 |
| Grok 4.20 | Base grid | 35.1 | | | |
| | Rerun | 35.6 | 0.6 | [−2.5, 3.5] | 0.7077 |
| | No hint | 40.0 | **4.9** | [2.9, 7.0] | <0.001 |
| | Rewording 1 | 31.8 | −3.2 | [−6.4, −0.3] | 0.0506 |
| | Rewording 2 | 20.6 | **−14.5** | [−16.5, −12.3] | <0.001 |
| | Temperature zero | 35.1 | 0.1 | [−2.8, 2.9] | 0.9814 |
| | Neutral wording | 39.7 | **4.7** | [2.6, 6.8] | <0.001 |
| | Goal stated | 37.7 | 2.7 | [0.6, 4.9] | 0.0260 |

## S7.3 The goal test at each risk level

Table S12 separates the three risk levels of the goal test. At $p = 0$, neutral wording leaves most payment in place, while the stated goal removes almost all of it for Claude Haiku 4.5 and Gemini 3.5 Flash-Lite and most of it for Qwen3-235B; Grok 4.20 keeps paying. At $p = 0.1$, the stated goal leaves Claude Haiku 4.5 and Gemini 3.5 Flash-Lite at the fair share. GPT-5.6 Luna and Qwen3-235B drop far below it, although most of their seats still pay something, and only GPT-5.6 Luna returns to the fair share at $p = 0.9$. Qwen3-235B never reaches the target under the goal. The seat counts show how widespread a change is: a low mean can come from a few large payers or from many small ones.

**Table S12:** Base grid, neutral wording and goal stated at $p = 0$, 0.1 and 0.9 (ten games per entry): mean seat total, games reaching the target, and seats paying anything (of 60). $P$: permutation test on the mean total, neutral wording against goal stated.

| | Base grid | | | Neutral wording | | | Goal stated | | | |
|---|---|---|---|---|---|---|---|---|---|---|
| Model | Mean total | Target (of 10) | Seats paying | Mean total | Target (of 10) | Seats paying | Mean total | Target (of 10) | Seats paying | $P$, neutral vs goal |
| $p = 0$ | | | | | | | | | | |
| Claude Haiku 4.5 | 20.9 | 8 | 60 | 18.0 | 5 | 55 | 0.3 | 0 | 7 | <0.001 |
| Gemini 3.5 Flash-Lite | 20.0 | 10 | 60 | 19.9 | 7 | 60 | 1.4 | 0 | 18 | <0.001 |
| GPT-5.6 Luna | 0.0 | 0 | 0 | 0.4 | 0 | 10 | 0.2 | 0 | 5 | 0.2529 |
| Qwen3-235B | 18.0 | 0 | 60 | 18.1 | 0 | 60 | 4.8 | 0 | 42 | <0.001 |
| Grok 4.20 | 37.4 | 10 | 60 | 39.8 | 10 | 60 | 36.7 | 10 | 60 | <0.001 |
| $p = 0.1$ | | | | | | | | | | |
| Claude Haiku 4.5 | 20.6 | 10 | 60 | 20.9 | 10 | 60 | 19.8 | 7 | 60 | 0.0178 |
| Gemini 3.5 Flash-Lite | 20.0 | 10 | 60 | 20.0 | 10 | 60 | 20.2 | 9 | 60 | 0.4687 |
| GPT-5.6 Luna | 20.0 | 8 | 60 | 19.4 | 2 | 60 | 7.7 | 0 | 53 | <0.001 |
| Qwen3-235B | 18.4 | 0 | 60 | 18.9 | 1 | 60 | 6.4 | 0 | 60 | <0.001 |
| Grok 4.20 | 31.9 | 10 | 60 | 39.4 | 10 | 60 | 36.7 | 10 | 60 | <0.001 |
| $p = 0.9$ | | | | | | | | | | |
| Claude Haiku 4.5 | 21.2 | 10 | 60 | 21.2 | 10 | 60 | 20.2 | 9 | 60 | 0.0054 |
| Gemini 3.5 Flash-Lite | 20.0 | 10 | 60 | 20.0 | 10 | 60 | 20.0 | 10 | 60 | 1.0000 |
| GPT-5.6 Luna | 19.8 | 7 | 60 | 19.4 | 4 | 60 | 20.1 | 9 | 60 | 0.0280 |
| Qwen3-235B | 18.4 | 0 | 60 | 18.6 | 0 | 60 | 10.1 | 0 | 60 | <0.001 |
| Grok 4.20 | 38.2 | 10 | 60 | 40.0 | 10 | 60 | 38.7 | 10 | 60 | <0.001 |

## S7.4 How each model's signature changes

Each model has a signature on the base grid: Qwen3-235B skips the last payment, Gemini 3.5 Flash-Lite and GPT-5.6 Luna play the exact split, and Grok 4.20 often pays 4 in every round. Table S13 follows these signatures across conditions. The averages in Table S11 are stable, but the signatures are not. Qwen3-235B's last-round drop ranges from less than half of its seats to all of them. GPT-5.6 Luna plays the exact split far less often under neutral wording and the stated goal. Grok 4.20's rise with risk appears on the base grid and at temperature zero but not in the rerun of the same prompt, which is why the paper does not read it as a response to risk.

**Table S13:** Model signatures by condition, at $p = 0.1$ and 0.9 pooled (20 games per model). *Exact split*: seats paying 2 in all ten rounds. *Total at 0.9 minus at 0.1*: Grok 4.20's mean seat total, with a permutation *P*. *Distinct games*: game records per ten-game cell that differ in at least one move, averaged over cells.

| | Qwen3 235B | Gemini 3.5 Flash-Lite | GPT-5.6 Luna | Grok 4.20 | | | |
|---|---|---|---|---|---|---|---|
| Condition | Pays 0 in round 10 (%) | Exact split (%) | Exact split (%) | Pays 4 in every round (%) | Total at 0.9 minus at 0.1 | *P* | Distinct games |
| Base grid | 97 | 99 | 75 | 54 | 6.2 | 0.0094 | 6.6 |
| Rerun | 92 | 100 | 84 | 60 | 1.9 | 0.4353 | 6.1 |
| No hint | 41 | 98 | 58 | 99 | 0.1 | 1.0000 | 6.1 |
| Rewording 1 | 57 | 97 | 71 | 26 | 3.9 | 0.1230 | 7.9 |
| Rewording 2 | 100 | 99 | 63 | 0 | 1.0 | 0.2158 | 4.6 |
| Temperature zero | 93 | 99 | 92 | 57 | 6.1 | 0.0116 | 5.5 |
| Neutral wording | 93 | 99 | 33 | 93 | 0.6 | 0.0126 | 6.4 |
| Goal stated | 100 | 73 | 32 | 49 | 2.0 | 0.0032 | 8.2 |

## S7.5 Removing the hint

The no-hint prompt was played at $p = 0.1$, 0.5 and 0.9, the same levels as the rerun, so Table S14 compares the two directly. Without the hint, no model pays significantly less on average, and two pay more. The hint mainly bought precision: fewer seats play the exact split without it, and Grok 4.20 moves almost entirely to paying 4 in every round. Qwen3-235B now pays enough to reach the target, but earns less, because it overpays.

**Table S14:** No hint against the rerun of the base prompt at $p = 0.1$, 0.5 and 0.9 (30 games per model in each). $P$: permutation test permuting condition labels within each risk level; panel rows also permute within each model.

| Measure and model | Rerun | No hint | No hint minus rerun | $P$ |
|---|---|---|---|---|
| Target reached (% of games) | | | | |
| Claude Haiku 4.5 | 100 | 100 | 0 | 1.0000 |
| Gemini 3.5 Flash-Lite | 100 | 100 | 0 | 1.0000 |
| GPT-5.6 Luna | 80 | 57 | −23 | 0.0880 |
| Qwen3-235B | 0 | 100 | 100 | <0.001 |
| Grok 4.20 | 100 | 100 | 0 | 1.0000 |
| *Panel* | 76 | 91 | 15 | <0.001 |
| Expected payoff per seat | | | | |
| Claude Haiku 4.5 | 18.7 | 18.4 | −0.3 | 0.2655 |
| Gemini 3.5 Flash-Lite | 20.0 | 19.9 | −0.1 | 0.1128 |
| GPT-5.6 Luna | 18.3 | 16.8 | −1.5 | 0.2096 |
| Qwen3-235B | 10.8 | 5.2 | −5.7 | <0.001 |
| Grok 4.20 | 4.7 | 0.0 | −4.6 | <0.001 |
| *Panel* | 14.5 | 12.1 | −2.4 | <0.001 |
| Mean seat total | | | | |
| Claude Haiku 4.5 | 21.3 | 21.6 | 0.3 | 0.2589 |
| Gemini 3.5 Flash-Lite | 20.0 | 20.1 | 0.1 | 0.1202 |
| GPT-5.6 Luna | 19.9 | 19.8 | −0.1 | 0.6269 |
| Qwen3-235B | 18.4 | 34.8 | 16.4 | <0.001 |
| Grok 4.20 | 35.3 | 40.0 | 4.6 | <0.001 |
| *Panel* | 23.0 | 27.2 | 4.3 | <0.001 |
| Exact split (% of seats) | | | | |
| Claude Haiku 4.5 | 15 | 1 | −14 | <0.001 |
| Gemini 3.5 Flash-Lite | 100 | 97 | −3 | 0.1210 |
| GPT-5.6 Luna | 83 | 56 | −27 | 0.0054 |
| Qwen3-235B | 6 | 0 | −6 | <0.001 |
| Grok 4.20 | 0 | 0 | 0 | 1.0000 |
| *Panel* | 41 | 31 | −10 | <0.001 |
| Pays 4 in every round (% of seats) | | | | |
| Claude Haiku 4.5 | 0 | 0 | 0 | 1.0000 |
| Gemini 3.5 Flash-Lite | 0 | 0 | 0 | 1.0000 |
| GPT-5.6 Luna | 0 | 0 | 0 | 1.0000 |
| Qwen3-235B | 0 | 16 | 16 | <0.001 |
| Grok 4.20 | 57 | 99 | 42 | <0.001 |
| *Panel* | 11 | 23 | 12 | <0.001 |

# S8 In-game questions

This section supports the value questions of Section 3.3 and the answers-against-play paragraph of Section 4.1.2 of the paper, and gives the full comparison of answers and play (Section S8.6). In the rerun games, seat 1 was asked ten questions in rounds 1, 5 and 10, about the state it had just decided on. Each question is a separate call that sees the decision prompt of that round with its final instruction replaced by the question, and its answer never enters the game. The answer is read from a final `ANSWER:` line; an unreadable answer would count as wrong. Each model answers 90 questions of each kind: 30 games, three rounds each.

## S8.1 Numbers behind the main text

| Statement in the paper, made precise | Value |
|---|---|
| Games with questions; questions asked; rule questions | 150; 4,500; 3,600 |
| Rule questions answered correctly (%) | 100 |
| Value comparison answered correctly (%): Claude Haiku 4.5; Gemini 3.5 Flash-Lite; GPT-5.6 Luna | 100; 99; 84 |
| Qwen3-235B; Grok 4.20 | 33; 46 |
| These three models at $p = 0.1$, rounds 1 and 5: answers saying that paying nothing pays more (%) | 97 |
| The same seats' mean total at $p = 0.1$; at $p = 0.9$ | 20.1; 20.3 |
| change, with interval | 0.3 [−0.5, 1.0] |
| Qwen3-235B: answers saying that paying nothing pays more, all risk levels (%) | 100 |
| Grok 4.20: the same answer in round 1, all risk levels (%) | 100 |
| GPT-5.6 Luna: games whose value answer changes across rounds, of 30 | 13 |

## S8.2 A question as the agent receives it

The prompt seat 1 receives in round 5 at $p = 0.9$ for the question that compares two strategies (the history is an example). It is the decision prompt of that round with the final instruction replaced; the replacement is in bold.

```
You are Player_1, player 1 of 6, in a group experiment that simulates a collective-risk social dilemma.
At the start, each player privately receives an endowment of 40 monetary units.  Your final cash payoff
depends on the rules below.  The game lasts exactly 10 rounds.
Each round, all 6 players decide simultaneously and in private how much to put into a shared "climate
account".  The only amounts allowed are 0, 2, 4.
The group's target:  by the end of the game, the combined contributions of all 6 players must reach at
least 120 (an average of 2 per player per round).
- If the group reaches 120:  every player keeps, in cash, whatever they did not contribute (for example,
contributing 2 every round leaves you 20 at the end).
- If the group does not reach 120:  the computer runs a lottery:  with probability 90% the group is hit
by a disaster and every player loses all of their remaining money (final cash payoff 0); with the
remaining 10% probability, each player keeps their remaining money.
- Money put into the climate account is gone for good and is never refunded, whether or not the target
is reached.
After each round, all 6 contributions for that round are revealed to everyone.  Names are hidden, but
each player always appears in the same fixed position.  Players cannot talk to or message each other.
Current state:
- Round 5 of 10.
- Your remaining money:  30.
Previous rounds (each player keeps the same position; you are position 1):
Round 1:  P1(you)=2, P2=2, P3=2, P4=2, P5=4, P6=2.
Round 2:  P1(you)=2, P2=0, P3=2, P4=2, P5=4, P6=2.
Round 3:  P1(you)=2, P2=2, P3=2, P4=2, P5=2, P6=2.
Round 4:  P1(you)=4, P2=2, P3=2, P4=0, P5=2, P6=2.
Now answer the following question about the game described above.  This is a comprehension check, not a
move:  do not make any contribution decision.
```

```
Compare two strategies for the whole game.  (A) You contribute 2 in every one of the 10 rounds and the
group reaches the target, so you finish with 20 for certain.  (B) You contribute 0 in every round and
the group does not reach the target, so you keep 40 unless the disaster strikes.  Which gives the higher
AVERAGE final cash?  Answer 1 if A is higher, 2 if B is higher, 0 if they are exactly equal.
Answer with a single final line in exactly this format, and write nothing after it:
ANSWER: <your answer>
```

## S8.3 The ten questions

Each question is asked in rounds 1, 5 and 10 of every game, in a separate call. No question text depends on the round or the risk level.

| Label | Question as asked | Answer format | Correct answer | | |
|---|---|---|---|---|---|
| | | | $p = 0.1$ | $p = 0.5$ | $p = 0.9$ |
| *Rule questions* | | | | | |
| Allowed amounts | `Which contribution amounts are you allowed to choose each round?` | A list of numbers | 0, 2, 4 | 0, 2, 4 | 0, 2, 4 |
| Starting endowment | `How many monetary units did each player receive as their starting endowment?` | A number | 40 | 40 | 40 |
| Group target | `What combined total must all players reach by the end of the game to avoid the risk?` | A number | 120 | 120 | 120 |
| Number of rounds | `How many rounds does the game last in total?` | A number | 10 | 10 | 10 |
| Disaster probability | `If the group does NOT reach the target, what is the percent chance (0-100) that everyone loses all of their remaining money?` | A number | 10 | 50 | 90 |
| Payoff after disaster | `If the group fails to reach the target and the disaster does strike, what is your final cash payoff?` | A number | 0 | 0 | 0 |
| Largest payment | `What is the largest amount any single player may put into the climate account in one round?` | A number | 4 | 4 | 4 |
| Smallest payment | `What is the smallest amount any single player may put into the climate account in one round?` | A number | 0 | 0 | 0 |
| *Value questions* | | | | | |
| Expected cash from paying nothing | `Suppose you contribute 0 in every round and the group ends up NOT reaching the target.  On average, over many such games, how much cash would you end up with?  Give a single whole number.` | A number | 36 | 20 | 4 |
| Two strategies compared | `Compare two strategies for the whole game.  (A) You contribute 2 in every one of the 10 rounds and the group reaches the target, so you finish with 20 for certain.  (B) You contribute 0 in every round and the group does not reach the target, so you keep 40 unless the disaster strikes.  Which gives the higher AVERAGE final cash?  Answer 1 if A is higher, 2 if B is higher, 0 if they are exactly equal.` | A code: 1, 2 or 0 | 2 (B) | 0 (equal) | 1 (A) |

## S8.4 Rule questions

Every rule question was answered correctly by every model in every round (Table S15). These questions ask for numbers printed in the prompt, so they show that the agents read the rules, not that they can use them.

**Table S15:** Rule questions: share answered correctly (%), with the number of answers in brackets.

| Question | Claude Haiku 4.5 | Gemini 3.5 Flash-Lite | GPT-5.6 Luna | Qwen3 235B | Grok 4.20 | All models |
|---|---|---|---|---|---|---|
| Allowed amounts | 100 (90) | 100 (90) | 100 (90) | 100 (90) | 100 (90) | 100 (450) |
| Starting endowment | 100 (90) | 100 (90) | 100 (90) | 100 (90) | 100 (90) | 100 (450) |
| Group target | 100 (90) | 100 (90) | 100 (90) | 100 (90) | 100 (90) | 100 (450) |
| Number of rounds | 100 (90) | 100 (90) | 100 (90) | 100 (90) | 100 (90) | 100 (450) |
| Disaster probability | 100 (90) | 100 (90) | 100 (90) | 100 (90) | 100 (90) | 100 (450) |
| Payoff after disaster | 100 (90) | 100 (90) | 100 (90) | 100 (90) | 100 (90) | 100 (450) |
| Largest payment | 100 (90) | 100 (90) | 100 (90) | 100 (90) | 100 (90) | 100 (450) |
| Smallest payment | 100 (90) | 100 (90) | 100 (90) | 100 (90) | 100 (90) | 100 (450) |
| All rule questions | 100 (720) | 100 (720) | 100 (720) | 100 (720) | 100 (720) | 100 (3,600) |

## S8.5 The value comparison

Rule questions only ask for numbers shown in the prompt, so the paper relies on one value question, which compares two strategies for the whole game: (A) contribute 2 in every round and reach the target, finishing with 20 for certain; (B) contribute nothing and miss the target, keeping 40 unless the disaster strikes. The correct answer is B below $p^*$, A above it, and equal at $p^*$. A second value question could be read in two ways, so it is reported only here (Section S8.7).

Table S16 counts every answer to the value comparison. Claude Haiku 4.5 is always correct and Gemini 3.5 Flash-Lite is almost always correct; both always answer "equal" at $p = 0.5$. GPT-5.6 Luna is correct most of the time. Qwen3-235B answers that paying nothing is better in all 90 answers, so it is correct only at $p = 0.1$. Grok 4.20 leans the same way.

**Table S16:** The value comparison of Section 3.3 of the paper: number of answers of each kind, of 30 per model and risk level, and share correct. A: pay 2 in every round and reach the target; B: pay nothing and miss it.

| Model | Risk | Correct answer | Answers saying: A higher | B higher | Equal | Correct (%) |
|---|---|---|---|---|---|---|
| Claude Haiku 4.5 | $p = 0.1$ | B | 0 | 30 | 0 | 100 |
| | $p = 0.5$ | Equal | 0 | 0 | 30 | 100 |
| | $p = 0.9$ | A | 30 | 0 | 0 | 100 |
| | All | | 30 | 30 | 30 | 100 |
| Gemini 3.5 Flash-Lite | $p = 0.1$ | B | 0 | 30 | 0 | 100 |
| | $p = 0.5$ | Equal | 0 | 0 | 30 | 100 |
| | $p = 0.9$ | A | 29 | 1 | 0 | 97 |
| | All | | 29 | 31 | 30 | 99 |
| GPT-5.6 Luna | $p = 0.1$ | B | 2 | 28 | 0 | 93 |
| | $p = 0.5$ | Equal | 0 | 6 | 24 | 80 |
| | $p = 0.9$ | A | 24 | 2 | 4 | 80 |
| | All | | 26 | 36 | 28 | 84 |
| Qwen3-235B | $p = 0.1$ | B | 0 | 30 | 0 | 100 |
| | $p = 0.5$ | Equal | 0 | 30 | 0 | 0 |
| | $p = 0.9$ | A | 0 | 30 | 0 | 0 |
| | All | | 0 | 90 | 0 | 33 |
| Grok 4.20 | $p = 0.1$ | B | 0 | 23 | 7 | 77 |
| | $p = 0.5$ | Equal | 3 | 15 | 12 | 40 |
| | $p = 0.9$ | A | 6 | 20 | 4 | 20 |
| | All | | 9 | 58 | 23 | 46 |

Figure S2 and Table S17 split the same answers by round. Claude Haiku 4.5, Gemini 3.5 Flash-Lite and Qwen3-235B give the same answer in all three rounds of almost every game. GPT-5.6 Luna drifts in round 10 at $p = 0.5$. Grok 4.20 changes the most: it says that paying nothing is better in round 1 at every risk level, and moves toward the correct answer later in the game.

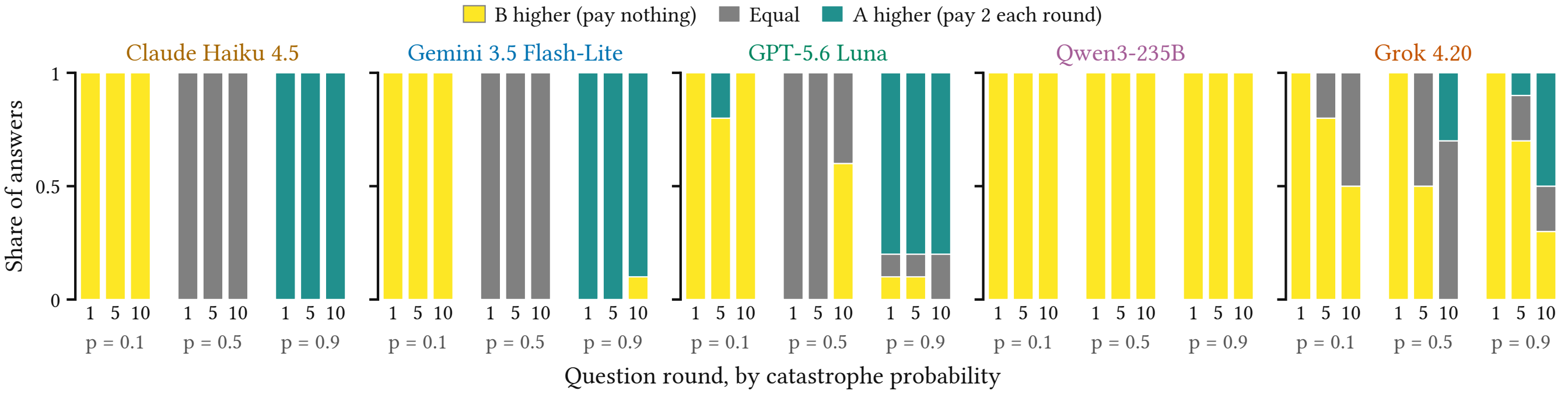


**Figure S2:** Answers to the value comparison by model, risk level and round (ten answers per bar).

**Table S17:** Value comparison by round: share of answers saying that B (paying nothing) is better, ten answers per entry; and the number of games whose answer is not the same in rounds 1, 5 and 10.

| | Answers saying B (paying nothing) is higher (%) | | | | | | | | | Games whose answer changes | | | |
|---|---|---|---|---|---|---|---|---|---|---|---|---|---|
| | $p = 0.1$ | | | $p = 0.5$ | | | $p = 0.9$ | | | By $p$ (of 10) | | | All |
| Model | R1 | R5 | R10 | R1 | R5 | R10 | R1 | R5 | R10 | 0.1 | 0.5 | 0.9 | (of 30) |
| Claude Haiku 4.5 | 100 | 100 | 100 | 0 | 0 | 0 | 0 | 0 | 0 | 0 | 0 | 0 | 0 |
| Gemini 3.5 Flash-Lite | 100 | 100 | 100 | 0 | 0 | 0 | 0 | 0 | 10 | 0 | 0 | 1 | 1 |
| GPT-5.6 Luna | 100 | 80 | 100 | 0 | 0 | 60 | 10 | 10 | 0 | 2 | 6 | 5 | 13 |
| Qwen3-235B | 100 | 100 | 100 | 100 | 100 | 100 | 100 | 100 | 100 | 0 | 0 | 0 | 0 |
| Grok 4.20 | 100 | 80 | 50 | 100 | 50 | 0 | 100 | 70 | 30 | 6 | 10 | 8 | 24 |

## S8.6 Answers against play

Every rule question is answered correctly. Claude Haiku 4.5 and Gemini 3.5 Flash-Lite answer the value question correctly almost every time, and GPT-5.6 Luna most of the time. All three change their answer with the risk (Figure S3a): at low risk they almost always say that keeping contributes more, and at high risk they mostly say the opposite. Their contributions barely change (Figure S3b). Qwen3-235B always says that keeping contributes more, and Grok 4.20 gives that answer in the first round, regardless of the risk. The question supplies two hypothetical outcomes rather than the partners' responses to a deviation, and it makes the comparison explicit, unlike the move request. Correct answers therefore do not establish knowledge of a best response, and the result is a gap between answers and play, not evidence about what an agent computes. A strongly risk-averse agent could give these answers and still contribute; contribution at $p = 0$ and after settlement provides the sharper test.

Table S18 gives the numbers behind Figure S3: for each model and risk level, the share of value answers saying that paying nothing is better and the total that the questioned seat pays in the same games, each with an interval. The figure plots the $p = 0.1$ and $p = 0.9$ rows with these intervals and joins them by an arrow; the *Shift* rows give that change, with intervals that resample the answer and the payment from the same game. For the three models that answer well, the answer moves by close to its full range while the payment does not move.

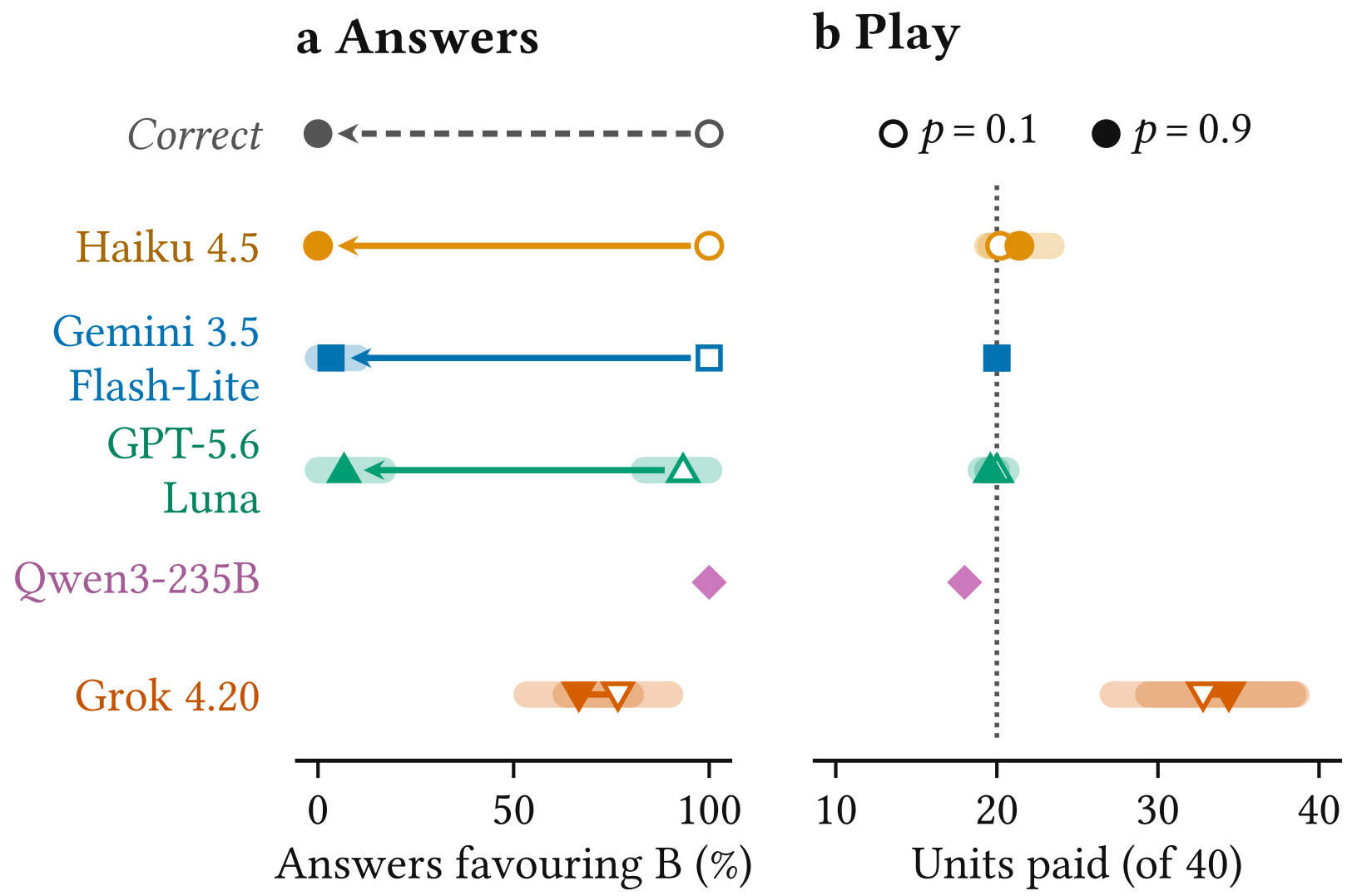


**Figure S3:** Answers and play at risk 0.1 (open) and 0.9 (filled), ten games per model and risk; arrows run from low to high risk, and shaded bars are 95% game-bootstrap intervals. (a) Share of answers favouring B; dashed: a fully correct switch. (b) Units contributed by the questioned seat out of 40; dotted: fair share. This compares answers and actions, not best responses; Table S18 gives the shifts with their intervals and all three risks.

**Table S18:** Answers and play of the questioned seat (ten games per risk level). *Shift*: drop in the share saying B from $p = 0.1$ to $p = 0.9$, and change in the units paid, with intervals that resample the games at the two risk levels independently.

| | | Share saying B (paying nothing) is higher | | Units seat 1 pays | |
|---|---|---|---|---|---|
| Model | Risk | Mean | 95% CI | Mean | 95% CI |
| Claude Haiku 4.5 | $p = 0.1$ | 1.00 | [1.00, 1.00] | 20.2 | [19.4, 21.0] |
| | $p = 0.5$ | 0.00 | [0.00, 0.00] | 21.8 | [20.8, 23.0] |
| | $p = 0.9$ | 0.00 | [0.00, 0.00] | 21.4 | [19.6, 23.4] |
| | Shift | 1.00 | [1.00, 1.00] | 1.2 | [−0.6, 3.4] |
| Gemini 3.5 Flash-Lite | $p = 0.1$ | 1.00 | [1.00, 1.00] | 20.0 | [20.0, 20.0] |
| | $p = 0.5$ | 0.00 | [0.00, 0.00] | 20.0 | [20.0, 20.0] |
| | $p = 0.9$ | 0.03 | [0.00, 0.10] | 20.0 | [20.0, 20.0] |
| | Shift | 0.97 | [0.90, 1.00] | 0.0 | [0.0, 0.0] |
| GPT-5.6 Luna | $p = 0.1$ | 0.93 | [0.83, 1.00] | 20.0 | [19.4, 20.6] |
| | $p = 0.5$ | 0.20 | [0.10, 0.30] | 19.8 | [19.4, 20.0] |
| | $p = 0.9$ | 0.07 | [0.00, 0.17] | 19.6 | [19.0, 20.0] |
| | Shift | 0.87 | [0.73, 0.97] | −0.4 | [−1.2, 0.4] |
| Qwen3-235B | $p = 0.1$ | 1.00 | [1.00, 1.00] | 18.0 | [18.0, 18.0] |
| | $p = 0.5$ | 1.00 | [1.00, 1.00] | 18.0 | [18.0, 18.0] |
| | $p = 0.9$ | 1.00 | [1.00, 1.00] | 18.0 | [18.0, 18.0] |
| | Shift | 0.00 | [0.00, 0.00] | 0.0 | [0.0, 0.0] |
| Grok 4.20 | $p = 0.1$ | 0.77 | [0.63, 0.90] | 32.8 | [27.2, 38.4] |
| | $p = 0.5$ | 0.50 | [0.40, 0.60] | 29.6 | [24.2, 35.2] |
| | $p = 0.9$ | 0.67 | [0.53, 0.80] | 34.4 | [29.4, 38.6] |
| | Shift | 0.10 | [−0.10, 0.27] | 1.6 | [−5.8, 8.8] |

## S8.7 The second value question

The second value question asks how much cash a player ends with, on average, if it pays nothing and the group misses the target. The intended answer is $(1 - p)\,40$ for the whole game. The question can also be read as starting from the moment it is asked, with the money the seat still holds, which gives $(1 - p)$ times that money. Table S19 shows that this second reading is the most common error: Claude Haiku 4.5 gives it in most of its round-5 answers, when it has already paid part of its endowment. Answers of 40, the cash kept if no disaster strikes, come only from Qwen3-235B at $p = 0.1$, where they are almost all of its answers. Answers of 0, the cash left if the disaster strikes, come mostly from GPT-5.6 Luna at $p = 0.5$ and 0.9 and from Grok 4.20 at $p = 0.9$, and Grok 4.20 gives other wrong answers to about a quarter of the questions at $p = 0.1$. Because a wrong answer can come from a reasonable reading rather than from a wrong computation, the paper does not use this question.

**Table S19:** The second value question, by model and risk level (30 answers per entry). *Correct (%)*: overall and in rounds 1, 5 and 10. *Wrong answers*, as shares of all answers: 40; 0; *Cash left*, which is $(1 - p)$ times the seat's money at the time of asking, rounded; and any other number.

| | | Correct | Correct (%) | | | | Wrong answers (% of all answers) | | | |
|---|---|---|---|---|---|---|---|---|---|---|
| Model | Risk | answer | All | R1 | R5 | R10 | 40 | 0 | Cash left | Other |
| Claude Haiku 4.5 | $p = 0.1$ | 36 | 77 | 100 | 30 | 100 | 0 | 0 | 23 | 0 |
| | $p = 0.5$ | 20 | 73 | 100 | 20 | 100 | 0 | 0 | 27 | 0 |
| | $p = 0.9$ | 4 | 67 | 100 | 0 | 100 | 0 | 0 | 33 | 0 |
| Gemini 3.5 Flash-Lite | $p = 0.1$ | 36 | 100 | 100 | 100 | 100 | 0 | 0 | 0 | 0 |
| | $p = 0.5$ | 20 | 100 | 100 | 100 | 100 | 0 | 0 | 0 | 0 |
| | $p = 0.9$ | 4 | 100 | 100 | 100 | 100 | 0 | 0 | 0 | 0 |
| GPT-5.6 Luna | $p = 0.1$ | 36 | 93 | 100 | 80 | 100 | 0 | 0 | 7 | 0 |
| | $p = 0.5$ | 20 | 77 | 50 | 80 | 100 | 0 | 17 | 7 | 0 |
| | $p = 0.9$ | 4 | 80 | 70 | 70 | 100 | 0 | 20 | 0 | 0 |
| Qwen3-235B | $p = 0.1$ | 36 | 3 | 0 | 10 | 0 | 93 | 0 | 0 | 3 |
| | $p = 0.5$ | 20 | 100 | 100 | 100 | 100 | 0 | 0 | 0 | 0 |
| | $p = 0.9$ | 4 | 97 | 100 | 100 | 90 | 0 | 3 | 0 | 0 |
| Grok 4.20 | $p = 0.1$ | 36 | 77 | 100 | 30 | 100 | 0 | 0 | 0 | 23 |
| | $p = 0.5$ | 20 | 97 | 100 | 90 | 100 | 0 | 0 | 3 | 0 |
| | $p = 0.9$ | 4 | 90 | 100 | 70 | 100 | 0 | 10 | 0 | 0 |

# S9 One model beside scripted partners

This section supports Section 4.2.1 of the paper, including its Table 5; Figure S4 shows the same games round by round. Seat 1 holds the model and the five other seats follow one of the four rules of Table S2, at five risk levels with ten games each, so every model plays 50 games beside each kind of partner. Best responses are derived in Section S2.7.

## S9.1 Numbers behind the main text

| Statement in the paper, made precise | Value |
|---|---|
| Games beside partners who always pay 0 or always pay 4 | 500 |
| of which the model pays nothing, which is the best response | 0 |
| smallest total the model pays in those games | 2 |
| Scripted games of any kind in which the model pays nothing | 0 of 1,000 |
| Scripted games that are exactly optimal (expected payoff equal to the best response) | 186 of 1,000 |
| total paid in every one of them | 20 |
| Share of the variation in the total explained by the partner kind, for Claude Haiku 4.5, Gemini 3.5 Flash-Lite, GPT-5.6 Luna and Qwen3-235B (%) | 69 to 92 |
| explained by risk and its interaction with the partner kind, largest of the four (%) | 5 |
| explained by risk and interaction for an exact best responder (%) | 50 |
| Largest change in a model's mean total between low and high risk beside fair-share partners | 2.0 |
| Grok 4.20: share explained by the partner kind; by risk and interaction (%) | 3; 8 |
| Grok 4.20: games opening with a payment of 4, of 200 | 103 |
| Grok 4.20: share of the variation explained by its opening move (%) | 74 |
| Units paid after the outcome is settled, per game: GPT-5.6 Luna; Qwen3-235B; Grok 4.20 | 0.0; 10.1; 13.7 |
| Qwen3-235B's mean total beside partners who always pay 0; always pay 2 | 29.9; 18.1 |
| Qwen3-235B beside partners who always pay 2: games missing the target, of 50; shortfall | 47; 2 |
| Gemini 3.5 Flash-Lite in the same seat: games reaching the target, of 50 | 50 |

## S9.2 What the models pay

Figure S4 shows, round by round, what each model pays beside partners who always pay 0 or always pay 4, the two kinds for which paying nothing is the best response. Look for the round at which each row turns light: Gemini 3.5 Flash-Lite and Claude Haiku 4.5 stop once the outcome is settled, GPT-5.6 Luna stops earlier, and Qwen3-235B and Grok 4.20 keep paying.

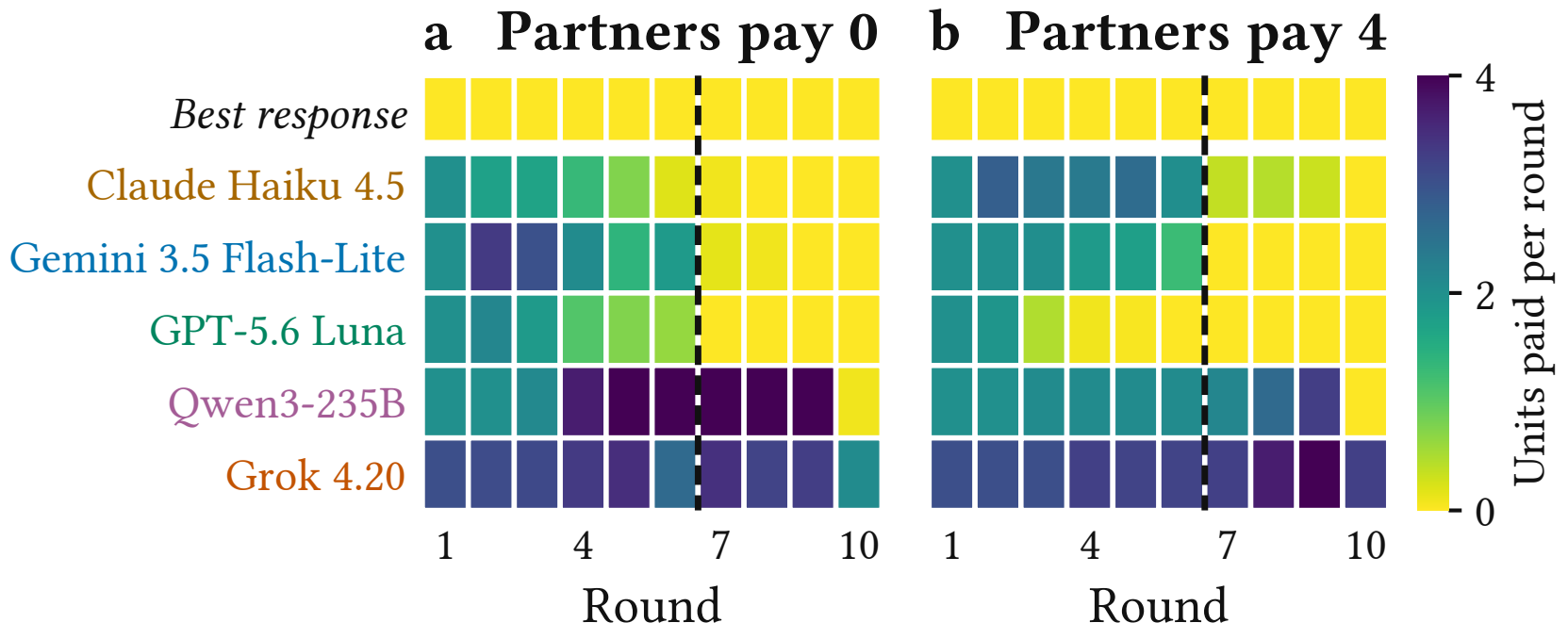


**Figure S4:** One model beside five scripted partners, pooled over five risk levels (50 games per row). A cell is the mean payment in one round; the top row is the best response, 0 in every round. The dashed rule follows round 6, after which the outcome is settled in most games.

Table S20 gives the mean total of the model's seat for every kind of partner and risk level, with the best response below each block. The last column is what a best responder would change: 20 units between low and high risk beside fair-share partners, and nothing beside the other two. No model comes close. Beside partners who always pay 0 or always pay 4, every model pays more than the best response of 0 in every game, and Qwen3-235B and Grok 4.20 pay the most, even though these partners make every unit a loss.

Figure S5 shows the same games round by round beside the two kinds of fair-share partner, where the best response switches at $p^*$. Gemini 3.5 Flash-Lite and GPT-5.6 Luna pay almost exactly 2 in every round at every risk level. Qwen3-235B pays 2 for nine rounds and 0 in the last, at every risk level, and Grok 4.20 pays more than 2 in every round. None of the rows changes between the low-risk and high-risk panels.

**Table S20:** Mean total paid by the model's seat (ten games per entry). *All risks*: the mean over the five risk levels. Last column: the mean at $p = 0.7$ and 0.9 minus the mean at $p = 0.1$ and 0.3. The best-response row gives the best total; "0 or 20" marks the tie at $p^*$.

| | Risk $p$ | | | | | All | $p > 0.5$ minus |
|---|---|---|---|---|---|---|---|
| Model | 0.1 | 0.3 | 0.5 | 0.7 | 0.9 | risks | $p < 0.5$ |
| *Always 0* | | | | | | | |
| Claude Haiku 4.5 | 6.8 | 7.2 | 8.2 | 8.6 | 8.0 | 7.8 | 1.3 |
| Gemini 3.5 Flash-Lite | 12.6 | 11.8 | 14.2 | 14.2 | 16.6 | 13.9 | 3.2 |
| GPT-5.6 Luna | 6.8 | 6.0 | 9.0 | 8.6 | 12.0 | 8.5 | 3.9 |
| Qwen3-235B | 30.2 | 30.0 | 30.4 | 30.0 | 28.8 | 29.9 | −0.7 |
| Grok 4.20 | 27.4 | 28.0 | 32.0 | 33.4 | 32.6 | 30.7 | 5.3 |
| Best response | 0 | 0 | 0 | 0 | 0 | 0 | 0 |
| *Always 2* | | | | | | | |
| Claude Haiku 4.5 | 20.4 | 21.8 | 20.0 | 21.6 | 21.8 | 21.1 | 0.6 |
| Gemini 3.5 Flash-Lite | 20.0 | 20.0 | 20.0 | 20.0 | 20.0 | 20.0 | 0.0 |
| GPT-5.6 Luna | 20.0 | 20.0 | 20.0 | 20.0 | 20.2 | 20.0 | 0.1 |
| Qwen3-235B | 18.0 | 18.0 | 18.6 | 18.0 | 18.0 | 18.1 | 0.0 |
| Grok 4.20 | 30.8 | 32.0 | 28.0 | 28.8 | 30.0 | 29.9 | −2.0 |
| Best response | 0 | 0 | 0 or 20 | 20 | 20 | 8 or 12 | 20 |
| *Always 4* | | | | | | | |
| Claude Haiku 4.5 | 15.8 | 15.0 | 17.0 | 15.0 | 13.8 | 15.3 | −1.0 |
| Gemini 3.5 Flash-Lite | 11.4 | 11.6 | 9.8 | 11.2 | 10.6 | 10.9 | −0.6 |
| GPT-5.6 Luna | 4.0 | 4.2 | 4.8 | 5.0 | 4.6 | 4.5 | 0.7 |
| Qwen3-235B | 21.8 | 20.8 | 19.4 | 20.8 | 19.2 | 20.4 | −1.3 |
| Grok 4.20 | 33.0 | 34.2 | 31.2 | 37.0 | 29.2 | 32.9 | −0.5 |
| Best response | 0 | 0 | 0 | 0 | 0 | 0 | 0 |
| *Copy the group* | | | | | | | |
| Claude Haiku 4.5 | 21.2 | 21.0 | 20.0 | 22.4 | 21.0 | 21.1 | 0.6 |
| Gemini 3.5 Flash-Lite | 20.0 | 20.0 | 20.0 | 20.0 | 20.0 | 20.0 | 0.0 |
| GPT-5.6 Luna | 20.0 | 20.2 | 19.8 | 19.8 | 20.0 | 20.0 | −0.2 |
| Qwen3-235B | 18.0 | 18.0 | 18.6 | 18.0 | 18.0 | 18.1 | 0.0 |
| Grok 4.20 | 30.0 | 28.0 | 30.0 | 30.0 | 32.0 | 30.0 | 2.0 |
| Best response | 0 | 0 | 0 or 20 | 20 | 20 | 8 or 12 | 20 |

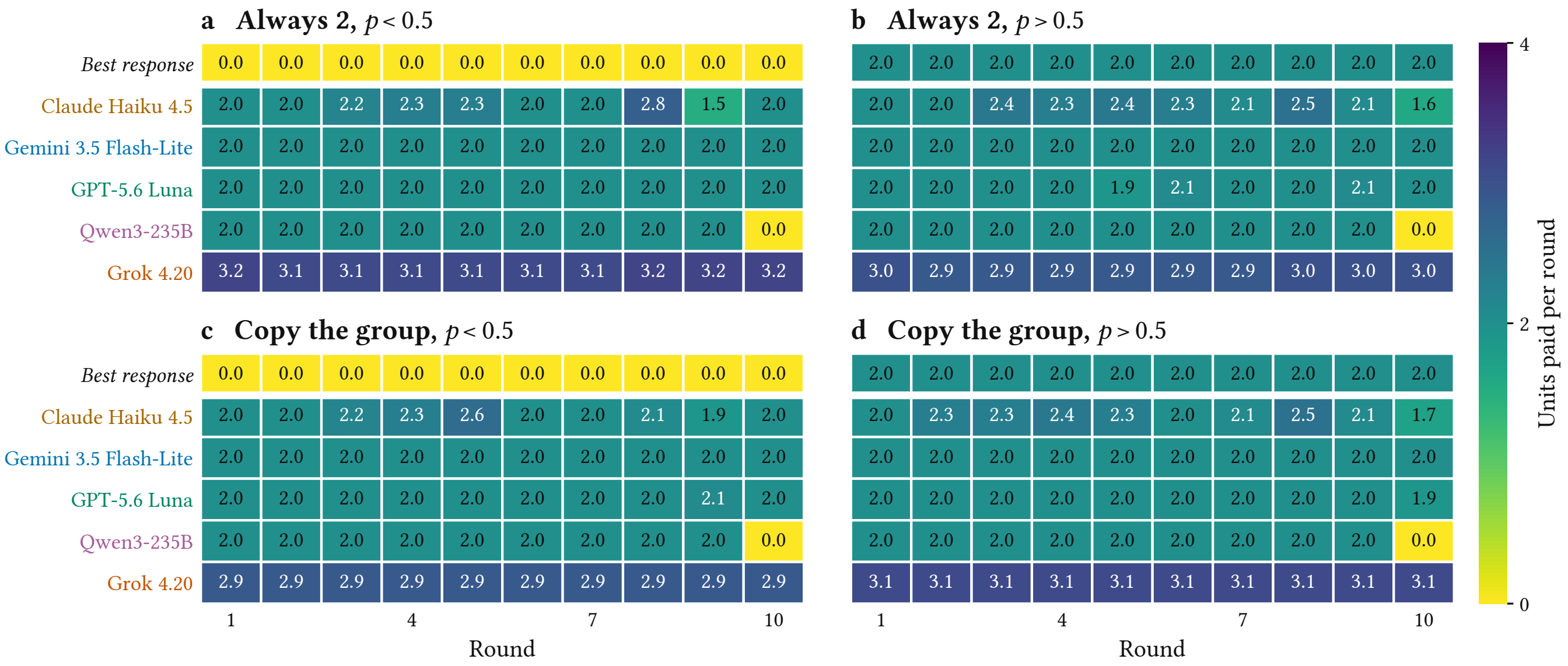


**Figure S5:** Mean payment in each round beside partners who always pay 2 (top) and partners who copy the group (bottom), at low risk ($p = 0.1$ and 0.3, left) and high risk ($p = 0.7$ and 0.9, right); 20 games per row. The top row of each panel is the best response. Above $p^*$ any sequence that adds up to 20 is a best response, and the row shows their round-by-round average, which is 2.

## S9.3 Outcomes and distance from the best response

Table S21 turns payments into outcomes. Beside partners who always pay 4, every model reaches the target in every game, so the whole gap to the best response is money paid for nothing; it ranges from small for GPT-5.6 Luna to most of the endowment for Grok 4.20. Beside fair-share partners, Claude Haiku 4.5, Gemini 3.5 Flash-Lite and GPT-5.6 Luna reach the target in almost every game and are exactly optimal in many games at or above the pivot, while Qwen3-235B almost never reaches the target.

**Table S21:** Outcomes by kind of partner, pooled over five risk levels (50 games per model). *Expected payoff* is the model seat's; *Gap to best response* is the best-response payoff minus it; *Exactly optimal* counts games with no gap. For the best-response rows beside fair-share partners, the target count is 20 or 30 depending on how the tie at $p^*$ is broken.

| Model | Reached (of 50) | Expected payoff | Gap to best response Mean | 95% CI | Exactly optimal (of 50) |
|---|---|---|---|---|---|
| *Always 0* | | | | | |
| Claude Haiku 4.5 | 0 | 16.3 | 3.7 | [3.1, 4.3] | 0 |
| Gemini 3.5 Flash-Lite | 0 | 13.5 | 6.5 | [5.5, 7.6] | 0 |
| GPT-5.6 Luna | 0 | 16.3 | 3.7 | [3.1, 4.3] | 0 |
| Qwen3-235B | 0 | 4.9 | 15.1 | [12.6, 17.5] | 0 |
| Grok 4.20 | 0 | 5.3 | 14.7 | [12.5, 17.0] | 0 |
| Best response | 0 | 20.0 | 0.0 | | 50 |
| *Always 2* | | | | | |
| Claude Haiku 4.5 | 48 | 18.4 | 6.4 | [4.6, 8.1] | 17 |
| Gemini 3.5 Flash-Lite | 50 | 20.0 | 4.8 | [3.0, 6.6] | 30 |
| GPT-5.6 Luna | 49 | 19.7 | 5.1 | [3.4, 6.9] | 27 |
| Qwen3-235B | 3 | 11.5 | 13.3 | [12.0, 14.4] | 3 |
| Grok 4.20 | 50 | 10.1 | 14.7 | [11.4, 18.1] | 16 |
| Best response | 20 or 30 | 24.8 | 0.0 | | 50 |
| *Always 4* | | | | | |
| Claude Haiku 4.5 | 50 | 24.7 | 15.3 | [13.8, 16.9] | 0 |
| Gemini 3.5 Flash-Lite | 50 | 29.1 | 10.9 | [10.4, 11.4] | 0 |
| GPT-5.6 Luna | 50 | 35.5 | 4.5 | [4.2, 4.8] | 0 |
| Qwen3-235B | 50 | 19.6 | 20.4 | [19.8, 21.1] | 0 |
| Grok 4.20 | 50 | 7.1 | 32.9 | [30.7, 35.2] | 0 |
| Best response | 50 | 40.0 | 0.0 | | 50 |
| *Copy the group* | | | | | |
| Claude Haiku 4.5 | 47 | 18.4 | 6.4 | [4.6, 8.3] | 18 |
| Gemini 3.5 Flash-Lite | 50 | 20.0 | 4.8 | [3.0, 6.6] | 30 |
| GPT-5.6 Luna | 48 | 19.5 | 5.3 | [3.5, 7.1] | 28 |
| Qwen3-235B | 3 | 11.5 | 13.3 | [11.9, 14.4] | 3 |
| Grok 4.20 | 50 | 10.0 | 14.8 | [11.7, 18.0] | 14 |
| Best response | 20 or 30 | 24.8 | 0.0 | | 50 |

Beside partners who always pay 0, the target cannot be reached, and every unit a model pays is a unit that a partner paying nothing keeps. Table S22 shows this cost at each risk level. It shrinks as the risk rises only because the lottery is more likely to destroy what the partner kept.

**Table S22:** Beside partners who always pay 0: expected payoff of one of these partners minus the model seat's expected payoff (ten games per risk level, 50 over all risks). Here it equals the gap to the best response in Table S21.

| | Risk $p$ | | | | | All risks | |
|---|---|---|---|---|---|---|---|
| Model | 0.1 | 0.3 | 0.5 | 0.7 | 0.9 | Mean | 95% CI |
| Claude Haiku 4.5 | 6.1 | 5.0 | 4.1 | 2.6 | 0.8 | 3.7 | [3.1, 4.3] |
| Gemini 3.5 Flash-Lite | 11.3 | 8.3 | 7.1 | 4.3 | 1.7 | 6.5 | [5.5, 7.6] |
| GPT-5.6 Luna | 6.1 | 4.2 | 4.5 | 2.6 | 1.2 | 3.7 | [3.1, 4.3] |
| Qwen3-235B | 27.2 | 21.0 | 15.2 | 9.0 | 2.9 | 15.1 | [12.6, 17.5] |
| Grok 4.20 | 24.7 | 19.6 | 16.0 | 10.0 | 3.3 | 14.7 | [12.5, 17.0] |

## S9.4 Partner versus risk

Table S23 splits the variation in what a model pays beside the three fixed-amount partners. For an exact best responder in the same games, risk and its interaction with the partner explain half of the variation; this uses a payment of 20 at the tie $p = p^*$, and breaking the tie toward 0 would raise the share further. For four models the partner kind explains most of the variation and risk almost none. Grok 4.20 is different: most of its variation is residual, which Table S24 traces to its first move.

**Table S23:** Share of the variation in the model seat's total explained by the partner kind, the risk and their interaction, over partners who always pay 0, 2 or 4 (150 games per model; balanced two-way sums of squares). *Risk and interaction* is computed before rounding, so it can differ by 1 from the sum of the two rounded columns.

| | Share of the variation in the total (%) | | | | |
|---|---|---|---|---|---|
| Model | Partner kind | Risk | Partner kind × risk | Risk and interaction | Residual |
| Claude Haiku 4.5 | 69 | 0 | 1 | 2 | 29 |
| Gemini 3.5 Flash-Lite | 69 | 1 | 4 | 5 | 26 |
| GPT-5.6 Luna | 92 | 1 | 2 | 3 | 5 |
| Qwen3-235B | 92 | 1 | 1 | 2 | 6 |
| Grok 4.20 | 3 | 2 | 6 | 8 | 90 |
| Best response | 50 | 17 | 33 | 50 | 0 |

In round 1 there is no history, so every kind of partner produces the same prompt, and the first move is sampled before the model can learn anything about its partners. Four models open with 2 in all 800 of their games. Grok 4.20 opens with 4 in about half of its games and then mostly repeats its opening, repeating a 4 in 9.4 and a 2 in 7.9 of the ten rounds on average (Table S24).

**Table S24:** Top: Grok 4.20's opening move and its mean ten-round total after each opening, by kind of partner (50 games each). No model ever opens with 0. Bottom: share of the variation in each model's total explained on its own by the opening move, by the partner kind (all four kinds) and by the risk (one-way sums of squares over 200 games). Because this uses all four kinds and one-way shares, the partner share can differ slightly from Table S23.

| | Games opening with | | | Mean total after opening with | | |
|---|---|---|---|---|---|---|
| Grok 4.20 beside | 0 | 2 | 4 | 0 | 2 | 4 |
| Always 0 | 0 | 24 | 26 | | 28.1 | 33.1 |
| Always 2 | 0 | 24 | 26 | | 20.0 | 39.1 |
| Always 4 | 0 | 24 | 26 | | 25.3 | 40.0 |
| Copy the group | 0 | 25 | 25 | | 20.0 | 40.0 |
| All four kinds | 0 | 97 | 103 | | 23.3 | 38.0 |

| | Share of the variation in the total explained by (%) | | |
|---|---|---|---|
| Model | Opening move | Partner kind | Risk |
| Claude Haiku 4.5 | 0 | 72 | 0 |
| Gemini 3.5 Flash-Lite | 0 | 76 | 1 |
| GPT-5.6 Luna | 0 | 94 | 1 |
| Qwen3-235B | 0 | 93 | 0 |
| Grok 4.20 | 74 | 2 | 1 |

## S9.5 Paying after the outcome is settled

Table S25 separates what a model pays before the history settles the outcome from what it pays afterwards, beside partners who always pay 0 or always pay 4. Payment before settlement can reflect a model that does not yet know its partners. Payment afterwards cannot help under any belief about the partners. Gemini 3.5 Flash-Lite, Claude Haiku 4.5 and GPT-5.6 Luna stop almost completely once the outcome is settled, and GPT-5.6 Luna never pays after settlement. Qwen3-235B and Grok 4.20 keep paying in almost every game.

**Table S25:** Top: after which round the outcome is settled, over all 250 games with each kind of partner. Bottom: units the model seat pays per game before and after settlement (50 games per entry), the interval for the units paid after, and the number of games with no payment after settlement.

| | Always 0 | Always 4 |
|---|---|---|
| *Games settled after each round, all five models (of 250)* | | |
| Round 5 | 0 | 26 |
| Round 6 | 236 | 224 |
| Round 7 | 14 | 0 |

| | Always 0 | | | | Always 4 | | | |
|---|---|---|---|---|---|---|---|---|
| | Units paid | | After: | No waste | Units paid | | After: | No waste |
| Model | Before | After | 95% CI | (of 50) | Before | After | 95% CI | (of 50) |
| Claude Haiku 4.5 | 7.7 | 0.1 | [0.0, 0.2] | 49 | 14.2 | 1.1 | [0.7, 1.6] | 33 |
| Gemini 3.5 Flash-Lite | 13.6 | 0.2 | [0.0, 0.6] | 47 | 10.9 | 0.0 | [0.0, 0.0] | 50 |
| GPT-5.6 Luna | 8.5 | 0.0 | [0.0, 0.0] | 50 | 4.5 | 0.0 | [0.0, 0.0] | 50 |
| Qwen3-235B | 17.8 | 12.1 | [12.0, 12.2] | 0 | 12.3 | 8.1 | [7.7, 8.5] | 0 |
| Grok 4.20 | 19.6 | 11.1 | [9.8, 12.4] | 1 | 16.7 | 16.2 | [15.0, 17.4] | 0 |

## S9.6 Fair-share partners

Table S26 shows the most common sequence of moves beside the two kinds of fair-share partner. Gemini 3.5 Flash-Lite plays the exact split in every game. Qwen3-235B's most common sequence pays 2 for nine rounds and 0 in the last, which leaves the table exactly one payment of 2 short of the target, and this sequence accounts for almost all of its games. Grok 4.20 splits between the exact split and paying 4 in every round.

**Table S26:** Beside partners who always pay 2 (top) or copy the group (bottom), pooled over five risk levels (50 games per model): games reaching the target, mean total, the most common sequence of the model's own moves and the number of games that play it, and, among games that miss the target, the most common shortfall and how many games have it. Empty cells: no game missed.

| | Reached | Mean | Most common sequence | | Most common shortfall | |
|---|---|---|---|---|---|---|
| Model | (of 50) | total | Own moves, rounds 1 to 10 | Games | Units | Games |
| *Always 2* | | | | | | |
| Claude Haiku 4.5 | 48 | 21.1 | 2 2 2 2 2 2 2 2 2 2 | 22 | 2 | 2 |
| Gemini 3.5 Flash-Lite | 50 | 20.0 | 2 2 2 2 2 2 2 2 2 2 | 50 | | |
| GPT-5.6 Luna | 49 | 20.0 | 2 2 2 2 2 2 2 2 2 2 | 46 | 2 | 1 |
| Qwen3-235B | 3 | 18.1 | 2 2 2 2 2 2 2 2 2 0 | 47 | 2 | 47 |
| Grok 4.20 | 50 | 29.9 | 2 2 2 2 2 2 2 2 2 2 or 4 4 4 4 4 4 4 4 4 4 | 24 each | | |
| *Copy the group* | | | | | | |
| Claude Haiku 4.5 | 47 | 21.1 | 2 2 2 2 2 2 2 2 2 2 | 21 | 2 | 3 |
| Gemini 3.5 Flash-Lite | 50 | 20.0 | 2 2 2 2 2 2 2 2 2 2 | 50 | | |
| GPT-5.6 Luna | 48 | 20.0 | 2 2 2 2 2 2 2 2 2 2 | 46 | 2 | 2 |
| Qwen3-235B | 3 | 18.1 | 2 2 2 2 2 2 2 2 2 0 | 47 | 2 | 47 |
| Grok 4.20 | 50 | 30.0 | 2 2 2 2 2 2 2 2 2 2 or 4 4 4 4 4 4 4 4 4 4 | 25 each | | |

# S10 Two models at one table

This section supports Section 4.2.2 of the paper and its Figures 3 and 4. Every pair of models shares a table, with one to five seats for the first-named model, at $p = 0.1$, 0.5 and 0.9, with ten games per cell: 150 games per pair and 1,500 in all. "Alone" means self-play of the same model at the same three risk levels, pooling the base grid and the rerun, which use the same prompt (20 games per model and risk level).

## S10.1 Numbers behind the main text

| Statement in the paper, made precise | Value |
|---|---|
| Largest payoff gap per seat between two models at one table; it is Grok 4.20 against GPT-5.6 Luna | 14.5 |
| Pair with the lowest target rate; games reaching the target (%) | Gemini 3.5 Flash-Lite and Qwen3-235B; 9 |
| Round 10 with the pool at 108: decisions paying 0 (%), Qwen3-235B; Gemini 3.5 Flash-Lite | 92; 0.1 |
| Target reached, of 30: five Gemini 3.5 Flash-Lite seats and one Qwen3-235B seat | 11 |
| Gemini 3.5 Flash-Lite alone (%) | 100 |
| Target reached, of 30: five GPT-5.6 Luna seats and one Qwen3-235B seat | 9 |
| GPT-5.6 Luna alone (%) | 78 |
| Target reached, of 30: five Claude Haiku 4.5 seats and one Qwen3-235B seat; the same with Grok 4.20 | 30; 30 |
| Missed tables containing Qwen3-235B; of which reaching 120 after replacing all its round-10 zeros by 2, other actions fixed | 323; 298 (92%) |
| Target reached, of 30: five Qwen3-235B seats and one Grok 4.20 seat; and one Claude Haiku 4.5 seat | 25; 7 |
| GPT-5.6 Luna: units per seat beside five Grok 4.20 seats; beside five seats of another model | 10.5; 19.1 |
| Units per seat beside five Grok 4.20 seats: Gemini 3.5 Flash-Lite; Claude Haiku 4.5 | 14.6; 16.1 |
| Qwen3-235B: units per seat in all tables with Grok 4.20; in all tables with other models | 20.8; 19.2 |
| Share of the shortfall below 2 per round that comes after the pool reaches 120, beside Grok 4.20 (%): Claude Haiku 4.5; Gemini 3.5 Flash-Lite; GPT-5.6 Luna; Qwen3-235B | 91; 88; 32; 72 |

## S10.2 Who earns more at a shared table

Table S27 gives the numbers behind the upper triangle of Figure 3 of the paper. By Proposition 2, the seat that pays less earns more at the same table, so the payoff gap mirrors the payment gap. Every model earns clearly more than Grok 4.20, and no interval of those gaps comes near zero. Among the other four models the gaps are small; Qwen3-235B is slightly ahead of all three, because it pays slightly less.

**Table S27:** Two models at one table (150 games per pair, pooled over seat splits and risk levels). Top: expected payoff per seat of the row model minus that of the column model in the same games, with its interval below. Bottom: units paid per seat, row minus column.

| Row model | Claude Haiku 4.5 | Gemini 3.5 Flash-Lite | GPT-5.6 Luna | Qwen3 235B | Grok 4.20 |
|---|---|---|---|---|---|
| *Expected payoff per seat, row model minus column model* | | | | | |
| Claude Haiku 4.5 | *self* | −1.9 | −2.7 | −0.6 | +9.7 |
| | | [−2.2, −1.5] | [−3.1, −2.3] | [−1.1, −0.1] | [8.4, 10.9] |
| Gemini 3.5 Flash-Lite | +1.9 | *self* | −0.1 | −0.8 | +11.3 |
| | [1.5, 2.2] | | [−0.1, 0.0] | [−0.9, −0.6] | [9.8, 12.8] |
| GPT-5.6 Luna | +2.7 | +0.1 | *self* | −0.2 | +14.5 |
| | [2.3, 3.1] | [0.0, 0.1] | | [−0.4, 0.0] | [12.8, 16.4] |
| Qwen3-235B | +0.6 | +0.8 | +0.2 | *self* | +8.2 |
| | [0.1, 1.1] | [0.6, 0.9] | [0.0, 0.4] | | [7.3, 9.0] |
| Grok 4.20 | −9.7 | −11.3 | −14.5 | −8.2 | *self* |
| | [−10.9, −8.4] | [−12.8, −9.8] | [−16.4, −12.8] | [−9.0, −7.3] | |
| *Units paid per seat, row model minus column model* | | | | | |
| Claude Haiku 4.5 | *self* | +1.9 | +3.0 | +0.9 | −9.7 |
| Gemini 3.5 Flash-Lite | −1.9 | *self* | +0.1 | +1.5 | −11.3 |
| GPT-5.6 Luna | −3.0 | −0.1 | *self* | +0.7 | −14.8 |
| Qwen3-235B | −0.9 | −1.5 | −0.7 | *self* | −8.4 |
| Grok 4.20 | +9.7 | +11.3 | +14.8 | +8.4 | *self* |

## S10.3 Which tables reach the target

Table S28 gives the numbers behind the lower triangle of Figure 3. Table S29 splits each pair by how many seats the first model holds and by risk level. Tables that mix Qwen3-235B with Gemini 3.5 Flash-Lite or GPT-5.6 Luna fail at almost every split, and only improve when Qwen3-235B holds a single seat. Tables with Grok 4.20 reach the target in the great majority of games, with two exceptions: five GPT-5.6 Luna seats beside one Grok 4.20 seat, and four or five Qwen3-235B seats beside Grok 4.20. Success barely changes with the risk level, as the self-play results would predict.

**Table S28:** Games reaching the target: count out of 150 for each pair, and out of 60 for self-play on the diagonal (the base grid and the rerun at $p = 0.1$, 0.5 and 0.9).

| Model | Claude Haiku 4.5 | Gemini 3.5 Flash-Lite | GPT-5.6 Luna | Qwen3 235B | Grok 4.20 |
|---|---|---|---|---|---|
| Claude Haiku 4.5 | *60/60* | 149/150 | 110/150 | 109/150 | 150/150 |
| Gemini 3.5 Flash-Lite | 149/150 | *60/60* | 137/150 | 13/150 | 150/150 |
| GPT-5.6 Luna | 110/150 | 137/150 | *47/60* | 15/150 | 142/150 |
| Qwen3-235B | 109/150 | 13/150 | 15/150 | *0/60* | 140/150 |
| Grok 4.20 | 150/150 | 150/150 | 142/150 | 140/150 | *60/60* |

**Table S29:** Games reaching the target (%). (a) By the number of seats the first-named model holds, pooled over risk levels (30 games per entry; *All*: 150). (b) By risk level, pooled over seat splits (50 games per entry), with self-play at the same risk levels (20 games per entry).

**(a)** By seats held by the first model

| | | Seats held by the first model | | | | | |
|---|---|---|---|---|---|---|---|
| First model | Second model | 1 | 2 | 3 | 4 | 5 | All |
| Claude Haiku 4.5 | Gemini 3.5 Flash-Lite | 100 | 100 | 97 | 100 | 100 | 99 |
| Claude Haiku 4.5 | GPT-5.6 Luna | 53 | 60 | 73 | 83 | 97 | 73 |
| Claude Haiku 4.5 | Qwen3-235B | 23 | 70 | 77 | 93 | 100 | 73 |
| Claude Haiku 4.5 | Grok 4.20 | 100 | 100 | 100 | 100 | 100 | 100 |
| Gemini 3.5 Flash-Lite | GPT-5.6 Luna | 90 | 87 | 93 | 87 | 100 | 91 |
| Gemini 3.5 Flash-Lite | Qwen3-235B | 0 | 0 | 7 | 0 | 37 | 9 |
| Gemini 3.5 Flash-Lite | Grok 4.20 | 100 | 100 | 100 | 100 | 100 | 100 |
| GPT-5.6 Luna | Qwen3-235B | 3 | 7 | 3 | 7 | 30 | 10 |
| GPT-5.6 Luna | Grok 4.20 | 97 | 100 | 100 | 100 | 77 | 95 |
| Qwen3-235B | Grok 4.20 | 100 | 100 | 100 | 83 | 83 | 93 |
| All pairs | | 67 | 72 | 75 | 75 | 82 | 74 |

**(b)** By risk level

| | | Risk $p$ | | |
|---|---|---|---|---|
| First model | Second model | 0.1 | 0.5 | 0.9 |
| Claude Haiku 4.5 | Gemini 3.5 Flash-Lite | 98 | 100 | 100 |
| Claude Haiku 4.5 | GPT-5.6 Luna | 66 | 72 | 82 |
| Claude Haiku 4.5 | Qwen3-235B | 66 | 74 | 78 |
| Claude Haiku 4.5 | Grok 4.20 | 100 | 100 | 100 |
| Gemini 3.5 Flash-Lite | GPT-5.6 Luna | 86 | 90 | 98 |
| Gemini 3.5 Flash-Lite | Qwen3-235B | 8 | 8 | 10 |
| Gemini 3.5 Flash-Lite | Grok 4.20 | 100 | 100 | 100 |
| GPT-5.6 Luna | Qwen3-235B | 10 | 8 | 12 |
| GPT-5.6 Luna | Grok 4.20 | 98 | 92 | 94 |
| Qwen3-235B | Grok 4.20 | 94 | 94 | 92 |
| *Self-play* | | | | |
| Claude Haiku 4.5 | | 100 | 100 | 100 |
| Gemini 3.5 Flash-Lite | | 100 | 100 | 100 |
| GPT-5.6 Luna | | 80 | 75 | 80 |
| Qwen3-235B | | 0 | 0 | 0 |
| Grok 4.20 | | 100 | 100 | 100 |

## S10.4 The last round

Before round 10, a pool of exactly 108 means that the fair share from every seat completes the target, and that no single seat can do it alone. Table S30 counts what each model pays in that situation, over all mixed tables. Qwen3-235B pays nothing in almost every such decision, while the other models rarely do. Replacing all Qwen3-235B round-10 zeros by 2 would bring 298 of the 323 failed tables containing it to the target. This arithmetic repair holds every other recorded action fixed; it does not show that changing one seat suffices or predict responses to an actual intervention.

**Table S30:** Top: round-10 decisions at mixed tables whose pool after round 9 is exactly 108, by model: share paying 0, number of decisions and number of games. Bottom: missed mixed tables, those with a Qwen3-235B seat, and those among them that reach 120 when Qwen3-235B's round-10 zeros are replaced by 2.

| Model | Paid 0 (%) | Decisions | Games |
|---|---|---|---|
| Claude Haiku 4.5 | 8.7 | 149 | 74 |
| Gemini 3.5 Flash-Lite | 0.1 | 956 | 287 |
| GPT-5.6 Luna | 1.4 | 775 | 245 |
| Qwen3-235B | 92.3 | 598 | 198 |
| Grok 4.20 | 0.0 | 78 | 48 |
| Two-model tables that missed the target | | | 385 |
| with a Qwen3-235B seat | | | 323 |
| reaching 120 if Qwen3-235B paid 2 instead of 0 in round 10 | | | 298 (92%) |

## S10.5 Slack and generous partners

Table S31 gives the numbers behind Figure 4 of the paper and extends it. Part (a) shows that one Qwen3-235B seat breaks tables of Gemini 3.5 Flash-Lite and GPT-5.6 Luna but not tables of Claude Haiku 4.5 or Grok 4.20, and, in the last column, that only Grok 4.20 pays enough to carry five Qwen3-235B seats to the target most of the time. Part (b) shows that the more Grok 4.20 seats a table holds, the less Claude Haiku 4.5, Gemini 3.5 Flash-Lite and GPT-5.6 Luna pay, while Qwen3-235B pays slightly more. Part (c) shows that this under-matching happens only beside Grok 4.20: beside five seats of any other model, each model pays close to the fair share.

Table S32 asks when the payment falls short beside Grok 4.20. The shortfall of a seat in a round is how far its payment falls below 2. Claude Haiku 4.5, Gemini 3.5 Flash-Lite and Qwen3-235B fall short mostly after the pool has reached 120, which fits stopping once the outcome is settled. GPT-5.6 Luna falls short mostly before, so it pays less from the start of the game.

**Table S31:** Slack at mixed tables, pooled over risk levels, with 95% intervals. (a) Games reaching the target: six seats of the model (60 games), five seats and one Qwen3-235B seat (30 games), and one seat beside five Qwen3-235B seats (30 games). (b) Units paid per seat of the model beside one to five Grok 4.20 seats (30 games each), and beside five seats of another model (90 games). (c) Units paid by a single seat beside five seats of each other model (30 games each).

*(a) Target reached*

| Model | Six seats of the model: % | Six seats of the model: 95% CI | Five seats and one Qwen3-235B seat: Reached | Five seats and one Qwen3-235B seat: % | Five seats and one Qwen3-235B seat: 95% CI | One seat and five Qwen3-235B seats |
|---|---|---|---|---|---|---|
| Claude Haiku 4.5 | 100 | [100, 100] | 30/30 | 100 | [100, 100] | 7/30 |
| Gemini 3.5 Flash-Lite | 100 | [100, 100] | 11/30 | 37 | [20, 53] | 0/30 |
| GPT-5.6 Luna | 78 | [68, 88] | 9/30 | 30 | [13, 47] | 1/30 |
| Grok 4.20 | 100 | [100, 100] | 30/30 | 100 | [100, 100] | 25/30 |

*(b) Units paid per seat of the model, by the seats of the other model*

| Model | Grok 4.20 seats: 1 | 2 | 3 | 4 | 5 | Five seats of another model |
|---|---|---|---|---|---|---|
| Claude Haiku 4.5 | 20.9 | 19.5 | 19.4 | 19.0 | 16.1 | 21.2 |
| | [20.5, 21.2] | [18.8, 20.2] | [19.0, 19.8] | [18.1, 20.0] | [14.9, 17.1] | [20.8, 21.6] |
| Gemini 3.5 Flash-Lite | 19.5 | 19.1 | 17.9 | 16.3 | 14.6 | 19.8 |
| | [19.2, 19.9] | [18.7, 19.4] | [17.0, 18.8] | [15.4, 17.1] | [13.1, 16.0] | [19.6, 19.9] |
| GPT-5.6 Luna | 18.7 | 18.0 | 14.3 | 12.0 | 10.5 | 19.1 |
| | [18.2, 19.3] | [17.4, 18.6] | [12.6, 15.9] | [10.1, 14.0] | [8.5, 12.7] | [18.8, 19.5] |
| Qwen3-235B | 19.1 | 20.6 | 20.2 | 23.4 | 20.8 | 19.7 |
| | [18.9, 19.3] | [19.9, 21.3] | [19.6, 20.7] | [22.5, 24.3] | [20.2, 21.4] | [19.2, 20.2] |

*(c) Units paid by one seat of the model beside five seats of each other model*

| Model | Claude Haiku 4.5 | Gemini 3.5 Flash-Lite | GPT-5.6 Luna | Qwen3 235B | Grok 4.20 |
|---|---|---|---|---|---|
| Claude Haiku 4.5 | *self* | 20.9 | 21.4 | 21.4 | 16.1 |
| | | [20.5, 21.5] | [20.6, 22.2] | [20.7, 22.2] | [14.9, 17.1] |
| Gemini 3.5 Flash-Lite | 19.2 | *self* | 20.1 | 20.0 | 14.6 |
| | [18.9, 19.5] | | [20.0, 20.3] | [20.0, 20.0] | [13.1, 16.0] |
| GPT-5.6 Luna | 18.1 | 20.1 | *self* | 19.3 | 10.5 |
| | [17.1, 18.9] | [20.0, 20.2] | | [18.9, 19.6] | [8.5, 12.7] |
| Qwen3-235B | 20.5 | 18.9 | 19.7 | *self* | 20.8 |
| | [19.4, 21.7] | [18.5, 19.3] | [19.0, 20.6] | | [20.2, 21.4] |

**Table S32:** Shortfall below 2 per seat and round beside one to five Grok 4.20 seats (150 games per model, pooled over seat splits and risk levels), before and after the round in which the pool reaches 120. *Per round*: mean over games of the mean shortfall per seat and round; *Games*: games with at least one round in that phase. Right: units short per seat per game before and after, and the share that comes after.

| Model | Before the pool reached 120: Per round | 95% CI | Games | After the pool reached 120: Per round | 95% CI | Games | Units short per seat per game: Before | After | After (%) |
|---|---|---|---|---|---|---|---|---|---|
| Claude Haiku 4.5 | 0.03 | [0.02, 0.05] | 150 | 1.88 | [1.83, 1.93] | 124 | 0.3 | 2.7 | 91 |
| Gemini 3.5 Flash-Lite | 0.04 | [0.02, 0.06] | 150 | 1.83 | [1.73, 1.92] | 88 | 0.3 | 2.3 | 88 |
| GPT-5.6 Luna | 0.46 | [0.38, 0.53] | 150 | 2.00 | [2.00, 2.00] | 74 | 3.8 | 1.8 | 32 |
| Qwen3-235B | 0.05 | [0.04, 0.06] | 150 | 1.25 | [1.12, 1.38] | 113 | 0.5 | 1.3 | 72 |

# S11 Selection among models

This section supports Section 4.3 of the paper and its Figure 5. It gives the method in full, the payoffs that feed it, and every selection setting we computed.

## S11.1 Method

**Payoff functions.** For two models $A$ and $B$ and a risk level $p$, $\pi_A(k)$ is the mean over games of the mean expected payoff of the $A$ seats at a table with $k$ seats of $A$ and $6-k$ seats of $B$. For $1 \le k \le 5$ it comes from the mixed tables (ten games per value). For $k = 6$ it is $A$'s self-play payoff at the same risk level (20 games, pooling the base grid and the rerun). Tables S33 to S35 list every value.

**Fitness.** A population holds $N$ agents, $i$ of type $A$ and $N-i$ of type $B$, and tables of six are drawn at random without replacement. An $A$ agent sits with $k-1$ other $A$ agents among its five tablemates with the hypergeometric probability $\mathcal{H}(k-1; N-1, i-1)$, and a $B$ agent sits with $k$ $A$ agents with probability $\mathcal{H}(k; N-1, i)$, where

$$\mathcal{H}(k;\, M,\, m) = \binom{m}{k}\binom{M-m}{5-k}\Big/\binom{M}{5}.$$

The fitness of each type is its payoff averaged over these tables, as in Equation (2) of the paper. With $N = 6$ the whole population is a single table, so the $A$ agents earn exactly $\pi_A(i)$ and the $B$ agents $\pi_B(6-i)$.

**Fixation.** Under the pairwise comparison rule, an agent copies another with probability $1/(1 + e^{-\beta\Delta f})$, where $\Delta f$ is the other agent's fitness minus its own. The probability that one $A$ newcomer takes over a population of $B$ is

$$\rho_{B\to A} = \left(1 + \sum_{j=1}^{N-1} \prod_{i=1}^{j} e^{-\beta\,(f_A(i) - f_B(i))}\right)^{-1}.$$

We write the transition from a population of $B$ to a population of $A$ as $B \longrightarrow A$, the direction of the arrows in Figure S7. A takeover is favoured by selection when $\rho_{B\to A} > 1/N$, the probability for a neutral newcomer.

**Rare-mutation Fermi selection.** In the limit where a new model enters rarely, the population almost always holds a single model. From a population of one model, a newcomer of each of the four other models arrives with equal probability and takes over with the probability above. These transitions form a Markov chain over the five one-model populations, and its stationary distribution is the *stationary mass* of each model. The population's target rate and payoff per seat are the self-play target rates and payoffs weighted by these masses. The masses in Figure 5 of the paper and the invasion diagrams in Figure S7 are computed with EGTtools, and the analysis code checks that EGTtools and its own implementation give the same fixation probabilities and masses to within $10^{-7}$. For the two strongest-selection settings at $N = 60$ and $N = 120$ at $p = 0.1$, the eigenvector method sees two unit eigenvalues because some transition probabilities are too small to represent; an exact computation in log space gives the same masses. We report this finite-intensity, frequency-dependent process, not an infinite-intensity ranking or group reproduction. Changing $N$ changes both sampling probabilities and the fixation path.

**Gradient of selection.** For a population with $i$ agents of type $A$ out of $N$, the expected change in the share $x = i/N$ of type $A$ in one step of the pairwise comparison rule is

$$G(x) = x\,(1-x)\,\frac{N}{N-1}\,\tanh\!\left(\frac{\beta}{2}\,\big(f_A(i) - f_B(i)\big)\right).$$

Where $G > 0$ selection pushes the share of $A$ up, and where $G < 0$ it pushes it down. A zero of $G$ at which the sign changes from positive to negative is stable, and one at which it changes from negative to positive is unstable. This population-level gradient is distinct from the table-composition payoffs plotted in Figure 5 of the paper.

**Uncertainty.** Each of 1,000 bootstrap replicates resamples the games of every mixed cell and every self-play cell, rebuilds all payoff functions and repeats the whole analysis. We report how often the winner of the point estimate stays on top, and percentile intervals for the population's target rate and payoff.

## S11.2 Numbers behind the main text

| Statement in the paper, made precise | Value |
|---|---|
| One table ($N = 6$, $\beta = 1$): Qwen3-235B's mass at $p = 0.1$; 0.5; 0.9 | 0.89; 0.66; 0.47 |
| One table: population target rate at $p = 0.5$; 0.9 (%) | 28; 47 |
| a model picked at random (%) | 75; 76 |
| stronger selection, $\beta = 10$ (%) | 0; 1 |
| Larger population ($N = 30$, $\beta = 1$): Gemini 3.5 Flash-Lite's mass at $p = 0.5$; 0.9 | 1.00; 1.00 |
| Larger population at $p = 0.9$: population target rate (%); payoff per seat | 100; 20.0 |
| Larger population at $p = 0.1$: model with the largest mass | Qwen3-235B |
| Winners unchanged for $\beta = 10$, for $N$ of 60 and 120, and with grid-only self-play | yes |
| Smallest share of bootstrap replicates in which the winner stays on top (%): pooled self-play; grid only | 93; 90 |
| Weak selection ($\beta = 0.1$), one table ($N = 6$): largest mass of any model at any risk level | 0.34 |
| model with the largest mass at $p = 0.5$; 0.9 | GPT-5.6 Luna; GPT-5.6 Luna |
| Weak selection, larger population ($N = 30$): model with the largest mass at $p = 0.5$; 0.9 | Gemini 3.5 Flash-Lite; Gemini 3.5 Flash-Lite |
| Larger population: smallest population from which Gemini 3.5 Flash-Lite stays the winner at $p = 0.9$; at $p = 0.5$ | 7; 8 |
| Cells (pair, split, risk) in which a single seat beside five Grok 4.20 seats earns more than they do | 12 of 12 |
| Cells with Qwen3-235B seats; Qwen3-235B significantly below its tablemates; significantly above | 60; one; 40 |

## S11.3 Payoff functions

Tables S33 to S35 give $\pi_X(k)$ for every ordered pair of models. Reading along a row shows how a model's payoff changes as its own share of the table grows. Two patterns drive the selection results. Beside Qwen3-235B, Gemini 3.5 Flash-Lite and GPT-5.6 Luna earn little whenever the target is missed, which is at almost every split. And Grok 4.20 earns less than its tablemates, because it pays the most.

**Table S33:** Payoff functions at $p = 0.1$: mean expected payoff per seat of the model in the first column at a table with $k$ of its seats beside $6 - k$ seats of the model in the second column (ten games per entry for $k \leq 5$; $k = 6$ is self-play, 20 games, and is the same for every partner).

| Model | Beside | Seats of the model at the table | | | | | |
|---|---|---|---|---|---|---|---|
| | | 1 | 2 | 3 | 4 | 5 | 6 |
| Claude Haiku 4.5 | Gemini 3.5 Flash-Lite | 19.00 | 18.30 | 18.31 | 18.55 | 18.20 | 18.97 |
| | GPT-5.6 Luna | 18.02 | 17.17 | 17.47 | 17.99 | 18.40 | 18.97 |
| | Qwen3-235B | 17.08 | 18.00 | 18.00 | 18.15 | 18.72 | 18.97 |
| | Grok 4.20 | 23.00 | 21.50 | 20.27 | 20.10 | 19.56 | 18.97 |
| Gemini 3.5 Flash-Lite | Claude Haiku 4.5 | 20.40 | 21.20 | 20.27 | 20.00 | 20.12 | 19.98 |
| | GPT-5.6 Luna | 19.60 | 19.60 | 19.60 | 19.65 | 19.96 | 19.98 |
| | Qwen3-235B | 18.00 | 18.09 | 18.20 | 18.00 | 18.71 | 19.98 |
| | Grok 4.20 | 21.60 | 22.00 | 21.60 | 20.85 | 20.84 | 19.98 |
| GPT-5.6 Luna | Claude Haiku 4.5 | 22.20 | 20.67 | 20.83 | 20.54 | 19.16 | 19.70 |
| | Gemini 3.5 Flash-Lite | 20.00 | 19.69 | 19.86 | 19.87 | 19.67 | 19.70 |
| | Qwen3-235B | 19.08 | 19.39 | 18.25 | 18.49 | 18.62 | 19.70 |
| | Grok 4.20 | 28.38 | 25.30 | 21.60 | 22.75 | 21.16 | 19.70 |
| Qwen3-235B | Claude Haiku 4.5 | 17.20 | 17.89 | 20.29 | 19.09 | 19.48 | 19.49 |
| | Gemini 3.5 Flash-Lite | 19.32 | 19.80 | 19.52 | 19.62 | 19.55 | 19.49 |
| | GPT-5.6 Luna | 18.92 | 19.42 | 19.76 | 18.64 | 19.19 | 19.49 |
| | Grok 4.20 | 20.00 | 18.20 | 20.67 | 19.70 | 20.64 | 19.49 |
| Grok 4.20 | Claude Haiku 4.5 | 14.60 | 12.00 | 14.33 | 10.45 | 10.60 | 6.68 |
| | Gemini 3.5 Flash-Lite | 11.20 | 14.60 | 14.00 | 12.00 | 14.20 | 6.68 |
| | GPT-5.6 Luna | 12.00 | 12.70 | 16.40 | 12.55 | 10.60 | 6.68 |
| | Qwen3-235B | 13.80 | 13.01 | 13.73 | 11.90 | 10.88 | 6.68 |

**Table S34:** Payoff functions at $p = 0.5$, as in Table S33.

| Model | Beside | Seats of the model at the table | | | | | |
|---|---|---|---|---|---|---|---|
| | | 1 | 2 | 3 | 4 | 5 | 6 |
| Claude Haiku 4.5 | Gemini 3.5 Flash-Lite | 19.40 | 18.30 | 17.73 | 18.60 | 18.56 | 18.75 |
| | GPT-5.6 Luna | 14.90 | 14.05 | 13.97 | 17.33 | 17.38 | 18.75 |
| | Qwen3-235B | 11.80 | 17.50 | 15.50 | 18.55 | 18.76 | 18.75 |
| | Grok 4.20 | 24.40 | 21.00 | 20.93 | 19.90 | 18.92 | 18.75 |
| Gemini 3.5 Flash-Lite | Claude Haiku 4.5 | 20.80 | 20.50 | 20.53 | 19.80 | 20.04 | 20.00 |
| | GPT-5.6 Luna | 18.90 | 18.95 | 19.00 | 18.00 | 20.00 | 20.00 |
| | Qwen3-235B | 10.00 | 10.00 | 10.00 | 10.00 | 13.96 | 20.00 |
| | Grok 4.20 | 25.00 | 23.70 | 21.13 | 20.75 | 20.36 | 20.00 |
| GPT-5.6 Luna | Claude Haiku 4.5 | 21.40 | 20.60 | 17.33 | 16.90 | 16.14 | 17.53 |
| | Gemini 3.5 Flash-Lite | 20.00 | 18.10 | 19.07 | 19.00 | 19.04 | 17.53 |
| | Qwen3-235B | 10.50 | 10.15 | 10.43 | 10.35 | 14.06 | 17.53 |
| | Grok 4.20 | 29.20 | 29.30 | 26.80 | 21.20 | 16.72 | 17.53 |
| Qwen3-235B | Claude Haiku 4.5 | 20.60 | 15.85 | 17.70 | 16.23 | 13.82 | 10.79 |
| | Gemini 3.5 Flash-Lite | 14.40 | 10.85 | 10.87 | 10.98 | 10.58 | 10.79 |
| | GPT-5.6 Luna | 13.30 | 10.60 | 10.67 | 10.90 | 10.66 | 10.79 |
| | Grok 4.20 | 19.00 | 16.10 | 18.80 | 17.00 | 19.72 | 10.79 |
| Grok 4.20 | Claude Haiku 4.5 | 14.80 | 12.70 | 13.20 | 11.45 | 7.12 | 5.88 |
| | Gemini 3.5 Flash-Lite | 15.20 | 13.10 | 14.67 | 8.35 | 8.24 | 5.88 |
| | GPT-5.6 Luna | 11.50 | 13.80 | 9.20 | 6.80 | 8.40 | 5.88 |
| | Qwen3-235B | 13.40 | 10.60 | 10.07 | 13.35 | 4.68 | 5.88 |

**Table S35:** Payoff functions at $p = 0.9$, as in Table S33.

| Model | Beside | Seats of the model at the table | | | | | |
|---|---|---|---|---|---|---|---|
| | | 1 | 2 | 3 | 4 | 5 | 6 |
| Claude Haiku 4.5 | Gemini 3.5 Flash-Lite | 18.80 | 18.60 | 17.60 | 18.50 | 19.00 | 18.85 |
| | GPT-5.6 Luna | 11.18 | 10.05 | 18.07 | 18.20 | 18.60 | 18.85 |
| | Qwen3-235B | 4.72 | 15.50 | 16.40 | 19.10 | 18.64 | 18.85 |
| | Grok 4.20 | 24.40 | 20.40 | 20.60 | 21.55 | 18.88 | 18.85 |
| Gemini 3.5 Flash-Lite | Claude Haiku 4.5 | 21.20 | 20.30 | 20.33 | 19.85 | 19.92 | 20.00 |
| | GPT-5.6 Luna | 19.80 | 18.10 | 19.93 | 19.95 | 20.00 | 20.00 |
| | Qwen3-235B | 2.00 | 2.01 | 3.81 | 2.00 | 9.16 | 20.00 |
| | Grok 4.20 | 29.60 | 25.50 | 23.53 | 21.25 | 20.20 | 20.00 |
| GPT-5.6 Luna | Claude Haiku 4.5 | 21.00 | 21.20 | 20.73 | 11.41 | 12.95 | 16.42 |
| | Gemini 3.5 Flash-Lite | 19.80 | 20.00 | 20.00 | 18.21 | 19.92 | 16.42 |
| | Qwen3-235B | 3.80 | 3.66 | 2.03 | 3.71 | 7.47 | 16.42 |
| | Grok 4.20 | 30.60 | 29.40 | 28.73 | 21.95 | 15.62 | 16.42 |
| Qwen3-235B | Claude Haiku 4.5 | 20.80 | 17.70 | 18.62 | 14.68 | 5.85 | 2.15 |
| | Gemini 3.5 Flash-Lite | 9.52 | 2.20 | 3.73 | 2.17 | 2.18 | 2.15 |
| | GPT-5.6 Luna | 7.26 | 3.78 | 2.15 | 4.00 | 3.96 | 2.15 |
| | Grok 4.20 | 18.60 | 15.50 | 20.07 | 18.85 | 13.46 | 2.15 |
| Grok 4.20 | Claude Haiku 4.5 | 14.80 | 7.80 | 10.87 | 13.15 | 2.08 | 2.62 |
| | Gemini 3.5 Flash-Lite | 17.20 | 12.00 | 7.40 | 5.50 | 1.24 | 2.62 |
| | GPT-5.6 Luna | 7.40 | 11.70 | 5.47 | 6.20 | 5.80 | 2.62 |
| | Qwen3-235B | 8.78 | 6.20 | 11.67 | 8.15 | 3.48 | 2.62 |

## S11.4 Every selection setting

The winners of Figure 5 of the paper are robust across several checks: each winner remains on top in at least 93% of bootstrap samples, and the same winners emerge at $\beta = 10$, at $N$ of 60 and 120 and when using grid-only self-play payoffs. Under weak selection ($\beta = 0.1$), the largest one-table mass is only 0.34; GPT-5.6 Luna narrowly leads at moderate and high risk. Which model wins therefore depends on the base-prompt payoffs, whose signatures vary with wording.

Table S36 lists the stationary mass of each model, and the population's target rate and payoff, for one table ($N = 6$) and for larger populations ($N = 30$, 60 and 120), at the selection intensities $\beta = 0.1$, 1 and 10; the two largest populations are computed only for $\beta \geq 1$. Figure S6 shows how the masses move with the selection intensity over a wider range. In a population of one table, stronger selection concentrates the weight on Qwen3-235B at every risk level and group success falls to almost zero. In the larger population, Gemini 3.5 Flash-Lite already leads at moderate and high risk under weak selection and takes all the weight from $\beta = 1$ on, while Qwen3-235B takes it at low risk. Grok 4.20 loses its weight under both population sizes as soon as selection is more than very weak.

**Table S36:** Stationary mass of each model under each selection setting, with the population's self-play target rate and expected payoff per seat. *Random model*: every model equally weighted.

| Population | $N$ | $\beta$ | Claude Haiku 4.5 | Gemini 3.5 Flash-Lite | GPT-5.6 Luna | Qwen3 235B | Grok 4.20 | Target (%) | Payoff per seat |
|---|---|---|---|---|---|---|---|---|---|
| $p = 0.1$ | | | | | | | | | |
| One table | 6 | 0.1 | 0.12 | 0.24 | 0.31 | 0.34 | 0.00 | 60 | 19.6 |
| | 6 | 1 | 0.00 | 0.01 | 0.09 | 0.89 | 0.00 | 9 | 19.5 |
| | 6 | 10 | 0.00 | 0.00 | 0.00 | 1.00 | 0.00 | 0 | 19.5 |
| Many tables | 30 | 0.1 | 0.01 | 0.08 | 0.22 | 0.69 | 0.00 | 27 | 19.6 |
| | 30 | 1 | 0.00 | 0.00 | 0.00 | 1.00 | 0.00 | 0 | 19.5 |
| | 30 | 10 | 0.00 | 0.00 | 0.00 | 1.00 | 0.00 | 0 | 19.5 |
| | 60 | 1 | 0.00 | 0.00 | 0.00 | 1.00 | 0.00 | 0 | 19.5 |
| | 60 | 10 | 0.00 | 0.00 | 0.00 | 1.00 | 0.00 | 0 | 19.5 |
| | 120 | 1 | 0.00 | 0.00 | 0.00 | 1.00 | 0.00 | 0 | 19.5 |
| | 120 | 10 | 0.00 | 0.00 | 0.00 | 1.00 | 0.00 | 0 | 19.5 |
| Random model | | | 0.20 | 0.20 | 0.20 | 0.20 | 0.20 | 76 | 17.0 |
| $p = 0.5$ | | | | | | | | | |
| One table | 6 | 0.1 | 0.12 | 0.26 | 0.33 | 0.29 | 0.00 | 63 | 16.4 |
| | 6 | 1 | 0.01 | 0.07 | 0.26 | 0.66 | 0.00 | 28 | 13.3 |
| | 6 | 10 | 0.00 | 0.00 | 0.00 | 1.00 | 0.00 | 0 | 10.8 |
| Many tables | 30 | 0.1 | 0.01 | 0.66 | 0.32 | 0.02 | 0.00 | 90 | 19.0 |
| | 30 | 1 | 0.00 | 1.00 | 0.00 | 0.00 | 0.00 | 100 | 20.0 |
| | 30 | 10 | 0.00 | 1.00 | 0.00 | 0.00 | 0.00 | 100 | 20.0 |
| | 60 | 1 | 0.00 | 1.00 | 0.00 | 0.00 | 0.00 | 100 | 20.0 |
| | 60 | 10 | 0.00 | 1.00 | 0.00 | 0.00 | 0.00 | 100 | 20.0 |
| | 120 | 1 | 0.00 | 1.00 | 0.00 | 0.00 | 0.00 | 100 | 20.0 |
| | 120 | 10 | 0.00 | 1.00 | 0.00 | 0.00 | 0.00 | 100 | 20.0 |
| Random model | | | 0.20 | 0.20 | 0.20 | 0.20 | 0.20 | 75 | 14.6 |
| $p = 0.9$ | | | | | | | | | |
| One table | 6 | 0.1 | 0.13 | 0.29 | 0.31 | 0.27 | 0.00 | 67 | 13.9 |
| | 6 | 1 | 0.00 | 0.23 | 0.29 | 0.47 | 0.00 | 47 | 10.6 |
| | 6 | 10 | 0.00 | 0.00 | 0.01 | 0.99 | 0.00 | 1 | 2.2 |
| Many tables | 30 | 0.1 | 0.00 | 0.79 | 0.21 | 0.00 | 0.00 | 96 | 19.3 |
| | 30 | 1 | 0.00 | 1.00 | 0.00 | 0.00 | 0.00 | 100 | 20.0 |
| | 30 | 10 | 0.00 | 1.00 | 0.00 | 0.00 | 0.00 | 100 | 20.0 |
| | 60 | 1 | 0.00 | 1.00 | 0.00 | 0.00 | 0.00 | 100 | 20.0 |
| | 60 | 10 | 0.00 | 1.00 | 0.00 | 0.00 | 0.00 | 100 | 20.0 |
| | 120 | 1 | 0.00 | 1.00 | 0.00 | 0.00 | 0.00 | 100 | 20.0 |
| | 120 | 10 | 0.00 | 1.00 | 0.00 | 0.00 | 0.00 | 100 | 20.0 |
| Random model | | | 0.20 | 0.20 | 0.20 | 0.20 | 0.20 | 76 | 12.0 |

Table S37 gives the bootstrap check for the two settings of the paper, with self-play from the grid and the rerun together and from the grid alone. The winners are the same in both cases and stay on top in the large majority of replicates. Some

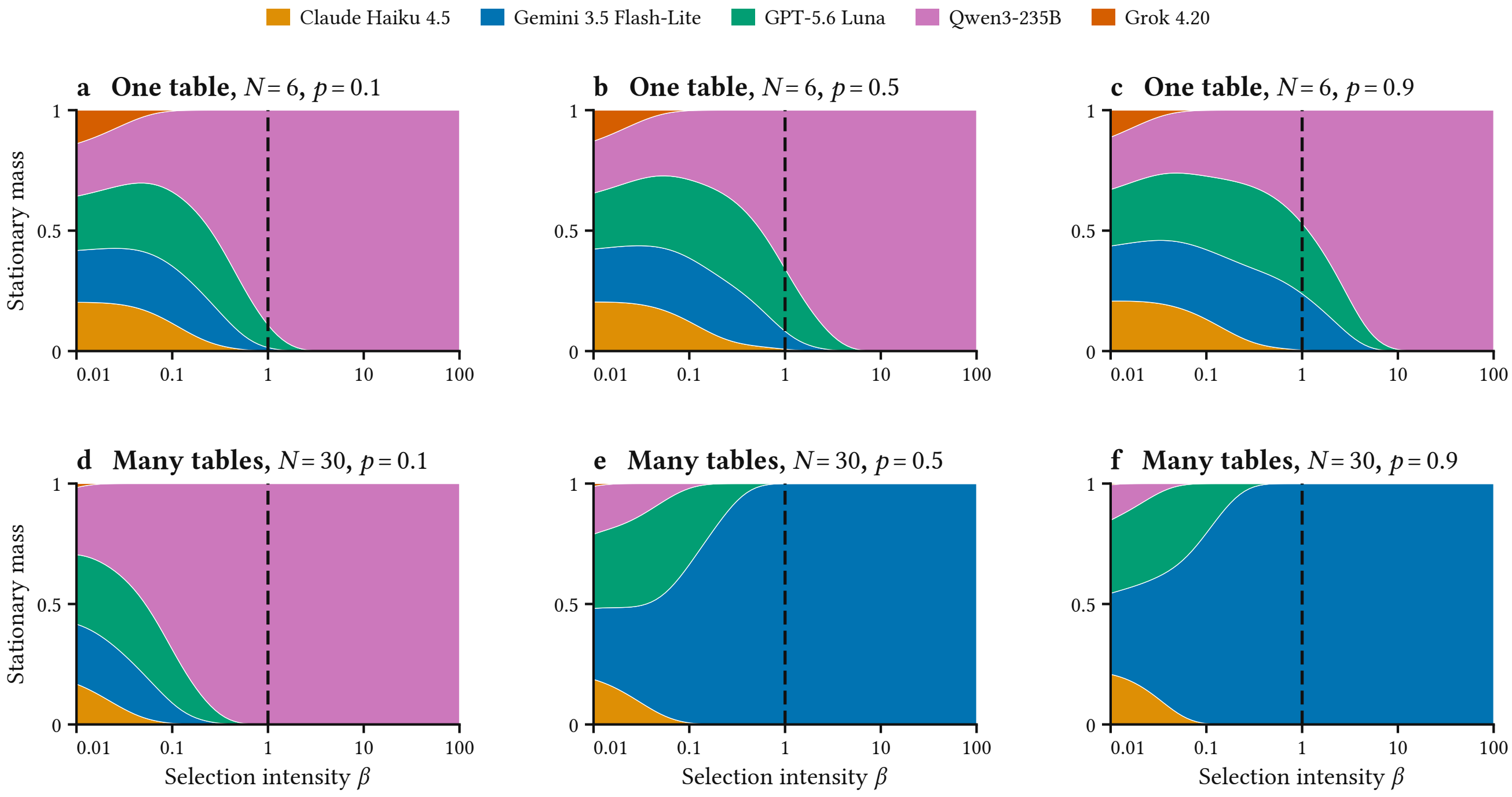


**Figure S6:** Stationary mass of the five models as the selection intensity $\beta$ grows from 0.01 to 100. Top: a population of one table ($N = 6$). Bottom: a larger population ($N = 30$). The dashed line marks $\beta = 1$, the setting of Figure 5 of the paper.

intervals for the target rate are wide, because a population's target rate depends on which model wins and the models differ sharply in self-play success.

Figure 5 of the paper shows the stationary masses and the table-composition payoffs for one pair. Figure S7 gives the full invasion diagrams behind those masses, for every pair of models at the three risk levels and both population sizes. Each arrow marks a takeover that selection favours.

**Table S37:** Bootstrap check of the two settings of Figure 5 ($\beta$ = 1; one table, $N$ = 6; larger population, $N$ = 30), over 1,000 replicates: the winner of the point estimate, the share of replicates in which it stays on top, and the population's target rate and payoff per seat with 95% intervals.

| Population | $p$ | Top model | Top in (%) | Target (%) | 95% CI | Payoff per seat | 95% CI |
|---|---|---|---|---|---|---|---|
| *Self-play from the grid and the rerun (20 games per model and risk level)* | | | | | | | |
| One table | 0.1 | Qwen3-235B | 98 | 9 | [2, 42] | 19.5 | [19.4, 19.6] |
| | 0.5 | Qwen3-235B | 95 | 28 | [11, 51] | 13.3 | [11.8, 15.4] |
| | 0.9 | Qwen3-235B | 93 | 47 | [37, 61] | 10.6 | [8.8, 13.0] |
| Many tables | 0.1 | Qwen3-235B | 95 | 0 | [0, 67] | 19.5 | [19.4, 19.7] |
| | 0.5 | Gemini 3.5 Flash-Lite | 99 | 100 | [98, 100] | 20.0 | [19.8, 20.0] |
| | 0.9 | Gemini 3.5 Flash-Lite | 99 | 100 | [100, 100] | 20.0 | [20.0, 20.0] |
| *Self-play from the grid only (10 games per model and risk level)* | | | | | | | |
| One table | 0.1 | Qwen3-235B | 97 | 9 | [1, 47] | 19.5 | [19.4, 19.6] |
| | 0.5 | Qwen3-235B | 95 | 29 | [12, 53] | 13.3 | [11.9, 15.6] |
| | 0.9 | Qwen3-235B | 94 | 44 | [33, 59] | 10.0 | [8.0, 12.6] |
| Many tables | 0.1 | Qwen3-235B | 94 | 0 | [0, 72] | 19.5 | [19.4, 19.7] |
| | 0.5 | Gemini 3.5 Flash-Lite | 90 | 100 | [97, 100] | 20.0 | [19.7, 20.0] |
| | 0.9 | Gemini 3.5 Flash-Lite | 98 | 100 | [100, 100] | 20.0 | [20.0, 20.0] |

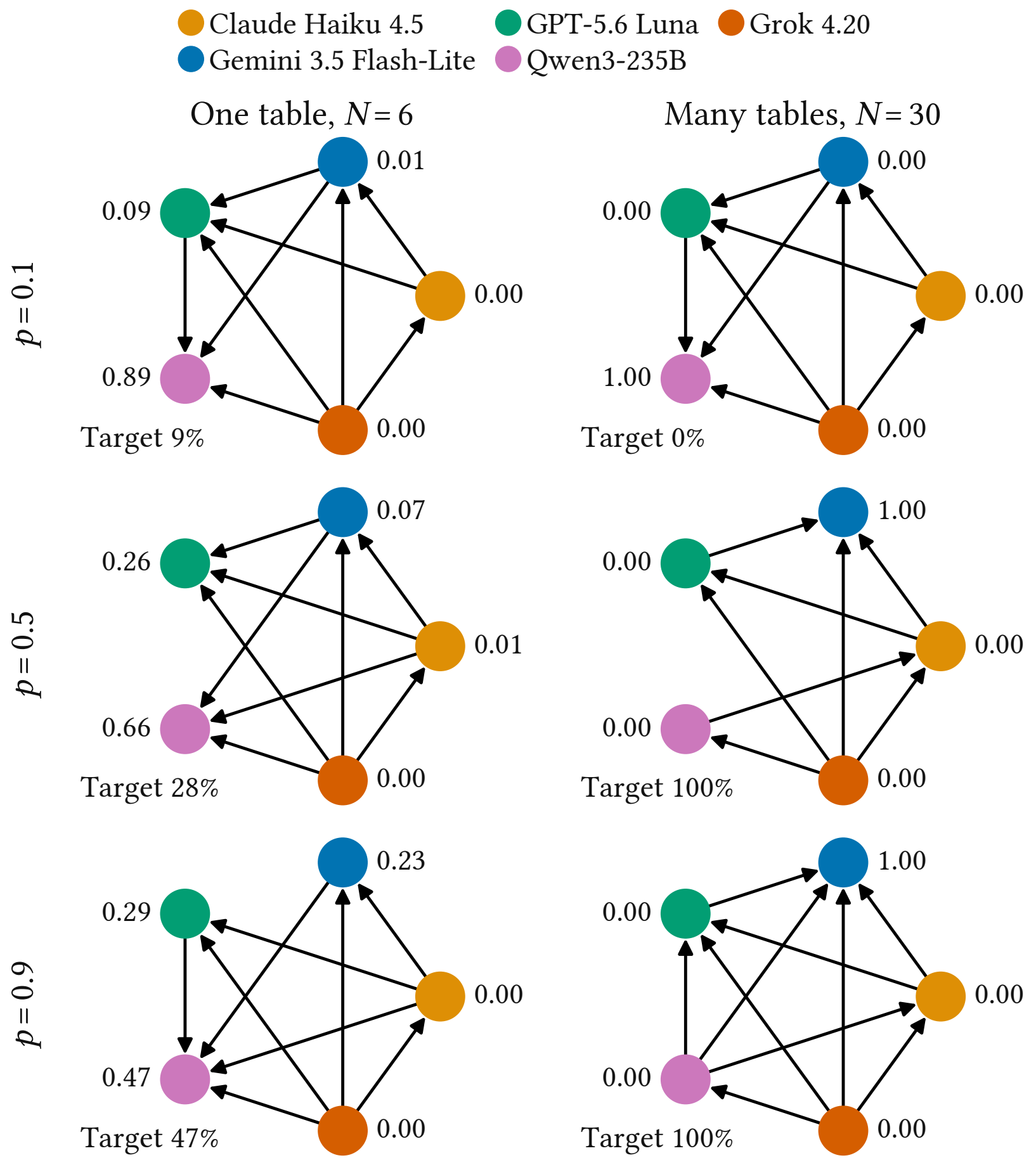


**Figure S7:** Invasion diagrams (rare-mutation Fermi selection, $\beta$ = 1), computed and drawn with EGTtools. A node is a one-model population, with its stationary mass beside it. An arrow from A to B means that one B newcomer takes over an A population with more than the neutral probability $1/N$. Left: one table, $N$ = 6. Right: larger randomly mixed population, $N$ = 30. Bottom left: self-play target rate weighted by mass.

## S11.5 Facts about single tables

Table S38 checks, cell by cell, the two table-level facts that the selection results rest on. A single seat of any model beside five Grok 4.20 seats earns more than they do in every cell. Qwen3-235B is significantly ahead of its tablemates in most cells and significantly behind in one, which is about what chance allows among so many tests. Part (c) shows how large the population must be before selection favours Gemini 3.5 Flash-Lite: a population only slightly larger than one table is enough.

**Table S38:** (a) Expected payoff of a single seat of each model beside five Grok 4.20 seats (ten games per risk level). (b) Qwen3-235B's expected payoff per seat against its tablemates in each cell of ten games, exact one-sided sign-flip tests in both directions. (c) The winner as the population size $N$ grows from 6 to 120.

| *(a) One seat beside five Grok 4.20 seats: expected payoff per seat* | | | | |
|---|---|---|---|---|
| Model | $p$ | One seat | Grok 4.20 seats | Ahead |
| Claude Haiku 4.5 | 0.1 | 23.00 | 10.60 | One seat |
| | 0.5 | 24.40 | 7.12 | One seat |
| | 0.9 | 24.40 | 2.08 | One seat |
| Gemini 3.5 Flash-Lite | 0.1 | 21.60 | 14.20 | One seat |
| | 0.5 | 25.00 | 8.24 | One seat |
| | 0.9 | 29.60 | 1.24 | One seat |
| GPT-5.6 Luna | 0.1 | 28.38 | 10.60 | One seat |
| | 0.5 | 29.20 | 8.40 | One seat |
| | 0.9 | 30.60 | 5.80 | One seat |
| Qwen3-235B | 0.1 | 20.00 | 10.88 | One seat |
| | 0.5 | 19.00 | 4.68 | One seat |
| | 0.9 | 18.60 | 3.48 | One seat |

| *(b) Qwen3-235B against its tablemates, exact one-sided sign-flip test per cell, $\alpha = 0.05$* | |
|---|---|
| Cells (pair, seats, $p$) with Qwen3-235B seats | 60 |
| Qwen3-235B significantly below its tablemates | 1 |
| Qwen3-235B significantly above its tablemates | 40 |
| Below: $p = 0.5$, two Qwen3-235B seats beside four Claude Haiku 4.5 seats, payoff gap −2.70, one-sided $P = 0.004$ | |

| *(c) Winner by population size, $\beta = 1$, $N$ from 6 to 120* | | | |
|---|---|---|---|
| $p$ | Winner at smaller $N$ | From $N$ | Winner up to $N = 120$ |
| 0.1 | none | 6 | Qwen3-235B |
| 0.5 | Qwen3-235B ($N = 6$ to 7) | 8 | Gemini 3.5 Flash-Lite |
| 0.9 | Qwen3-235B ($N = 6$) | 7 | Gemini 3.5 Flash-Lite |

## S11.6 Single entrants, candidate sets and welfare

The winners depend on the candidate set: without Qwen3-235B, GPT-5.6 Luna wins the one-table population at all tested risks, and without Gemini 3.5 Flash-Lite, it wins the larger population at moderate and high risk.

The following post-hoc checks reuse the observed payoffs without new model calls. We compare each single entrant's payoff with both resident self-play and its actual tablemates, then compute fixation along the full composition path. These are different comparisons: a favourable initial payoff gap alone does not determine fixation or stationary mass. We also remove each model in turn and rebuild the mutation chain over the retained strategies. This checks dependence on the candidate set, not on the prompt. Every risk and both population sizes are reported, with no new uncertainty intervals: Tables S39 to S41 give the single entrants and Table S42 the winner for every candidate set.

At $p = 0.5$, reaching the target is not required for optimal welfare: zero contribution also earns 20 per seat in expectation. We therefore report stationary-mass-weighted expected payoff and its share of the welfare optimum as well as target rate (Table S43). The outcomes describe one-model populations in the rare-mutation limit, not the success of mixed groups during invasion.

**Table S39:** Single entrant beside five residents at $p = 0.1$ (post-hoc point estimates, no new intervals). $u_B^{\text{self}}$: the resident's self-play payoff (20 games); $u_A(1)$ and $u_B(5)$: the entrant's and the residents' payoffs at the same mixed tables (10 games). $\Delta_{\text{self}} = u_A(1) - u_B^{\text{self}}$ compares separate conditions; $\Delta_{\text{table}} = u_A(1) - u_B(5)$ compares seats in the same games. $\rho_N$: fixation probability of one entrant in a population of $N$ at $\beta = 1$, over all table compositions; neutral fixation is $1/N$. These comparisons do not establish evolutionary stability.

| Resident $B$ | Entrant $A$ | $u_B^{\text{self}}$ | $u_A(1)$ | $u_B(5)$ | $\Delta_{\text{self}}$ | $\Delta_{\text{table}}$ | $\rho_6$ | $\rho_{30}$ |
|---|---|---|---|---|---|---|---|---|
| Claude Haiku 4.5 | Gemini 3.5 Flash-Lite | 18.97 | 20.40 | 18.20 | 1.43 | 2.20 | 0.89 | 0.80 |
| Claude Haiku 4.5 | GPT-5.6 Luna | 18.97 | 22.20 | 18.40 | 3.23 | 3.80 | 0.98 | 0.96 |
| Claude Haiku 4.5 | Qwen3-235B | 18.97 | 17.20 | 18.72 | −1.77 | −1.52 | 0.081 | 0.00012 |
| Claude Haiku 4.5 | Grok 4.20 | 18.97 | 14.60 | 19.56 | −4.37 | −4.96 | $3.7 \times 10^{-19}$ | $3.1 \times 10^{-114}$ |
| Gemini 3.5 Flash-Lite | Claude Haiku 4.5 | 19.98 | 19.00 | 20.12 | −0.98 | −1.12 | $5.8 \times 10^{-5}$ | $1.6 \times 10^{-23}$ |
| Gemini 3.5 Flash-Lite | GPT-5.6 Luna | 19.98 | 20.00 | 19.96 | 0.02 | 0.04 | 0.22 | 0.048 |
| Gemini 3.5 Flash-Lite | Qwen3-235B | 19.98 | 19.32 | 18.71 | −0.66 | 0.61 | 0.60 | 0.12 |
| Gemini 3.5 Flash-Lite | Grok 4.20 | 19.98 | 11.20 | 20.84 | −8.78 | −9.64 | $1.7 \times 10^{-18}$ | $1.3 \times 10^{-112}$ |
| GPT-5.6 Luna | Claude Haiku 4.5 | 19.70 | 18.02 | 19.16 | −1.68 | −1.14 | $5.8 \times 10^{-7}$ | $3.6 \times 10^{-33}$ |
| GPT-5.6 Luna | Gemini 3.5 Flash-Lite | 19.70 | 19.60 | 19.67 | −0.10 | −0.07 | 0.11 | 0.0048 |
| GPT-5.6 Luna | Qwen3-235B | 19.70 | 18.92 | 18.62 | −0.78 | 0.30 | 0.43 | 0.062 |
| GPT-5.6 Luna | Grok 4.20 | 19.70 | 12.00 | 21.16 | −7.70 | −9.16 | $1.4 \times 10^{-24}$ | $1.2 \times 10^{-142}$ |
| Qwen3-235B | Claude Haiku 4.5 | 19.48 | 17.08 | 19.48 | −2.40 | −2.40 | 0.0015 | $3.2 \times 10^{-14}$ |
| Qwen3-235B | Gemini 3.5 Flash-Lite | 19.48 | 18.00 | 19.55 | −1.48 | −1.55 | 0.00066 | $1.4 \times 10^{-15}$ |
| Qwen3-235B | GPT-5.6 Luna | 19.48 | 19.08 | 19.19 | −0.40 | −0.11 | 0.052 | $4.1 \times 10^{-6}$ |
| Qwen3-235B | Grok 4.20 | 19.48 | 13.80 | 20.64 | −5.68 | −6.84 | $2.6 \times 10^{-16}$ | $7.1 \times 10^{-101}$ |
| Grok 4.20 | Claude Haiku 4.5 | 6.68 | 23.00 | 10.60 | 16.32 | 12.40 | 1.0 | 1.0 |
| Grok 4.20 | Gemini 3.5 Flash-Lite | 6.68 | 21.60 | 14.20 | 14.92 | 7.40 | 1.0 | 1.0 |
| Grok 4.20 | GPT-5.6 Luna | 6.68 | 28.38 | 10.60 | 21.70 | 17.78 | 1.0 | 1.0 |
| Grok 4.20 | Qwen3-235B | 6.68 | 20.00 | 10.88 | 13.32 | 9.12 | 1.0 | 1.0 |

**Table S40:** Single entrant beside five residents at $p = 0.5$, as in Table S39.

| Resident $B$ | Entrant $A$ | $u_B^{\text{self}}$ | $u_A(1)$ | $u_B(5)$ | $\Delta_{\text{self}}$ | $\Delta_{\text{table}}$ | $\rho_6$ | $\rho_{30}$ |
|---|---|---|---|---|---|---|---|---|
| Claude Haiku 4.5 | Gemini 3.5 Flash-Lite | 18.75 | 20.80 | 18.56 | 2.05 | 2.24 | 0.89 | 0.87 |
| Claude Haiku 4.5 | GPT-5.6 Luna | 18.75 | 21.40 | 17.38 | 2.65 | 4.02 | 0.98 | 0.94 |
| Claude Haiku 4.5 | Qwen3-235B | 18.75 | 20.60 | 18.76 | 1.85 | 1.84 | 0.21 | $2.6 \times 10^{-10}$ |
| Claude Haiku 4.5 | Grok 4.20 | 18.75 | 14.80 | 18.92 | −3.95 | −4.12 | $1.2 \times 10^{-20}$ | $1.0 \times 10^{-122}$ |
| Gemini 3.5 Flash-Lite | Claude Haiku 4.5 | 20.00 | 19.40 | 20.04 | −0.60 | −0.64 | 0.00010 | $1.1 \times 10^{-22}$ |
| Gemini 3.5 Flash-Lite | GPT-5.6 Luna | 20.00 | 20.00 | 20.00 | 0.00 | 0.00 | 0.19 | 0.00016 |
| Gemini 3.5 Flash-Lite | Qwen3-235B | 20.00 | 14.40 | 13.96 | −5.60 | 0.44 | 0.48 | $9.4 \times 10^{-14}$ |
| Gemini 3.5 Flash-Lite | Grok 4.20 | 20.00 | 15.20 | 20.36 | −4.80 | −5.16 | $4.8 \times 10^{-23}$ | $7.8 \times 10^{-137}$ |
| GPT-5.6 Luna | Claude Haiku 4.5 | 17.52 | 14.90 | 16.14 | −2.62 | −1.24 | $3.8 \times 10^{-7}$ | $1.2 \times 10^{-30}$ |
| GPT-5.6 Luna | Gemini 3.5 Flash-Lite | 17.52 | 18.90 | 19.04 | 1.38 | −0.14 | 0.13 | 0.55 |
| GPT-5.6 Luna | Qwen3-235B | 17.52 | 13.30 | 14.06 | −4.22 | −0.76 | 0.14 | $8.1 \times 10^{-14}$ |
| GPT-5.6 Luna | Grok 4.20 | 17.52 | 11.50 | 16.72 | −6.02 | −5.22 | $1.2 \times 10^{-32}$ | $1.3 \times 10^{-180}$ |
| Qwen3-235B | Claude Haiku 4.5 | 10.79 | 11.80 | 13.82 | 1.01 | −2.02 | 0.026 | 0.52 |
| Qwen3-235B | Gemini 3.5 Flash-Lite | 10.79 | 10.00 | 10.58 | −0.79 | −0.58 | 0.012 | $8.1 \times 10^{-6}$ |
| Qwen3-235B | GPT-5.6 Luna | 10.79 | 10.50 | 10.66 | −0.29 | −0.16 | 0.073 | 0.0017 |
| Qwen3-235B | Grok 4.20 | 10.79 | 13.40 | 19.72 | 2.61 | −6.32 | $1.9 \times 10^{-17}$ | $6.6 \times 10^{-93}$ |
| Grok 4.20 | Claude Haiku 4.5 | 5.88 | 24.40 | 7.12 | 18.52 | 17.28 | 1.0 | 1.0 |
| Grok 4.20 | Gemini 3.5 Flash-Lite | 5.88 | 25.00 | 8.24 | 19.12 | 16.76 | 1.0 | 1.0 |
| Grok 4.20 | GPT-5.6 Luna | 5.88 | 29.20 | 8.40 | 23.32 | 20.80 | 1.0 | 1.0 |
| Grok 4.20 | Qwen3-235B | 5.88 | 19.00 | 4.68 | 13.12 | 14.32 | 1.0 | 1.0 |

**Table S41:** Single entrant beside five residents at $p = 0.9$, as in Table S39.

| Resident $B$ | Entrant $A$ | $u_B^{\text{self}}$ | $u_A(1)$ | $u_B(5)$ | $\Delta_{\text{self}}$ | $\Delta_{\text{table}}$ | $\rho_6$ | $\rho_{30}$ |
|---|---|---|---|---|---|---|---|---|
| Claude Haiku 4.5 | Gemini 3.5 Flash-Lite | 18.85 | 21.20 | 19.00 | 2.35 | 2.20 | 0.88 | 0.90 |
| Claude Haiku 4.5 | GPT-5.6 Luna | 18.85 | 21.00 | 18.60 | 2.15 | 2.40 | 0.91 | 0.89 |
| Claude Haiku 4.5 | Qwen3-235B | 18.85 | 20.80 | 18.64 | 1.95 | 2.16 | 0.56 | $1.1 \times 10^{-22}$ |
| Claude Haiku 4.5 | Grok 4.20 | 18.85 | 14.80 | 18.88 | −4.05 | −4.08 | $1.5 \times 10^{-25}$ | $5.4 \times 10^{-153}$ |
| Gemini 3.5 Flash-Lite | Claude Haiku 4.5 | 20.00 | 18.80 | 19.92 | −1.20 | −1.12 | $9.8 \times 10^{-5}$ | $1.5 \times 10^{-22}$ |
| Gemini 3.5 Flash-Lite | GPT-5.6 Luna | 20.00 | 19.80 | 20.00 | −0.20 | −0.20 | 0.16 | $1.1 \times 10^{-6}$ |
| Gemini 3.5 Flash-Lite | Qwen3-235B | 20.00 | 9.52 | 9.16 | −10.48 | 0.36 | 0.26 | $4.1 \times 10^{-31}$ |
| Gemini 3.5 Flash-Lite | Grok 4.20 | 20.00 | 17.20 | 20.20 | −2.80 | −3.00 | $4.7 \times 10^{-34}$ | $1.4 \times 10^{-197}$ |
| GPT-5.6 Luna | Claude Haiku 4.5 | 16.41 | 11.18 | 12.95 | −5.23 | −1.77 | $1.2 \times 10^{-5}$ | $7.3 \times 10^{-21}$ |
| GPT-5.6 Luna | Gemini 3.5 Flash-Lite | 16.41 | 19.80 | 19.92 | 3.39 | −0.12 | 0.14 | 0.93 |
| GPT-5.6 Luna | Qwen3-235B | 16.41 | 7.26 | 7.47 | −9.15 | −0.21 | 0.17 | $1.0 \times 10^{-24}$ |
| GPT-5.6 Luna | Grok 4.20 | 16.41 | 7.40 | 15.62 | −9.01 | −8.22 | $1.1 \times 10^{-39}$ | $1.5 \times 10^{-219}$ |
| Qwen3-235B | Claude Haiku 4.5 | 2.15 | 4.72 | 5.85 | 2.57 | −1.13 | 0.021 | 0.87 |
| Qwen3-235B | Gemini 3.5 Flash-Lite | 2.15 | 2.00 | 2.18 | −0.15 | −0.18 | 0.11 | 0.071 |
| Qwen3-235B | GPT-5.6 Luna | 2.15 | 3.80 | 3.96 | 1.65 | −0.16 | 0.11 | 0.53 |
| Qwen3-235B | Grok 4.20 | 2.15 | 8.78 | 13.46 | 6.63 | −4.68 | $1.2 \times 10^{-21}$ | $1.4 \times 10^{-104}$ |
| Grok 4.20 | Claude Haiku 4.5 | 2.62 | 24.40 | 2.08 | 21.78 | 22.32 | 1.0 | 1.0 |
| Grok 4.20 | Gemini 3.5 Flash-Lite | 2.62 | 29.60 | 1.24 | 26.98 | 28.36 | 1.0 | 1.0 |
| Grok 4.20 | GPT-5.6 Luna | 2.62 | 30.60 | 5.80 | 27.98 | 24.80 | 1.0 | 1.0 |
| Grok 4.20 | Qwen3-235B | 2.62 | 18.60 | 3.48 | 15.98 | 15.12 | 1.0 | 1.0 |

**Table S42:** Winner of the stationary distribution (with its mass) for every candidate set, $\beta = 1$. *Winner differs from the full five-model set. Removing a candidate leaves every retained payoff unchanged and changes only the available models and mutation destinations.

| Omitted model | $p = 0.1$ | $p = 0.5$ | $p = 0.9$ |
|---|---|---|---|
| *One table, $N = 6$* | | | |
| None (all five) | Qwen3-235B (0.891) | Qwen3-235B (0.658) | Qwen3-235B (0.470) |
| Claude Haiku 4.5 | Qwen3-235B (0.894) | Qwen3-235B (0.709) | Qwen3-235B (0.490) |
| Gemini 3.5 Flash-Lite | Qwen3-235B (0.887) | Qwen3-235B (0.586) | Qwen3-235B (0.588) |
| GPT-5.6 Luna | Qwen3-235B (0.995) | Qwen3-235B (0.916) | Qwen3-235B (0.665) |
| Qwen3-235B | GPT-5.6 Luna (0.665)* | GPT-5.6 Luna (0.588)* | GPT-5.6 Luna (0.535)* |
| Grok 4.20 | Qwen3-235B (0.891) | Qwen3-235B (0.658) | Qwen3-235B (0.470) |
| *Larger population, $N = 30$* | | | |
| None (all five) | Qwen3-235B (1.000) | Gemini 3.5 Flash-Lite (1.000) | Gemini 3.5 Flash-Lite (1.000) |
| Claude Haiku 4.5 | Qwen3-235B (1.000) | Gemini 3.5 Flash-Lite (1.000) | Gemini 3.5 Flash-Lite (1.000) |
| Gemini 3.5 Flash-Lite | Qwen3-235B (1.000) | GPT-5.6 Luna (1.000)* | GPT-5.6 Luna (1.000)* |
| GPT-5.6 Luna | Qwen3-235B (1.000) | Gemini 3.5 Flash-Lite (1.000) | Gemini 3.5 Flash-Lite (1.000) |
| Qwen3-235B | GPT-5.6 Luna (0.909)* | Gemini 3.5 Flash-Lite (1.000) | Gemini 3.5 Flash-Lite (1.000) |
| Grok 4.20 | Qwen3-235B (1.000) | Gemini 3.5 Flash-Lite (1.000) | Gemini 3.5 Flash-Lite (1.000) |

**Table S43:** Target rate and payoff of the selected populations at $p = 0.5$, $\beta = 1$, where $\text{opt}(p) = 20$. Weighted: each retained model's self-play weighted by its stationary mass; uniform: equal weight on the retained models. These are rare-mutation one-model outcomes, not success rates during invasion. At $p = 0.5$ zero contribution and an exact target both earn 20, although their target rates are 0 and 100%, so a low target rate alone does not imply a payoff loss.

| | | Weighted by stationary mass | | | Uniform weight | |
|---|---|---|---|---|---|---|
| Omitted model | $N$ | Target (%) | Payoff per seat | Share of opt (%) | Target (%) | Share of opt (%) |
| None (all five) | 6 | 27.7 | 13.28 | 66.4 | 75.0 | 72.9 |
| None (all five) | 30 | 100.0 | 20.00 | 100.0 | 75.0 | 72.9 |
| Claude Haiku 4.5 | 6 | 23.3 | 12.90 | 64.5 | 68.8 | 67.8 |
| Claude Haiku 4.5 | 30 | 100.0 | 20.00 | 100.0 | 68.8 | 67.8 |
| Gemini 3.5 Flash-Lite | 6 | 31.4 | 13.60 | 68.0 | 68.8 | 66.2 |
| Gemini 3.5 Flash-Lite | 30 | 75.0 | 17.52 | 87.6 | 68.8 | 66.2 |
| GPT-5.6 Luna | 6 | 8.4 | 11.54 | 57.7 | 75.0 | 69.3 |
| GPT-5.6 Luna | 30 | 100.0 | 20.00 | 100.0 | 75.0 | 69.3 |
| Qwen3-235B | 6 | 85.3 | 18.54 | 92.7 | 93.8 | 77.7 |
| Qwen3-235B | 30 | 100.0 | 20.00 | 100.0 | 93.8 | 77.7 |
| Grok 4.20 | 6 | 27.7 | 13.28 | 66.4 | 68.8 | 83.8 |
| Grok 4.20 | 30 | 100.0 | 20.00 | 100.0 | 68.8 | 83.8 |

# S12 Strategies, the last round and equilibrium benchmarks

This section reports how the models play the last round, and supports three statements in Sections 4.1.1 and 4.2.2 of the paper: how the models compare with the equilibria of Proposition 1, how each model's moves can be described by a simple rule, and how partners respond to Grok 4.20's opening move. All tables use the games played with the base prompt.

## S12.1 The last round

The last move reveals three rules. Qwen3-235B's final-round skipping appears to be driven by the timing of the move rather than by the state of the collective pool. Across all base-prompt games, it skips the last contribution as often when its contribution is needed as when the others' fair shares already suffice (93% of seats in both cases), and it almost never stops earlier once the target is met. Human contributions similarly decline over the rounds of repeated public goods games (Andreoni, 1988). GPT-5.6 Luna skips in 2% of cases when the target needs its contribution but in 53% when the others' fair shares would suffice. When six copies each count on another seat to contribute, tables ahead of the fair pace of 12 per round often miss the target, while tables exactly on pace reach it. Claude Haiku 4.5 stops contributing in the final round mostly when the target has already been reached, showing that models can infer the current pool size from the prompt.

Table S44 splits the last move of every seat by the state of the pool before round 10. In self-play, GPT-5.6 Luna tables that reach round 10 ahead of the fair pace reach the target in 27% of 26 games, against 92% of 99 games exactly on pace: when every seat reasons that another will pay, too many skip.

**Table S44:** Share of seats paying nothing in round 10, by the state of the pool after round 9, and in rounds 7 to 9 once the pool has reached 120, with 95% bootstrap intervals over games in the table file.

| | Target already reached | | Others' 2s suffice | | Own 2 is pivotal | | Rounds 7 to 9, reached | |
|---|---|---|---|---|---|---|---|---|
| Model | Paid 0 (%) | Games | Paid 0 (%) | Games | Paid 0 (%) | Games | Paid 0 (%) | Games |
| Claude Haiku 4.5 | 92 [90, 94] | 337 | 13 [11, 14] | 413 | 9 [5, 13] | 135 | 90 [87, 93] | 106 |
| Gemini 3.5 Flash-Lite | 94 [91, 97] | 164 | 9 [7, 11] | 197 | 0 [0, 0] | 525 | 96 [93, 99] | 97 |
| GPT-5.6 Luna | 100 [100, 100] | 143 | 53 [50, 57] | 285 | 2 [1, 2] | 441 | 100 [100, 100] | 83 |
| Qwen3-235B | 80 [76, 84] | 221 | 93 [92, 95] | 307 | 93 [92, 95] | 361 | 1 [0, 2] | 98 |
| Grok 4.20 | 2 [1, 3] | 633 | 4 [2, 6] | 159 | 0 [0, 0] | 98 | 0 [0, 0] | 341 |
| *GPT-5.6 Luna self-play: target reached (%) by the pool after round 9* | | | | | | | | |
| Ahead of pace (110 to 118) | 27 [12, 42] | 26 | | | | | | |
| On pace (exactly 108) | 92 [86, 97] | 99 | | | | | | |

Top: share of decisions paying 0 with a 95% bootstrap interval over games, and the number of games. Round 10 is split by the pool the seats saw before deciding: at least 120 (reached); 110 to 118, so the other five paying 2 already reach 120 (others' 2s suffice); exactly 108 (own 2 is pivotal). The last pair of columns pools rounds 7 to 9 that start with the pool already at 120 or more. Base-prompt games: self-play, mixed tables and scripted partners.

## S12.2 Play against the equilibria of Proposition 1

For $p < 1$ the worst Nash equilibrium is that nobody pays, which leaves $(1-p)\,40$ per seat; every target-reaching equilibrium leaves 20. Table S45 counts the self-play games whose expected payoff per seat falls below the worst equilibrium. Above the pivot this happens in most games of Qwen3-235B, which misses the target after paying, and of Grok 4.20, which overpays, and never for Gemini 3.5 Flash-Lite or Claude Haiku 4.5.

**Table S45:** Self-play games on the risk grid whose expected payoff per seat is below that of the worst Nash equilibrium.

| | $0.5 < p < 1$ | | $p < 0.5$ | |
|---|---|---|---|---|
| Model | Games | Below (%) | Games | Below (%) |
| Claude Haiku 4.5 | 40 | 0 [0, 0] | 50 | 100 [100, 100] |
| Gemini 3.5 Flash-Lite | 40 | 0 [0, 0] | 50 | 100 [100, 100] |
| GPT-5.6 Luna | 40 | 25 [13, 38] | 50 | 80 [68, 90] |
| Qwen3-235B | 40 | 100 [100, 100] | 50 | 100 [100, 100] |
| Grok 4.20 | 40 | 93 [83, 100] | 50 | 100 [100, 100] |

Share of self-play games on the risk grid whose expected payoff per seat falls below that of the worst Nash equilibrium, with a 95% bootstrap interval over games. By Proposition 1, for $p < 1$ the equilibria are $x = 0$, worth $(1 - p)40$ per seat, and profiles with $X = T$ and $x_i \le 40p$, which exist only for $p \ge 0.5$ and are worth 20 per seat, never less than $(1 - p)40$. The worst equilibrium is therefore nobody pays, $(1 - p)40$, at every $p < 1$; below $p = 0.5$ it is also the only one. $p = 1$ is left out because no payoff can fall below 0 there, and $p = 0.5$ because the two kinds of equilibrium tie.

## S12.3 Simple rules describe the moves

Table S46 compares each model's moves with named rules and with a small decision tree fitted by cross-validation. For every model one rule predicts most decisions, and the tree adds little. Removing the risk level from the tree's inputs does not lower its accuracy for any model, which is the sense in which the paper calls the play a default. A rule close to the risk-neutral optimum predicts far fewer moves than the best rule for every model.

**Table S46:** How well named rules and a fitted tree predict each model's moves.

| Accuracy (%) | Claude Haiku 4.5 | Gemini 3.5 Flash-Lite | GPT-5.6 Luna | Qwen3 235B | Grok 4.20 |
|---|---|---|---|---|---|
| *Named rules (best per model in bold)* | | | | | |
| Always 2 | 81 [80, 82] | 96 [95, 97] | 89 [87, 90] | 86 [85, 86] | 41 [39, 43] |
| Repeat own move | 78 [77, 78] | 98 [98, 98] | **94 [93, 94]** | 88 [87, 88] | **86 [85, 86]** |
| Copy the group | 81 [80, 81] | 95 [94, 96] | 91 [90, 92] | 84 [84, 85] | 68 [67, 70] |
| Last-round skipper | 81 [81, 82] | 89 [88, 89] | 85 [84, 86] | **94 [94, 95]** | 41 [38, 43] |
| Stop when settled | **86 [86, 86]** | 99 [98, 99] | 92 [91, 93] | 87 [86, 87] | 41 [39, 43] |
| Pool tracker | 85 [84, 86] | **99 [99, 99]** | 92 [91, 93] | 87 [86, 87] | 41 [39, 44] |
| Risk-neutral step | 57 [54, 60] | 63 [60, 67] | 63 [60, 67] | 58 [55, 61] | 23 [21, 25] |
| *Baselines and fitted policy* | | | | | |
| Majority class | 81 [80, 82] | 96 [95, 97] | 89 [87, 90] | 86 [85, 86] | 58 [56, 61] |
| Majority move | 2 | 2 | 2 | 2 | 4 |
| Tree, cross-validated | 87 [86, 87] | 99 [99, 99] | 95 [95, 96] | 96 [96, 97] | 91 [91, 91] |
| Tree without risk $p$ | 87 [86, 87] | 99 [99, 99] | 95 [95, 96] | 96 [96, 97] | 91 [91, 91] |
| Change without $p$ (pp) | 0.0 [0.0, 0.0] | 0.0 [0.0, 0.0] | 0.0 [0.0, 0.0] | 0.0 [0.0, 0.0] | 0.1 [0.0, 0.2] |
| Tree minus best rule (pp) | 0.5 [0.3, 0.7] | 0.1 [0.0, 0.2] | 1.9 [1.7, 2.2] | 2.3 [1.9, 2.7] | 5.2 [4.8, 5.7] |
| Games | 940 | 940 | 940 | 940 | 940 |
| Decisions | 28,400 | 28,400 | 28,400 | 28,400 | 28,400 |

Accuracy is the share of a model's decisions a policy predicts exactly, with a 95% bootstrap interval over games; base-prompt games (self-play, mixed tables, scripted partners). Always 2. Repeat own move: last round's own move, 2 in round 1. Copy the group: the others' mean last round rounded to 0, 2 or 4, 2 in round 1. Last-round skipper: 2, then 0 in round 10. Stop when settled: 2 until the target is reached or out of reach, then 0. Pool tracker: 0 once settled, else the share each seat still needs per round rounded up to 2 or 4. Risk-neutral step: 0 at $p < 0.5$; at $p \geq 0.5$, 2 until the target is reached, then 0. Tree: depth 5, at least 50 decisions per leaf, five folds that keep each game in one fold; features are the own previous move, the others' previous mean, the pool against the fair pace, rounds left, a last-round flag, reached, out of reach, the amount each seat still needs per round, and $p$. Changes are paired over the same decisions.

## S12.4 Responses to Grok 4.20's opening move

Grok 4.20 chooses its first move before any history exists, and that move varies across repetitions of a cell. Table S47 relates a partner's payments to Grok 4.20's mean opening move within each cell. Higher openings are associated with lower payments by GPT-5.6 Luna while the pool is short. Claude Haiku 4.5's payment in that phase barely changes and Gemini 3.5 Flash-Lite's changes only slightly; both pay less in total mainly through earlier target completion. Qwen3-235B pays more. This is an observational association, not a randomised opening intervention or proof of an internal reasoning mechanism.

**Table S47:** Change in a partner seat's payment per unit of Grok 4.20's mean opening move, within cells of the two-model tables, with 95% bootstrap intervals over games.

| | | Partner's payment per seat | | |
|---|---|---|---|---|
| Partner model | Games | Total, rounds 2 to 10 | Per round, pool below 120 | Per round, pool at 120 or more |
| Claude Haiku 4.5 | 150 | −1.81 [−2.40, −1.40] | 0.00 [−0.05, 0.06] | −0.09 [−0.22, 0.03] (124) |
| Gemini 3.5 Flash-Lite | 150 | −2.00 [−2.44, −1.61] | −0.03 [−0.05, −0.01] | −0.19 [−0.49, 0.08] (88) |
| GPT-5.6 Luna | 150 | −3.69 [−4.34, −3.05] | −0.37 [−0.45, −0.30] | 0.00 [0.00, 0.00] (74) |
| Qwen3-235B | 150 | 0.78 [0.19, 1.39] | 0.03 [−0.06, 0.11] | 0.97 [0.44, 1.52] (113) |

Change in the partner's payment per seat for one more unit in the mean opening move of the Grok 4.20 seats, which is made before any history. Mixed tables of Grok 4.20 with one other model; OLS within cells of the same number of Grok seats and the same $p$, controlling for the partner's own opening move; 95% bootstrap interval resampling games within cells. Per-round columns average rounds 2 to 10 whose starting pool is below 120, or at 120 or more; the last column counts only games with such rounds (in brackets after the interval: games).

}{}

\end{document}

%% file: tables/num_selfplay.tex
\newcommand{\CntGames}{3{,}950}
\newcommand{\CntGamesBaseline}{550}
\newcommand{\CntGamesNoCue}{150}
\newcommand{\CntGamesProbe}{150}
\newcommand{\CntGamesScripted}{1{,}000}
\newcommand{\CntGamesMixed}{1{,}500}
\newcommand{\CntGamesPara}{200}
\newcommand{\CntGamesTempZero}{100}
\newcommand{\CntGamesNeutral}{150}
\newcommand{\CntGamesWording}{150}

\newcommand{\GridStepMaxOther}{0.2}

\newcommand{\GridOptMean}{25.5}

%% file: tables/num_scripted.tex
\newcommand{\ScrGames}{1{,}000}

\newcommand{\ScrOptimalGames}{186}

\newcommand{\ScrOpponentShareMin}{69}
\newcommand{\ScrOpponentShareMax}{92}
\newcommand{\ScrRiskShareMax}{5}

%% file: tables/num_mixed.tex
\newcommand{\MixFlashQwenReach}{11}

\newcommand{\MixQwenPotZeroPct}{92}
\newcommand{\MixFlashPotZeroPct}{0.1}
\newcommand{\MixMissWithQwen}{323}

\newcommand{\MixQwenPivotalFix}{92}
\newcommand{\MixLunaBesideFiveGrok}{10.5}

%% file: tables/num_selection.tex
\newcommand{\SelWithinReachHigh}{47}

\newcommand{\SelUniformReachHigh}{76}

%% file: tables/tab_models.tex
\begin{tabular}{@{}l@{\hspace{3pt}}r@{\hspace{3pt}}r@{\hspace{3pt}}r@{\hspace{3pt}}r@{\hspace{3pt}}r@{\hspace{3pt}}l@{}}
\toprule
 & Pays at & Exact & Target & Payoff & Lost & Main \\
Model & $p{=}0$ & split (\%) & (\%) & & (\%) & loss \\
\midrule
Claude Haiku 4.5 & 20.9 & 17 & 98 & 18.9 & 26 & below $p^*$ \\
Gemini 3.5 Flash-Lite & 20.0 & 99.7 & 100 & 20.0 & 21 & below $p^*$ \\
GPT-5.6 Luna & 0.0 & 72 & 70 & 19.2 & 25 & below $p^*$ \\
Qwen3-235B & 18.0 & 4 & 2 & 10.9 & 57 & missed \\
Grok 4.20 & 37.4 & 2 & 100 & 5.0 & 80 & past 120 \\
\bottomrule
\end{tabular}

%% file: tables/tab_robust.tex
\begin{tabular}{@{}lrrrrr@{}}
\toprule
 & Claude & Gemini 3.5 & GPT-5.6 & Qwen3 & Grok \\
Condition & Haiku 4.5 & Flash-Lite & Luna & 235B & 4.20 \\
\midrule
Base & 20.9 & 20.0 & 19.9 & 18.4 & 35.1 \\
Rerun & 21.3 & 20.0 & 19.9 & 18.5 & 35.6 \\
No hint & 21.7 & 20.0 & 19.9 & \textbf{34.7} & \textbf{40.0} \\
Para.\,1 & 21.3 & 20.0 & 19.9 & 19.4 & 31.8 \\
Para.\,2 & 20.9 & 20.0 & 19.8 & 18.0 & \textbf{20.6} \\
Temp.\,0 & 21.2 & 20.0 & 20.0 & 18.3 & 35.1 \\
Wording & 21.0 & 20.0 & 19.4 & 18.7 & \textbf{39.7} \\
Goal & 20.0 & 20.1 & \textbf{13.9} & \textbf{8.2} & 37.7 \\
\midrule
Base $p{=}0$ & 20.9 & 20.0 & 0.0 & 18.0 & 37.4 \\
Wording $p{=}0$ & 18.0 & 19.9 & 0.4 & 18.1 & 39.8 \\
Goal $p{=}0$ & 0.3 & 1.4 & 0.2 & 4.8 & 36.7 \\
Goal $p{=}0.1$ & 19.8 & 20.2 & 7.7 & 6.4 & 36.7 \\
Goal $p{=}0.9$ & 20.2 & 20.0 & 20.1 & 10.1 & 38.7 \\
\bottomrule
\end{tabular}

%% file: tables/tab_settled.tex
\begin{tabular}{@{}lrrrr@{}}
\toprule
 & \multicolumn{2}{c}{Always 0} & \multicolumn{2}{c}{Always 4} \\
\cmidrule(lr){2-3}\cmidrule(l){4-5}
Model & Before & After & Before & After \\
\midrule
Claude Haiku 4.5 & 7.7 & 0.1 & 14.2 & 1.1 \\
Gemini 3.5 Flash-Lite & 13.6 & 0.2 & 10.9 & 0.0 \\
GPT-5.6 Luna & 8.5 & 0.0 & 4.5 & 0.0 \\
Qwen3-235B & 17.8 & 12.1 & 12.3 & 8.1 \\
Grok 4.20 & 19.6 & 11.1 & 16.7 & 16.2 \\
\bottomrule
\end{tabular}

%% file: main.bbl

\begin{thebibliography}{61}


\ifx \showCODEN    \undefined \def \showCODEN     #1{\unskip}     \fi
\ifx \showISBNx    \undefined \def \showISBNx     #1{\unskip}     \fi
\ifx \showISBNxiii \undefined \def \showISBNxiii  #1{\unskip}     \fi
\ifx \showISSN     \undefined \def \showISSN      #1{\unskip}     \fi
\ifx \showLCCN     \undefined \def \showLCCN      #1{\unskip}     \fi
\ifx \shownote     \undefined \def \shownote      #1{#1}          \fi
\ifx \showarticletitle \undefined \def \showarticletitle #1{#1}   \fi
\ifx \showURL      \undefined \def \showURL       {\relax}        \fi
\providecommand\bibfield[2]{#2}
\providecommand\bibinfo[2]{#2}
\providecommand\natexlab[1]{#1}
\providecommand\showeprint[2][]{arXiv:#2}

\bibitem[Abou~Chakra and Traulsen(2012)]%
        {abouchakra2012evolutionary}
\bibfield{author}{\bibinfo{person}{Maria Abou~Chakra} {and}
  \bibinfo{person}{Arne Traulsen}.} \bibinfo{year}{2012}\natexlab{}.
\newblock \showarticletitle{Evolutionary dynamics of strategic behavior in a
  collective-risk dilemma}.
\newblock \bibinfo{journal}{\emph{PLoS Computational Biology}}
  \bibinfo{volume}{8}, \bibinfo{number}{8} (\bibinfo{year}{2012}),
  \bibinfo{pages}{e1002652}.
\newblock
\href{https://doi.org/10.1371/journal.pcbi.1002652}{doi:\nolinkurl{10.1371/journal.pcbi.1002652}}


\bibitem[Akata et~al\mbox{.}(2025)]%
        {akata2025playing}
\bibfield{author}{\bibinfo{person}{Elif Akata}, \bibinfo{person}{Lion Schulz},
  \bibinfo{person}{Julian Coda-Forno}, \bibinfo{person}{Seong~Joon Oh},
  \bibinfo{person}{Matthias Bethge}, {and} \bibinfo{person}{Eric Schulz}.}
  \bibinfo{year}{2025}\natexlab{}.
\newblock \showarticletitle{Playing repeated games with large language models}.
\newblock \bibinfo{journal}{\emph{Nature Human Behaviour}} \bibinfo{volume}{9},
  \bibinfo{number}{7} (\bibinfo{year}{2025}), \bibinfo{pages}{1380--1390}.
\newblock
\href{https://doi.org/10.1038/s41562-025-02172-y}{doi:\nolinkurl{10.1038/s41562-025-02172-y}}


\bibitem[Andreoni(1988)]%
        {andreoni1988free}
\bibfield{author}{\bibinfo{person}{James Andreoni}.}
  \bibinfo{year}{1988}\natexlab{}.
\newblock \showarticletitle{Why free ride? Strategies and learning in public
  goods experiments}.
\newblock \bibinfo{journal}{\emph{Journal of Public Economics}}
  \bibinfo{volume}{37}, \bibinfo{number}{3} (\bibinfo{year}{1988}),
  \bibinfo{pages}{291--304}.
\newblock
\href{https://doi.org/10.1016/0047-2727(88)90043-6}{doi:\nolinkurl{10.1016/0047-2727(88)90043-6}}


\bibitem[Buscemi et~al\mbox{.}(2025a)]%
        {buscemi2025llms}
\bibfield{author}{\bibinfo{person}{Alessio Buscemi}, \bibinfo{person}{Daniele
  Proverbio}, \bibinfo{person}{Paolo Bova}, \bibinfo{person}{Nataliya
  Balabanova}, \bibinfo{person}{Adeela Bashir}, \bibinfo{person}{Theodor
  Cimpeanu}, \bibinfo{person}{Henrique Correia~da Fonseca},
  \bibinfo{person}{Manh~Hong Duong}, \bibinfo{person}{Elias
  Fern{\'a}ndez~Domingos}, \bibinfo{person}{Antonio~M. Fernandes},
  \bibinfo{person}{Marcus Krellner}, \bibinfo{person}{Ndidi~Bianca Ogbo},
  \bibinfo{person}{Simon~T. Powers}, \bibinfo{person}{Fernando~P. Santos},
  \bibinfo{person}{Zia~Ush Shamszaman}, \bibinfo{person}{Zhao Song},
  \bibinfo{person}{Alessandro Di~Stefano}, {and} \bibinfo{person}{The~Anh
  Han}.} \bibinfo{year}{2025}\natexlab{a}.
\newblock \bibinfo{title}{Do {LLMs} trust {AI} regulation? Emerging behaviour
  of game-theoretic {LLM} agents}.
\newblock
\showeprint[arxiv]{2504.08640}
\href{https://doi.org/10.48550/arXiv.2504.08640}{doi:\nolinkurl{10.48550/arXiv.2504.08640}}


\bibitem[Buscemi et~al\mbox{.}(2025b)]%
        {buscemi2025fairgame}
\bibfield{author}{\bibinfo{person}{Alessio Buscemi}, \bibinfo{person}{Daniele
  Proverbio}, \bibinfo{person}{Alessandro Di~Stefano}, \bibinfo{person}{The~Anh
  Han}, \bibinfo{person}{German Castignani}, {and} \bibinfo{person}{Pietro
  Li{\`o}}.} \bibinfo{year}{2025}\natexlab{b}.
\newblock \showarticletitle{{FAIRGAME}: A Framework for {AI} Agents Bias
  Recognition Using Game Theory}. In \bibinfo{booktitle}{\emph{{ECAI} 2025 --
  28th European Conference on Artificial Intelligence}}
  \emph{(\bibinfo{series}{Frontiers in Artificial Intelligence and
  Applications}, Vol.~\bibinfo{volume}{413})}. \bibinfo{publisher}{IOS Press},
  \bibinfo{address}{Bologna, Italy}, \bibinfo{pages}{4097--4104}.
\newblock
\href{https://doi.org/10.3233/FAIA251300}{doi:\nolinkurl{10.3233/FAIA251300}}


\bibitem[Celiktemel et~al\mbox{.}(2026)]%
        {celiktemel2026group}
\bibfield{author}{\bibinfo{person}{Luis Celiktemel}, \bibinfo{person}{Edward
  Eichhorn}, \bibinfo{person}{Levin Brinkmann}, \bibinfo{person}{Robin
  Schimmelpfennig}, \bibinfo{person}{Aron Vallinder}, \bibinfo{person}{Yaomin
  Jiang}, \bibinfo{person}{Edward Hughes}, {and} \bibinfo{person}{Iyad
  Rahwan}.} \bibinfo{year}{2026}\natexlab{}.
\newblock \bibinfo{title}{Group Selection Promotes Prosocial Prompts in
  Populations of {LLM} Agents}.
\newblock
\showeprint[arxiv]{2606.23343}
\href{https://doi.org/10.48550/arXiv.2606.23343}{doi:\nolinkurl{10.48550/arXiv.2606.23343}}


\bibitem[Chan et~al\mbox{.}(2024)]%
        {chan2024visibility}
\bibfield{author}{\bibinfo{person}{Alan Chan}, \bibinfo{person}{Carson Ezell},
  \bibinfo{person}{Max Kaufmann}, \bibinfo{person}{Kevin Wei},
  \bibinfo{person}{Lewis Hammond}, \bibinfo{person}{Herbie Bradley},
  \bibinfo{person}{Emma Bluemke}, \bibinfo{person}{Nitarshan Rajkumar},
  \bibinfo{person}{David Krueger}, \bibinfo{person}{Noam Kolt},
  \bibinfo{person}{Lennart Heim}, {and} \bibinfo{person}{Markus Anderljung}.}
  \bibinfo{year}{2024}\natexlab{}.
\newblock \showarticletitle{Visibility into {AI} Agents}. In
  \bibinfo{booktitle}{\emph{Proceedings of the 2024 {ACM} Conference on
  Fairness, Accountability, and Transparency ({FAccT})}}.
  \bibinfo{publisher}{ACM}, \bibinfo{address}{Rio de Janeiro, Brazil},
  \bibinfo{pages}{958--973}.
\newblock
\href{https://doi.org/10.1145/3630106.3658948}{doi:\nolinkurl{10.1145/3630106.3658948}}


\bibitem[Chen et~al\mbox{.}(2023)]%
        {chen2023emergence}
\bibfield{author}{\bibinfo{person}{Yiting Chen}, \bibinfo{person}{Tracy~Xiao
  Liu}, \bibinfo{person}{You Shan}, {and} \bibinfo{person}{Songfa Zhong}.}
  \bibinfo{year}{2023}\natexlab{}.
\newblock \showarticletitle{The emergence of economic rationality of {GPT}}.
\newblock \bibinfo{journal}{\emph{Proceedings of the National Academy of
  Sciences}} \bibinfo{volume}{120}, \bibinfo{number}{51}
  (\bibinfo{year}{2023}), \bibinfo{pages}{e2316205120}.
\newblock
\href{https://doi.org/10.1073/pnas.2316205120}{doi:\nolinkurl{10.1073/pnas.2316205120}}


\bibitem[Dafoe et~al\mbox{.}(2020)]%
        {dafoe2020open}
\bibfield{author}{\bibinfo{person}{Allan Dafoe}, \bibinfo{person}{Edward
  Hughes}, \bibinfo{person}{Yoram Bachrach}, \bibinfo{person}{Tantum Collins},
  \bibinfo{person}{Kevin~R. McKee}, \bibinfo{person}{Joel~Z. Leibo},
  \bibinfo{person}{Kate Larson}, {and} \bibinfo{person}{Thore Graepel}.}
  \bibinfo{year}{2020}\natexlab{}.
\newblock \bibinfo{title}{Open Problems in Cooperative {AI}}.
\newblock
\showeprint[arxiv]{2012.08630}
\href{https://doi.org/10.48550/arXiv.2012.08630}{doi:\nolinkurl{10.48550/arXiv.2012.08630}}


\bibitem[Dodivers and Rafai(2025)]%
        {dodivers2025uncovering}
\bibfield{author}{\bibinfo{person}{Eleonore Dodivers} {and}
  \bibinfo{person}{Ismael Rafai}.} \bibinfo{year}{2025}\natexlab{}.
\newblock \showarticletitle{Uncovering the fairness of {AI}: exploring focal
  point, inequality aversion, and altruism in {ChatGPT}'s dictator game
  decisions}.
\newblock \bibinfo{journal}{\emph{Economics Bulletin}} \bibinfo{volume}{45},
  \bibinfo{number}{4} (\bibinfo{year}{2025}), \bibinfo{pages}{1818--1825}.
\newblock
\urldef\tempurl%
\url{https://www.accessecon.com/pubs/EB/default.aspx?PaperID=EB-25-00292&topic=Abstract}
\showURL{%
\tempurl}


\bibitem[Duong et~al\mbox{.}(2026)]%
        {duong2026cost}
\bibfield{author}{\bibinfo{person}{M.~H. Duong}, \bibinfo{person}{C.~M.
  Durbac}, {and} \bibinfo{person}{T.~A. Han}.} \bibinfo{year}{2026}\natexlab{}.
\newblock \bibinfo{title}{Cost of Institutional Incentives for Promoting
  Cooperation in $2\times 2$ Games and Collective Risk Games}.
\newblock \bibinfo{howpublished}{Dynamic Games and Applications}.
\newblock
\shownote{Online first, 7 February 2026}.
\newblock
\href{https://doi.org/10.1007/s13235-025-00696-6}{doi:\nolinkurl{10.1007/s13235-025-00696-6}}


\bibitem[Fan et~al\mbox{.}(2024)]%
        {fan2024can}
\bibfield{author}{\bibinfo{person}{Caoyun Fan}, \bibinfo{person}{Jindou Chen},
  \bibinfo{person}{Yaohui Jin}, {and} \bibinfo{person}{Hao He}.}
  \bibinfo{year}{2024}\natexlab{}.
\newblock \showarticletitle{Can {Large Language Models} Serve as Rational
  Players in {Game Theory}? A Systematic Analysis}. In
  \bibinfo{booktitle}{\emph{Proceedings of the AAAI Conference on Artificial
  Intelligence}}, Vol.~\bibinfo{volume}{38}. \bibinfo{publisher}{AAAI Press},
  \bibinfo{address}{Washington, DC, USA}, \bibinfo{pages}{17960--17967}.
\newblock
\href{https://doi.org/10.1609/aaai.v38i16.29751}{doi:\nolinkurl{10.1609/aaai.v38i16.29751}}


\bibitem[Fern{\'a}ndez~Domingos et~al\mbox{.}(2023)]%
        {fernandezdomingos2023egttools}
\bibfield{author}{\bibinfo{person}{Elias Fern{\'a}ndez~Domingos},
  \bibinfo{person}{Francisco~C. Santos}, {and} \bibinfo{person}{Tom Lenaerts}.}
  \bibinfo{year}{2023}\natexlab{}.
\newblock \showarticletitle{{EGTtools}: Evolutionary game dynamics in
  {Python}}.
\newblock \bibinfo{journal}{\emph{iScience}} \bibinfo{volume}{26},
  \bibinfo{number}{4} (\bibinfo{year}{2023}), \bibinfo{pages}{106419}.
\newblock
\href{https://doi.org/10.1016/j.isci.2023.106419}{doi:\nolinkurl{10.1016/j.isci.2023.106419}}


\bibitem[Fern{\'a}ndez~Domingos et~al\mbox{.}(2022)]%
        {fernandezdomingos2022delegation}
\bibfield{author}{\bibinfo{person}{Elias Fern{\'a}ndez~Domingos},
  \bibinfo{person}{In{\^e}s Terrucha}, \bibinfo{person}{R{\'e}mi Suchon},
  \bibinfo{person}{Jelena Gruji{\'c}}, \bibinfo{person}{Juan~C. Burguillo},
  \bibinfo{person}{Francisco~C. Santos}, {and} \bibinfo{person}{Tom Lenaerts}.}
  \bibinfo{year}{2022}\natexlab{}.
\newblock \showarticletitle{Delegation to artificial agents fosters prosocial
  behaviors in the collective risk dilemma}.
\newblock \bibinfo{journal}{\emph{Scientific Reports}}  \bibinfo{volume}{12}
  (\bibinfo{year}{2022}), \bibinfo{pages}{8492}.
\newblock
\href{https://doi.org/10.1038/s41598-022-11518-9}{doi:\nolinkurl{10.1038/s41598-022-11518-9}}


\bibitem[Filippas et~al\mbox{.}(2024)]%
        {horton2023large}
\bibfield{author}{\bibinfo{person}{Apostolos Filippas},
  \bibinfo{person}{John~J. Horton}, {and} \bibinfo{person}{Benjamin~S.
  Manning}.} \bibinfo{year}{2024}\natexlab{}.
\newblock \showarticletitle{Large Language Models as Simulated Economic Agents:
  What Can We Learn from {Homo Silicus}?}. In
  \bibinfo{booktitle}{\emph{Proceedings of the 25th ACM Conference on Economics
  and Computation}} \emph{(\bibinfo{series}{EC '24})}.
  \bibinfo{publisher}{ACM}, \bibinfo{address}{New Haven, CT, USA},
  \bibinfo{pages}{614--615}.
\newblock
\href{https://doi.org/10.1145/3670865.3673513}{doi:\nolinkurl{10.1145/3670865.3673513}}


\bibitem[Fontana et~al\mbox{.}(2025)]%
        {fontana2024nicer}
\bibfield{author}{\bibinfo{person}{Nicol{\`o} Fontana},
  \bibinfo{person}{Francesco Pierri}, {and} \bibinfo{person}{Luca~Maria
  Aiello}.} \bibinfo{year}{2025}\natexlab{}.
\newblock \showarticletitle{Nicer than Humans: How Do Large Language Models
  Behave in the Prisoner's Dilemma?}
\newblock \bibinfo{journal}{\emph{Proceedings of the International AAAI
  Conference on Web and Social Media}}  \bibinfo{volume}{19}
  (\bibinfo{year}{2025}), \bibinfo{pages}{522--535}.
\newblock
\href{https://doi.org/10.1609/icwsm.v19i1.35829}{doi:\nolinkurl{10.1609/icwsm.v19i1.35829}}


\bibitem[Fudenberg and Imhof(2006)]%
        {fudenberg2006imitation}
\bibfield{author}{\bibinfo{person}{Drew Fudenberg} {and}
  \bibinfo{person}{Lorens~A. Imhof}.} \bibinfo{year}{2006}\natexlab{}.
\newblock \showarticletitle{Imitation processes with small mutations}.
\newblock \bibinfo{journal}{\emph{Journal of Economic Theory}}
  \bibinfo{volume}{131}, \bibinfo{number}{1} (\bibinfo{year}{2006}),
  \bibinfo{pages}{251--262}.
\newblock
\href{https://doi.org/10.1016/j.jet.2005.04.006}{doi:\nolinkurl{10.1016/j.jet.2005.04.006}}


\bibitem[Guzman~Piedrahita et~al\mbox{.}(2025)]%
        {guzmanpiedrahita2025corrupted}
\bibfield{author}{\bibinfo{person}{David Guzman~Piedrahita},
  \bibinfo{person}{Yongjin Yang}, \bibinfo{person}{Mrinmaya Sachan},
  \bibinfo{person}{Giorgia Ramponi}, \bibinfo{person}{Bernhard Sch{\"o}lkopf},
  {and} \bibinfo{person}{Zhijing Jin}.} \bibinfo{year}{2025}\natexlab{}.
\newblock \showarticletitle{Corrupted by Reasoning: Reasoning Language Models
  Become Free-Riders in Public Goods Games}. In
  \bibinfo{booktitle}{\emph{Conference on Language Modeling (COLM)}}.
  \bibinfo{publisher}{OpenReview.net}, \bibinfo{address}{Montreal, Canada},
  \bibinfo{numpages}{37}~pages.
\newblock
\urldef\tempurl%
\url{https://openreview.net/forum?id=kH6LOHGjEl}
\showURL{%
\tempurl}


\bibitem[Hammond et~al\mbox{.}(2025)]%
        {hammond2025multi}
\bibfield{author}{\bibinfo{person}{Lewis Hammond}, \bibinfo{person}{Alan Chan},
  \bibinfo{person}{Jesse Clifton}, \bibinfo{person}{Jason Hoelscher-Obermaier},
  \bibinfo{person}{Akbir Khan}, \bibinfo{person}{Euan McLean},
  \bibinfo{person}{Chandler Smith}, \bibinfo{person}{Wolfram Barfuss},
  \bibinfo{person}{Jakob Foerster}, \bibinfo{person}{Tom{\'a}{\v{s}}
  Gaven{\v{c}}iak}, \bibinfo{person}{The~Anh Han}, \bibinfo{person}{Edward
  Hughes}, \bibinfo{person}{Vojt{\v{e}}ch Kova{\v{r}}{\'\i}k},
  \bibinfo{person}{Jan Kulveit}, \bibinfo{person}{Joel~Z. Leibo},
  \bibinfo{person}{Caspar Oesterheld}, \bibinfo{person}{Christian Schroeder~de
  Witt}, \bibinfo{person}{Nisarg Shah}, \bibinfo{person}{Michael Wellman},
  \bibinfo{person}{Paolo Bova}, \bibinfo{person}{Theodor Cimpeanu},
  \bibinfo{person}{Carson Ezell}, \bibinfo{person}{Quentin Feuillade-Montixi},
  \bibinfo{person}{Matija Franklin}, \bibinfo{person}{Esben Kran},
  \bibinfo{person}{Igor Krawczuk}, \bibinfo{person}{Max Lamparth},
  \bibinfo{person}{Niklas Lauffer}, \bibinfo{person}{Alexander Meinke},
  \bibinfo{person}{Sumeet Motwani}, \bibinfo{person}{Anka Reuel},
  \bibinfo{person}{Vincent Conitzer}, \bibinfo{person}{Michael Dennis},
  \bibinfo{person}{Iason Gabriel}, \bibinfo{person}{Adam Gleave},
  \bibinfo{person}{Gillian Hadfield}, \bibinfo{person}{Nika Haghtalab},
  \bibinfo{person}{Atoosa Kasirzadeh}, \bibinfo{person}{S{\'e}bastien Krier},
  \bibinfo{person}{Kate Larson}, \bibinfo{person}{Joel Lehman},
  \bibinfo{person}{David~C. Parkes}, \bibinfo{person}{Georgios Piliouras},
  {and} \bibinfo{person}{Iyad Rahwan}.} \bibinfo{year}{2025}\natexlab{}.
\newblock \bibinfo{title}{Multi-Agent Risks from Advanced {AI}}.
\newblock
\shownote{Cooperative AI Foundation, Technical Report \#1}.
\newblock
\showeprint[arxiv]{2502.14143}
\href{https://doi.org/10.48550/arXiv.2502.14143}{doi:\nolinkurl{10.48550/arXiv.2502.14143}}


\bibitem[Han(2022)]%
        {han2022emergent}
\bibfield{author}{\bibinfo{person}{The~Anh Han}.}
  \bibinfo{year}{2022}\natexlab{}.
\newblock \showarticletitle{Emergent behaviours in multi-agent systems with
  Evolutionary Game Theory}.
\newblock \bibinfo{journal}{\emph{AI Communications}} \bibinfo{volume}{35},
  \bibinfo{number}{4} (\bibinfo{year}{2022}), \bibinfo{pages}{327--337}.
\newblock
\href{https://doi.org/10.3233/AIC-220104}{doi:\nolinkurl{10.3233/AIC-220104}}


\bibitem[Han et~al\mbox{.}(2024)]%
        {han2024evolutionary}
\bibfield{author}{\bibinfo{person}{The~Anh Han}, \bibinfo{person}{Manh~Hong
  Duong}, {and} \bibinfo{person}{Matja{\v{z}} Perc}.}
  \bibinfo{year}{2024}\natexlab{}.
\newblock \showarticletitle{Evolutionary mechanisms that promote cooperation
  may not promote social welfare}.
\newblock \bibinfo{journal}{\emph{Journal of the Royal Society Interface}}
  \bibinfo{volume}{21}, \bibinfo{number}{220} (\bibinfo{year}{2024}),
  \bibinfo{pages}{20240547}.
\newblock
\href{https://doi.org/10.1098/rsif.2024.0547}{doi:\nolinkurl{10.1098/rsif.2024.0547}}


\bibitem[Han et~al\mbox{.}(2017)]%
        {han2017evolution}
\bibfield{author}{\bibinfo{person}{The~Anh Han},
  \bibinfo{person}{Lu{\'\i}s~Moniz Pereira}, {and} \bibinfo{person}{Tom
  Lenaerts}.} \bibinfo{year}{2017}\natexlab{}.
\newblock \showarticletitle{Evolution of commitment and level of participation
  in public goods games}.
\newblock \bibinfo{journal}{\emph{Autonomous Agents and Multi-Agent Systems}}
  \bibinfo{volume}{31}, \bibinfo{number}{3} (\bibinfo{year}{2017}),
  \bibinfo{pages}{561--583}.
\newblock
\href{https://doi.org/10.1007/s10458-016-9338-4}{doi:\nolinkurl{10.1007/s10458-016-9338-4}}


\bibitem[Han et~al\mbox{.}(2012)]%
        {han2012emergence}
\bibfield{author}{\bibinfo{person}{The~Anh Han},
  \bibinfo{person}{Lu{\'i}s~Moniz Pereira}, {and} \bibinfo{person}{Francisco~C.
  Santos}.} \bibinfo{year}{2012}\natexlab{}.
\newblock \showarticletitle{The emergence of commitments and cooperation}. In
  \bibinfo{booktitle}{\emph{Proceedings of the 11th International Conference on
  Autonomous Agents and Multiagent Systems (AAMAS 2012)}}.
  \bibinfo{publisher}{International Foundation for Autonomous Agents and
  Multiagent Systems}, \bibinfo{address}{Richland, SC},
  \bibinfo{pages}{559--566}.
\newblock


\bibitem[Huang et~al\mbox{.}(2025)]%
        {huang2025competing}
\bibfield{author}{\bibinfo{person}{Jen-tse Huang}, \bibinfo{person}{Eric~John
  Li}, \bibinfo{person}{Man~Ho Lam}, \bibinfo{person}{Tian Liang},
  \bibinfo{person}{Wenxuan Wang}, \bibinfo{person}{Youliang Yuan},
  \bibinfo{person}{Wenxiang Jiao}, \bibinfo{person}{Xing Wang},
  \bibinfo{person}{Zhaopeng Tu}, {and} \bibinfo{person}{Michael~R. Lyu}.}
  \bibinfo{year}{2025}\natexlab{}.
\newblock \showarticletitle{Competing Large Language Models in Multi-Agent
  Gaming Environments}. In \bibinfo{booktitle}{\emph{International Conference
  on Learning Representations (ICLR)}}. \bibinfo{publisher}{OpenReview.net},
  \bibinfo{address}{Singapore}, \bibinfo{pages}{65639--65681}.
\newblock
\urldef\tempurl%
\url{https://proceedings.iclr.cc/paper_files/paper/2025/hash/a46adbe2f0ca0e16ef8857e188991ad7-Abstract-Conference.html}
\showURL{%
\tempurl}


\bibitem[Hughes et~al\mbox{.}(2018)]%
        {hughes2018inequity}
\bibfield{author}{\bibinfo{person}{Edward Hughes}, \bibinfo{person}{Joel~Z.
  Leibo}, \bibinfo{person}{Matthew Phillips}, \bibinfo{person}{Karl Tuyls},
  \bibinfo{person}{Edgar Due{\~n}ez-Guzman}, \bibinfo{person}{Antonio
  Garc{\'\i}a~Casta{\~n}eda}, \bibinfo{person}{Iain Dunning},
  \bibinfo{person}{Tina Zhu}, \bibinfo{person}{Kevin McKee},
  \bibinfo{person}{Raphael Koster}, \bibinfo{person}{Heather Roff}, {and}
  \bibinfo{person}{Thore Graepel}.} \bibinfo{year}{2018}\natexlab{}.
\newblock \showarticletitle{Inequity Aversion Improves Cooperation in
  Intertemporal Social Dilemmas}. In \bibinfo{booktitle}{\emph{Advances in
  Neural Information Processing Systems}}, Vol.~\bibinfo{volume}{31}.
  \bibinfo{publisher}{Curran Associates, Inc.}, \bibinfo{address}{Red Hook, NY,
  USA}, \bibinfo{pages}{3330--3340}.
\newblock
\urldef\tempurl%
\url{https://proceedings.neurips.cc/paper_files/paper/2018/hash/7fea637fd6d02b8f0adf6f7dc36aed93-Abstract.html}
\showURL{%
\tempurl}


\bibitem[Huynh et~al\mbox{.}(2025)]%
        {huynh2025understanding}
\bibfield{author}{\bibinfo{person}{Trung-Kiet Huynh}, \bibinfo{person}{Duy-Minh
  Dao-Sy}, \bibinfo{person}{Thanh-Bang Cao}, \bibinfo{person}{Phong-Hao Le},
  \bibinfo{person}{Hong-Dan Nguyen}, \bibinfo{person}{Phu-Quy Nguyen-Lam},
  \bibinfo{person}{Minh-Luan Nguyen-Vo}, \bibinfo{person}{Hong-Phat Pham},
  \bibinfo{person}{Phu-Hoa Pham}, \bibinfo{person}{Thien-Kim Than},
  \bibinfo{person}{Chi-Nguyen Tran}, \bibinfo{person}{Huy Tran},
  \bibinfo{person}{Gia-Thoai Tran-Le}, \bibinfo{person}{Alessio Buscemi},
  \bibinfo{person}{Hong~Trang Le}, {and} \bibinfo{person}{The~Anh Han}.}
  \bibinfo{year}{2025}\natexlab{}.
\newblock \bibinfo{title}{Understanding {LLM} Agent Behaviours via Game Theory:
  Strategy Recognition, Biases and Multi-Agent Dynamics}.
\newblock
\showeprint[arxiv]{2512.07462}
\href{https://doi.org/10.48550/arXiv.2512.07462}{doi:\nolinkurl{10.48550/arXiv.2512.07462}}


\bibitem[Huynh et~al\mbox{.}(2026)]%
        {huynh2026payoff}
\bibfield{author}{\bibinfo{person}{Trung-Kiet Huynh}, \bibinfo{person}{Duy-Minh
  Dao-Sy}, \bibinfo{person}{Thanh-Bang Cao}, \bibinfo{person}{Phong-Hao Le},
  \bibinfo{person}{Hong-Dan Nguyen}, \bibinfo{person}{Phu-Quy Nguyen-Lam},
  \bibinfo{person}{Minh-Luan Nguyen-Vo}, \bibinfo{person}{Hong-Phat Pham},
  \bibinfo{person}{Phu-Hoa Pham}, \bibinfo{person}{Thien-Kim Than},
  \bibinfo{person}{Chi-Nguyen Tran}, \bibinfo{person}{Huy Tran},
  \bibinfo{person}{Gia-Thoai Tran-Le}, \bibinfo{person}{Alessio Buscemi},
  \bibinfo{person}{Hong~Trang Le}, {and} \bibinfo{person}{The~Anh Han}.}
  \bibinfo{year}{2026}\natexlab{}.
\newblock \bibinfo{title}{Payoff scaling shapes cooperation in {LLM} agents
  across languages}.
\newblock
\showeprint[arxiv]{2601.19082}
\href{https://doi.org/10.48550/arXiv.2601.19082}{doi:\nolinkurl{10.48550/arXiv.2601.19082}}


\bibitem[Kempinski et~al\mbox{.}(2025)]%
        {kempinski2025game}
\bibfield{author}{\bibinfo{person}{Benjamin Kempinski}, \bibinfo{person}{Ian
  Gemp}, \bibinfo{person}{Kate Larson}, \bibinfo{person}{Marc Lanctot},
  \bibinfo{person}{Yoram Bachrach}, {and} \bibinfo{person}{Tal Kachman}.}
  \bibinfo{year}{2025}\natexlab{}.
\newblock \showarticletitle{Game of Thoughts: Iterative Reasoning in
  Game-Theoretic Domains with Large Language Models}. In
  \bibinfo{booktitle}{\emph{Proceedings of the 24th International Conference on
  Autonomous Agents and Multiagent Systems (AAMAS 2025)}}.
  \bibinfo{publisher}{IFAAMAS}, \bibinfo{address}{Detroit, MI, USA},
  \bibinfo{pages}{1088--1097}.
\newblock
\href{https://doi.org/10.65109/gzfu8152}{doi:\nolinkurl{10.65109/gzfu8152}}


\bibitem[Kumar et~al\mbox{.}(2026)]%
        {kumar2026cooperative}
\bibfield{author}{\bibinfo{person}{Shivani Kumar}, \bibinfo{person}{Adarsh
  Bharathwaj}, {and} \bibinfo{person}{David Jurgens}.}
  \bibinfo{year}{2026}\natexlab{}.
\newblock \bibinfo{title}{Cooperative Profiles Predict Multi-Agent {LLM} Team
  Performance in {AI} for Science Workflows}.
\newblock
\showeprint[arxiv]{2604.20658}
\href{https://doi.org/10.48550/arXiv.2604.20658}{doi:\nolinkurl{10.48550/arXiv.2604.20658}}


\bibitem[Leibo et~al\mbox{.}(2021)]%
        {leibo2021scalable}
\bibfield{author}{\bibinfo{person}{Joel~Z. Leibo}, \bibinfo{person}{Edgar~A.
  Due{\~n}ez-Guzman}, \bibinfo{person}{Alexander Vezhnevets},
  \bibinfo{person}{John~P. Agapiou}, \bibinfo{person}{Peter Sunehag},
  \bibinfo{person}{Raphael Koster}, \bibinfo{person}{Jayd Matyas},
  \bibinfo{person}{Charlie Beattie}, \bibinfo{person}{Igor Mordatch}, {and}
  \bibinfo{person}{Thore Graepel}.} \bibinfo{year}{2021}\natexlab{}.
\newblock \showarticletitle{Scalable Evaluation of Multi-Agent Reinforcement
  Learning with {Melting Pot}}. In \bibinfo{booktitle}{\emph{Proceedings of the
  38th International Conference on Machine Learning}}
  \emph{(\bibinfo{series}{Proceedings of Machine Learning Research},
  Vol.~\bibinfo{volume}{139})}. \bibinfo{publisher}{PMLR},
  \bibinfo{address}{Virtual Event}, \bibinfo{pages}{6187--6199}.
\newblock
\urldef\tempurl%
\url{https://proceedings.mlr.press/v139/leibo21a.html}
\showURL{%
\tempurl}


\bibitem[Leibo et~al\mbox{.}(2017)]%
        {leibo2017multi}
\bibfield{author}{\bibinfo{person}{Joel~Z. Leibo}, \bibinfo{person}{Vinicius
  Zambaldi}, \bibinfo{person}{Marc Lanctot}, \bibinfo{person}{Janusz Marecki},
  {and} \bibinfo{person}{Thore Graepel}.} \bibinfo{year}{2017}\natexlab{}.
\newblock \showarticletitle{Multi-agent Reinforcement Learning in Sequential
  Social Dilemmas}. In \bibinfo{booktitle}{\emph{Proceedings of the 16th
  International Conference on Autonomous Agents and Multiagent Systems (AAMAS
  2017)}}. \bibinfo{publisher}{IFAAMAS}, \bibinfo{address}{S\~{a}o Paulo,
  Brazil}, \bibinfo{pages}{464--473}.
\newblock
\urldef\tempurl%
\url{https://www.ifaamas.org/Proceedings/aamas2017/pdfs/p464.pdf}
\showURL{%
\tempurl}


\bibitem[Lor{\`e} and Heydari(2024)]%
        {lore2024strategic}
\bibfield{author}{\bibinfo{person}{Nunzio Lor{\`e}} {and}
  \bibinfo{person}{Babak Heydari}.} \bibinfo{year}{2024}\natexlab{}.
\newblock \showarticletitle{Strategic behavior of large language models and the
  role of game structure versus contextual framing}.
\newblock \bibinfo{journal}{\emph{Scientific Reports}}  \bibinfo{volume}{14}
  (\bibinfo{year}{2024}), \bibinfo{pages}{18490}.
\newblock
\href{https://doi.org/10.1038/s41598-024-69032-z}{doi:\nolinkurl{10.1038/s41598-024-69032-z}}


\bibitem[Mei et~al\mbox{.}(2024)]%
        {mei2024turing}
\bibfield{author}{\bibinfo{person}{Qiaozhu Mei}, \bibinfo{person}{Yutong Xie},
  \bibinfo{person}{Walter Yuan}, {and} \bibinfo{person}{Matthew~O. Jackson}.}
  \bibinfo{year}{2024}\natexlab{}.
\newblock \showarticletitle{A {Turing} test of whether {AI} chatbots are
  behaviorally similar to humans}.
\newblock \bibinfo{journal}{\emph{Proceedings of the National Academy of
  Sciences}} \bibinfo{volume}{121}, \bibinfo{number}{9} (\bibinfo{year}{2024}),
  \bibinfo{pages}{e2313925121}.
\newblock
\href{https://doi.org/10.1073/pnas.2313925121}{doi:\nolinkurl{10.1073/pnas.2313925121}}


\bibitem[Merhej et~al\mbox{.}(2021)]%
        {merhej2021cooperation}
\bibfield{author}{\bibinfo{person}{Ramona Merhej}, \bibinfo{person}{Fernando~P.
  Santos}, \bibinfo{person}{Francisco~S. Melo}, {and}
  \bibinfo{person}{Francisco~C. Santos}.} \bibinfo{year}{2021}\natexlab{}.
\newblock \showarticletitle{Cooperation between Independent Reinforcement
  Learners under Wealth Inequality and Collective Risks}. In
  \bibinfo{booktitle}{\emph{Proceedings of the 20th International Conference on
  Autonomous Agents and MultiAgent Systems}} (Virtual Event, United Kingdom)
  \emph{(\bibinfo{series}{AAMAS '21})}. \bibinfo{publisher}{International
  Foundation for Autonomous Agents and Multiagent Systems},
  \bibinfo{address}{Richland, SC}, \bibinfo{pages}{898--906}.
\newblock
\showISBNx{9781450383073}
\urldef\tempurl%
\url{https://www.ifaamas.org/Proceedings/aamas2021/pdfs/p898.pdf}
\showURL{%
\tempurl}


\bibitem[Milinski et~al\mbox{.}(2008)]%
        {milinski2008collective}
\bibfield{author}{\bibinfo{person}{Manfred Milinski}, \bibinfo{person}{Ralf~D.
  Sommerfeld}, \bibinfo{person}{Hans-J{\"u}rgen Krambeck},
  \bibinfo{person}{Floyd~A. Reed}, {and} \bibinfo{person}{Jochem Marotzke}.}
  \bibinfo{year}{2008}\natexlab{}.
\newblock \showarticletitle{The collective-risk social dilemma and the
  prevention of simulated dangerous climate change}.
\newblock \bibinfo{journal}{\emph{Proceedings of the National Academy of
  Sciences}} \bibinfo{volume}{105}, \bibinfo{number}{7} (\bibinfo{year}{2008}),
  \bibinfo{pages}{2291--2294}.
\newblock
\href{https://doi.org/10.1073/pnas.0709546105}{doi:\nolinkurl{10.1073/pnas.0709546105}}


\bibitem[Omidshafiei et~al\mbox{.}(2019)]%
        {omidshafiei2019alpha}
\bibfield{author}{\bibinfo{person}{Shayegan Omidshafiei},
  \bibinfo{person}{Christos Papadimitriou}, \bibinfo{person}{Georgios
  Piliouras}, \bibinfo{person}{Karl Tuyls}, \bibinfo{person}{Mark Rowland},
  \bibinfo{person}{Jean-Baptiste Lespiau}, \bibinfo{person}{Wojciech~M.
  Czarnecki}, \bibinfo{person}{Marc Lanctot}, \bibinfo{person}{Julien Perolat},
  {and} \bibinfo{person}{Remi Munos}.} \bibinfo{year}{2019}\natexlab{}.
\newblock \showarticletitle{{$\alpha$-Rank}: Multi-Agent Evaluation by
  Evolution}.
\newblock \bibinfo{journal}{\emph{Scientific Reports}} \bibinfo{volume}{9},
  \bibinfo{number}{1} (\bibinfo{year}{2019}), \bibinfo{pages}{9937}.
\newblock
\href{https://doi.org/10.1038/s41598-019-45619-9}{doi:\nolinkurl{10.1038/s41598-019-45619-9}}


\bibitem[Pal et~al\mbox{.}(2026)]%
        {pal2026large}
\bibfield{author}{\bibinfo{person}{Saptarshi Pal}, \bibinfo{person}{Abhishek
  Mallela}, \bibinfo{person}{Lenz Pracher}, \bibinfo{person}{Chiyu Wei},
  \bibinfo{person}{Feng Fu}, \bibinfo{person}{Santiago Schnell}, {and}
  \bibinfo{person}{Martin~A. Nowak}.} \bibinfo{year}{2026}\natexlab{}.
\newblock \showarticletitle{Large language models instantiate evolutionarily
  robust strategies of cooperation}.
\newblock \bibinfo{journal}{\emph{PNAS Nexus}} \bibinfo{volume}{5},
  \bibinfo{number}{6} (\bibinfo{year}{2026}), \bibinfo{pages}{pgag210}.
\newblock
\href{https://doi.org/10.1093/pnasnexus/pgag210}{doi:\nolinkurl{10.1093/pnasnexus/pgag210}}


\bibitem[Palfrey and Rosenthal(1984)]%
        {palfrey1984participation}
\bibfield{author}{\bibinfo{person}{Thomas~R. Palfrey} {and}
  \bibinfo{person}{Howard Rosenthal}.} \bibinfo{year}{1984}\natexlab{}.
\newblock \showarticletitle{Participation and the provision of discrete public
  goods: a strategic analysis}.
\newblock \bibinfo{journal}{\emph{Journal of Public Economics}}
  \bibinfo{volume}{24}, \bibinfo{number}{2} (\bibinfo{year}{1984}),
  \bibinfo{pages}{171--193}.
\newblock
\href{https://doi.org/10.1016/0047-2727(84)90023-9}{doi:\nolinkurl{10.1016/0047-2727(84)90023-9}}


\bibitem[Payne and Alloui-Cros(2025)]%
        {payne2025strategic}
\bibfield{author}{\bibinfo{person}{Kenneth Payne} {and}
  \bibinfo{person}{Baptiste Alloui-Cros}.} \bibinfo{year}{2025}\natexlab{}.
\newblock \bibinfo{title}{Strategic Intelligence in Large Language Models:
  Evidence from evolutionary Game Theory}.
\newblock
\showeprint[arxiv]{2507.02618}
\href{https://doi.org/10.48550/arXiv.2507.02618}{doi:\nolinkurl{10.48550/arXiv.2507.02618}}


\bibitem[Pham et~al\mbox{.}(2026)]%
        {pham2026humans}
\bibfield{author}{\bibinfo{person}{Phu~Hoa Pham}, \bibinfo{person}{Duy~Minh
  {Dao Sy}}, \bibinfo{person}{Trung~Kiet Huynh}, \bibinfo{person}{Phu~Quy
  {Nguyen Lam}}, \bibinfo{person}{Chi~Nguyen Tran}, \bibinfo{person}{Minh~Trung
  Le}, \bibinfo{person}{Phong~Hao Le}, \bibinfo{person}{Dinh~Nam Nguyen},
  \bibinfo{person}{Thien~Ky {Nguyen Dong}}, \bibinfo{person}{Elias {Fernandez
  Domingos}}, \bibinfo{person}{Hong~Trang Le}, {and} \bibinfo{person}{The~Anh
  Han}.} \bibinfo{year}{2026}\natexlab{}.
\newblock \bibinfo{title}{Humans Are More Diverse: Frontier {LLMs} Show Extreme
  Policies in Idealised {AI} Development Races}.
\newblock
\showeprint[arxiv]{2608.01193}
\href{https://doi.org/10.48550/arXiv.2608.01193}{doi:\nolinkurl{10.48550/arXiv.2608.01193}}


\bibitem[Piatti et~al\mbox{.}(2024)]%
        {piatti2024cooperate}
\bibfield{author}{\bibinfo{person}{Giorgio Piatti}, \bibinfo{person}{Zhijing
  Jin}, \bibinfo{person}{Max Kleiman-Weiner}, \bibinfo{person}{Bernhard
  Sch{\"o}lkopf}, \bibinfo{person}{Mrinmaya Sachan}, {and}
  \bibinfo{person}{Rada Mihalcea}.} \bibinfo{year}{2024}\natexlab{}.
\newblock \showarticletitle{Cooperate or Collapse: Emergence of Sustainable
  Cooperation in a Society of {LLM} Agents}. In
  \bibinfo{booktitle}{\emph{Advances in Neural Information Processing
  Systems}}, Vol.~\bibinfo{volume}{37}. \bibinfo{publisher}{Curran Associates,
  Inc.}, \bibinfo{address}{Vancouver, Canada}, \bibinfo{pages}{111715--111759}.
\newblock
\href{https://doi.org/10.52202/079017-3548}{doi:\nolinkurl{10.52202/079017-3548}}


\bibitem[Rahwan et~al\mbox{.}(2019)]%
        {rahwan2019machine}
\bibfield{author}{\bibinfo{person}{Iyad Rahwan}, \bibinfo{person}{Manuel
  Cebrian}, \bibinfo{person}{Nick Obradovich}, \bibinfo{person}{Josh Bongard},
  \bibinfo{person}{Jean-Fran{\c{c}}ois Bonnefon}, \bibinfo{person}{Cynthia
  Breazeal}, \bibinfo{person}{Jacob~W. Crandall}, \bibinfo{person}{Nicholas~A.
  Christakis}, \bibinfo{person}{Iain~D. Couzin}, \bibinfo{person}{Matthew~O.
  Jackson}, \bibinfo{person}{Nicholas~R. Jennings}, \bibinfo{person}{Ece
  Kamar}, \bibinfo{person}{Isabel~M. Kloumann}, \bibinfo{person}{Hugo
  Larochelle}, \bibinfo{person}{David Lazer}, \bibinfo{person}{Richard
  McElreath}, \bibinfo{person}{Alan Mislove}, \bibinfo{person}{David~C.
  Parkes}, \bibinfo{person}{Alex~`Sandy' Pentland},
  \bibinfo{person}{Margaret~E. Roberts}, \bibinfo{person}{Azim Shariff},
  \bibinfo{person}{Joshua~B. Tenenbaum}, {and} \bibinfo{person}{Michael
  Wellman}.} \bibinfo{year}{2019}\natexlab{}.
\newblock \showarticletitle{Machine behaviour}.
\newblock \bibinfo{journal}{\emph{Nature}} \bibinfo{volume}{568},
  \bibinfo{number}{7753} (\bibinfo{year}{2019}), \bibinfo{pages}{477--486}.
\newblock
\href{https://doi.org/10.1038/s41586-019-1138-y}{doi:\nolinkurl{10.1038/s41586-019-1138-y}}


\bibitem[Rand et~al\mbox{.}(2013)]%
        {rand2013evolution}
\bibfield{author}{\bibinfo{person}{David~G. Rand}, \bibinfo{person}{Corina~E.
  Tarnita}, \bibinfo{person}{Hisashi Ohtsuki}, {and} \bibinfo{person}{Martin~A.
  Nowak}.} \bibinfo{year}{2013}\natexlab{}.
\newblock \showarticletitle{Evolution of fairness in the one-shot anonymous
  {Ultimatum Game}}.
\newblock \bibinfo{journal}{\emph{Proceedings of the National Academy of
  Sciences}} \bibinfo{volume}{110}, \bibinfo{number}{7} (\bibinfo{year}{2013}),
  \bibinfo{pages}{2581--2586}.
\newblock
\href{https://doi.org/10.1073/pnas.1214167110}{doi:\nolinkurl{10.1073/pnas.1214167110}}


\bibitem[Santos and Pacheco(2011)]%
        {santos2011risk}
\bibfield{author}{\bibinfo{person}{Francisco~C. Santos} {and}
  \bibinfo{person}{Jorge~M. Pacheco}.} \bibinfo{year}{2011}\natexlab{}.
\newblock \showarticletitle{Risk of collective failure provides an escape from
  the tragedy of the commons}.
\newblock \bibinfo{journal}{\emph{Proceedings of the National Academy of
  Sciences}} \bibinfo{volume}{108}, \bibinfo{number}{26}
  (\bibinfo{year}{2011}), \bibinfo{pages}{10421--10425}.
\newblock
\href{https://doi.org/10.1073/pnas.1015648108}{doi:\nolinkurl{10.1073/pnas.1015648108}}


\bibitem[Santos et~al\mbox{.}(2020)]%
        {santos2020picky}
\bibfield{author}{\bibinfo{person}{Fernando~P. Santos}, \bibinfo{person}{Samuel
  Mascarenhas}, \bibinfo{person}{Francisco~C. Santos}, \bibinfo{person}{Filipa
  Correia}, \bibinfo{person}{Samuel Gomes}, {and} \bibinfo{person}{Ana Paiva}.}
  \bibinfo{year}{2020}\natexlab{}.
\newblock \showarticletitle{Picky losers and carefree winners prevail in
  collective risk dilemmas with partner selection}.
\newblock \bibinfo{journal}{\emph{Autonomous Agents and Multi-Agent Systems}}
  \bibinfo{volume}{34}, \bibinfo{number}{2} (\bibinfo{year}{2020}),
  \bibinfo{pages}{40}.
\newblock
\href{https://doi.org/10.1007/s10458-020-09463-w}{doi:\nolinkurl{10.1007/s10458-020-09463-w}}


\bibitem[Santos et~al\mbox{.}(2019)]%
        {santos2019outcome}
\bibfield{author}{\bibinfo{person}{Fernando~P. Santos},
  \bibinfo{person}{Samuel~F. Mascarenhas}, \bibinfo{person}{Francisco~C.
  Santos}, \bibinfo{person}{Filipa Correia}, \bibinfo{person}{Samuel Gomes},
  {and} \bibinfo{person}{Ana Paiva}.} \bibinfo{year}{2019}\natexlab{}.
\newblock \showarticletitle{Outcome-based Partner Selection in Collective Risk
  Dilemmas}. In \bibinfo{booktitle}{\emph{Proceedings of the 18th International
  Conference on Autonomous Agents and MultiAgent Systems}} (Montreal, QC,
  Canada) \emph{(\bibinfo{series}{AAMAS '19})}.
  \bibinfo{publisher}{International Foundation for Autonomous Agents and
  Multiagent Systems}, \bibinfo{address}{Richland, SC},
  \bibinfo{pages}{1556--1564}.
\newblock
\showISBNx{9781450363099}
\urldef\tempurl%
\url{https://www.ifaamas.org/Proceedings/aamas2019/pdfs/p1556.pdf}
\showURL{%
\tempurl}


\bibitem[Schaffer(1988)]%
        {schaffer1988evolutionarily}
\bibfield{author}{\bibinfo{person}{Mark~E. Schaffer}.}
  \bibinfo{year}{1988}\natexlab{}.
\newblock \showarticletitle{Evolutionarily stable strategies for a finite
  population and a variable contest size}.
\newblock \bibinfo{journal}{\emph{Journal of Theoretical Biology}}
  \bibinfo{volume}{132}, \bibinfo{number}{4} (\bibinfo{year}{1988}),
  \bibinfo{pages}{469--478}.
\newblock
\href{https://doi.org/10.1016/S0022-5193(88)80085-7}{doi:\nolinkurl{10.1016/S0022-5193(88)80085-7}}


\bibitem[Schmied et~al\mbox{.}(2026)]%
        {schmied2026llms}
\bibfield{author}{\bibinfo{person}{Thomas Schmied}, \bibinfo{person}{J{\"o}rg
  Bornschein}, \bibinfo{person}{Jordi Grau-Moya}, \bibinfo{person}{Markus
  Wulfmeier}, {and} \bibinfo{person}{Razvan Pascanu}.}
  \bibinfo{year}{2026}\natexlab{}.
\newblock \showarticletitle{{LLMs} are Greedy Agents: Effects of {RL}
  Fine-tuning on Decision-Making Abilities}. In
  \bibinfo{booktitle}{\emph{International Conference on Learning
  Representations (ICLR)}}. \bibinfo{publisher}{OpenReview.net},
  \bibinfo{address}{Rio de Janeiro, Brazil}, \bibinfo{pages}{127538--127571}.
\newblock
\urldef\tempurl%
\url{https://proceedings.iclr.cc/paper_files/paper/2026/hash/cf2fcda40de8520027af2eadcafdecd6-Abstract-Conference.html}
\showURL{%
\tempurl}


\bibitem[Sigmund et~al\mbox{.}(2010)]%
        {sigmund2010social}
\bibfield{author}{\bibinfo{person}{Karl Sigmund}, \bibinfo{person}{Hannelore
  De~Silva}, \bibinfo{person}{Arne Traulsen}, {and} \bibinfo{person}{Christoph
  Hauert}.} \bibinfo{year}{2010}\natexlab{}.
\newblock \showarticletitle{Social learning promotes institutions for governing
  the commons}.
\newblock \bibinfo{journal}{\emph{Nature}} \bibinfo{volume}{466},
  \bibinfo{number}{7308} (\bibinfo{year}{2010}), \bibinfo{pages}{861--863}.
\newblock
\href{https://doi.org/10.1038/nature09203}{doi:\nolinkurl{10.1038/nature09203}}


\bibitem[Slumbers et~al\mbox{.}(2025)]%
        {slumbers2025using}
\bibfield{author}{\bibinfo{person}{Oliver Slumbers}, \bibinfo{person}{Joel~Z.
  Leibo}, {and} \bibinfo{person}{Marco~A. Janssen}.}
  \bibinfo{year}{2025}\natexlab{}.
\newblock \bibinfo{title}{Using Large Language Models to Simulate Human
  Behavioural Experiments: Port of Mars}.
\newblock
\showeprint[arxiv]{2506.05555}
\href{https://doi.org/10.48550/arXiv.2506.05555}{doi:\nolinkurl{10.48550/arXiv.2506.05555}}


\bibitem[Stone et~al\mbox{.}(2010)]%
        {stone2010ad}
\bibfield{author}{\bibinfo{person}{Peter Stone}, \bibinfo{person}{Gal~A.
  Kaminka}, \bibinfo{person}{Sarit Kraus}, {and} \bibinfo{person}{Jeffrey~S.
  Rosenschein}.} \bibinfo{year}{2010}\natexlab{}.
\newblock \showarticletitle{Ad Hoc Autonomous Agent Teams: Collaboration
  without Pre-Coordination}. In \bibinfo{booktitle}{\emph{Proceedings of the
  AAAI Conference on Artificial Intelligence}}, Vol.~\bibinfo{volume}{24}.
  \bibinfo{publisher}{AAAI Press}, \bibinfo{address}{Atlanta, Georgia, USA},
  \bibinfo{pages}{1504--1509}.
\newblock
\href{https://doi.org/10.1609/aaai.v24i1.7529}{doi:\nolinkurl{10.1609/aaai.v24i1.7529}}


\bibitem[Tavoni et~al\mbox{.}(2011)]%
        {tavoni2011inequality}
\bibfield{author}{\bibinfo{person}{Alessandro Tavoni}, \bibinfo{person}{Astrid
  Dannenberg}, \bibinfo{person}{Giorgos Kallis}, {and} \bibinfo{person}{Andreas
  L{\"o}schel}.} \bibinfo{year}{2011}\natexlab{}.
\newblock \showarticletitle{Inequality, communication, and the avoidance of
  disastrous climate change in a public goods game}.
\newblock \bibinfo{journal}{\emph{Proceedings of the National Academy of
  Sciences}} \bibinfo{volume}{108}, \bibinfo{number}{29}
  (\bibinfo{year}{2011}), \bibinfo{pages}{11825--11829}.
\newblock
\href{https://doi.org/10.1073/pnas.1102493108}{doi:\nolinkurl{10.1073/pnas.1102493108}}


\bibitem[Terrucha et~al\mbox{.}(2024)]%
        {terrucha2024art}
\bibfield{author}{\bibinfo{person}{In{\^e}s Terrucha}, \bibinfo{person}{Elias
  Fern{\'a}ndez~Domingos}, \bibinfo{person}{Francisco~C. Santos},
  \bibinfo{person}{Pieter Simoens}, {and} \bibinfo{person}{Tom Lenaerts}.}
  \bibinfo{year}{2024}\natexlab{}.
\newblock \showarticletitle{The art of compensation: How hybrid teams solve
  collective-risk dilemmas}.
\newblock \bibinfo{journal}{\emph{PLOS ONE}} \bibinfo{volume}{19},
  \bibinfo{number}{2} (\bibinfo{year}{2024}), \bibinfo{pages}{e0297213}.
\newblock
\href{https://doi.org/10.1371/journal.pone.0297213}{doi:\nolinkurl{10.1371/journal.pone.0297213}}


\bibitem[Tewolde et~al\mbox{.}(2026)]%
        {tewolde2026coopeval}
\bibfield{author}{\bibinfo{person}{Emanuel Tewolde}, \bibinfo{person}{Xiao
  Zhang}, \bibinfo{person}{David Guzman~Piedrahita}, \bibinfo{person}{Vincent
  Conitzer}, {and} \bibinfo{person}{Zhijing Jin}.}
  \bibinfo{year}{2026}\natexlab{}.
\newblock \showarticletitle{{CoopEval}: Benchmarking Cooperation-Sustaining
  Mechanisms and {LLM} Agents in Social Dilemmas}. In
  \bibinfo{booktitle}{\emph{Proceedings of the 43rd International Conference on
  Machine Learning (ICML)}}. \bibinfo{publisher}{PMLR},
  \bibinfo{address}{Seoul, South Korea}, \bibinfo{numpages}{65}~pages.
\newblock
\urldef\tempurl%
\url{https://icml.cc/virtual/2026/poster/66503}
\showURL{%
\tempurl}


\bibitem[Traulsen and Nowak(2006)]%
        {traulsen2006evolution}
\bibfield{author}{\bibinfo{person}{Arne Traulsen} {and}
  \bibinfo{person}{Martin~A. Nowak}.} \bibinfo{year}{2006}\natexlab{}.
\newblock \showarticletitle{Evolution of cooperation by multilevel selection}.
\newblock \bibinfo{journal}{\emph{Proceedings of the National Academy of
  Sciences}} \bibinfo{volume}{103}, \bibinfo{number}{29}
  (\bibinfo{year}{2006}), \bibinfo{pages}{10952--10955}.
\newblock
\href{https://doi.org/10.1073/pnas.0602530103}{doi:\nolinkurl{10.1073/pnas.0602530103}}


\bibitem[Traulsen et~al\mbox{.}(2006)]%
        {traulsen2006stochastic}
\bibfield{author}{\bibinfo{person}{Arne Traulsen}, \bibinfo{person}{Martin~A.
  Nowak}, {and} \bibinfo{person}{Jorge~M. Pacheco}.}
  \bibinfo{year}{2006}\natexlab{}.
\newblock \showarticletitle{Stochastic dynamics of invasion and fixation}.
\newblock \bibinfo{journal}{\emph{Physical Review E}} \bibinfo{volume}{74},
  \bibinfo{number}{1} (\bibinfo{year}{2006}), \bibinfo{pages}{011909}.
\newblock
\href{https://doi.org/10.1103/PhysRevE.74.011909}{doi:\nolinkurl{10.1103/PhysRevE.74.011909}}


\bibitem[Tuyls et~al\mbox{.}(2018)]%
        {tuyls2018generalised}
\bibfield{author}{\bibinfo{person}{Karl Tuyls}, \bibinfo{person}{Julien
  Perolat}, \bibinfo{person}{Marc Lanctot}, \bibinfo{person}{Joel~Z. Leibo},
  {and} \bibinfo{person}{Thore Graepel}.} \bibinfo{year}{2018}\natexlab{}.
\newblock \showarticletitle{A Generalised Method for Empirical Game Theoretic
  Analysis}. In \bibinfo{booktitle}{\emph{Proceedings of the 17th International
  Conference on Autonomous Agents and Multiagent Systems (AAMAS 2018)}}.
  \bibinfo{publisher}{IFAAMAS}, \bibinfo{address}{Stockholm, Sweden},
  \bibinfo{pages}{77--85}.
\newblock
\urldef\tempurl%
\url{https://ifaamas.org/Proceedings/aamas2018/pdfs/p77.pdf}
\showURL{%
\tempurl}


\bibitem[Vallinder and Hughes(2025)]%
        {vallinder2025cultural}
\bibfield{author}{\bibinfo{person}{Aron Vallinder} {and}
  \bibinfo{person}{Edward Hughes}.} \bibinfo{year}{2025}\natexlab{}.
\newblock \showarticletitle{Cultural Evolution of Cooperation among {LLM}
  Agents}. In \bibinfo{booktitle}{\emph{Proceedings of the 24th International
  Conference on Autonomous Agents and Multiagent Systems (AAMAS 2025)}}.
  \bibinfo{publisher}{IFAAMAS}, \bibinfo{address}{Detroit, MI, USA},
  \bibinfo{pages}{2771--2773}.
\newblock
\href{https://doi.org/10.65109/jnmb7739}{doi:\nolinkurl{10.65109/jnmb7739}}


\bibitem[Wellman(2006)]%
        {wellman2006methods}
\bibfield{author}{\bibinfo{person}{Michael~P. Wellman}.}
  \bibinfo{year}{2006}\natexlab{}.
\newblock \showarticletitle{Methods for Empirical Game-Theoretic Analysis
  (Extended Abstract)}. In \bibinfo{booktitle}{\emph{Proceedings of the
  Twenty-First National Conference on Artificial Intelligence (AAAI-06)}}.
  \bibinfo{publisher}{AAAI Press}, \bibinfo{address}{Boston, Massachusetts,
  USA}, \bibinfo{pages}{1552--1555}.
\newblock
\urldef\tempurl%
\url{https://aaai.org/Papers/AAAI/2006/AAAI06-248.pdf}
\showURL{%
\tempurl}


\bibitem[Willis et~al\mbox{.}(2026)]%
        {willis2026evaluating}
\bibfield{author}{\bibinfo{person}{Richard Willis}, \bibinfo{person}{Jianing
  Zhao}, \bibinfo{person}{Yali Du}, {and} \bibinfo{person}{Joel~Z. Leibo}.}
  \bibinfo{year}{2026}\natexlab{}.
\newblock \bibinfo{title}{Evaluating Collective Behaviour of Hundreds of {LLM}
  Agents}.
\newblock
\shownote{Strategic Engineering Workshop on LLMs and Game Theory (SE@AAMAS
  2026)}.
\newblock
\showeprint[arxiv]{2602.16662}
\href{https://doi.org/10.48550/arXiv.2602.16662}{doi:\nolinkurl{10.48550/arXiv.2602.16662}}


\bibitem[Zisis et~al\mbox{.}(2015)]%
        {zisis2015generosity}
\bibfield{author}{\bibinfo{person}{Ioannis Zisis}, \bibinfo{person}{Sibilla
  Di~Guida}, \bibinfo{person}{The~Anh Han}, \bibinfo{person}{Georg
  Kirchsteiger}, {and} \bibinfo{person}{Tom Lenaerts}.}
  \bibinfo{year}{2015}\natexlab{}.
\newblock \showarticletitle{Generosity motivated by acceptance - evolutionary
  analysis of an anticipation game}.
\newblock \bibinfo{journal}{\emph{Scientific Reports}} \bibinfo{volume}{5},
  \bibinfo{number}{1} (\bibinfo{year}{2015}), \bibinfo{pages}{18076}.
\newblock
\href{https://doi.org/10.1038/srep18076}{doi:\nolinkurl{10.1038/srep18076}}


\end{thebibliography}
